\documentclass[12pt]{amsart}
\usepackage{geometry} 
\usepackage{graphicx}
\usepackage{amsfonts}
\usepackage{amsthm}
\usepackage{amsmath}
\usepackage{amssymb}
\usepackage{mathrsfs}
\usepackage[arrow,matrix,curve]{xy}
\usepackage{enumerate}
\usepackage{color}
\usepackage[normalem]{ulem}
\usepackage{tikz}
\usepackage{tikz}
\usetikzlibrary{matrix}
\usepackage{float}
\usepackage{xspace}
\usepackage{subcaption}
\definecolor{light-blue}{rgb}{0.8,0.85,1}
\definecolor{light-red}{rgb}{1,.4,.4}
\definecolor{purp}{rgb}{.7,.3,1}
\definecolor{yel}{rgb}{1,1,.5}
\definecolor{cy}{rgb}{0,1,1}
\usepackage{colortbl}
\usepackage{multirow}
\usepackage{array}
\usepackage{makecell}
\newtheorem{theorem}{Theorem}[section]
\newtheorem{corollary}[theorem]{Corollary}
\newtheorem{physics}[theorem]{Physical Assertion}
\newtheorem{lemma}[theorem]{Lemma}
\newtheorem{proposition}[theorem]{Proposition}

\theoremstyle{definition}

\newtheorem{remark}[theorem]{Remark}
\numberwithin{figure}{section}
\numberwithin{equation}{section}
\newcommand{\co}{\colon\,}

\newcommand{\bT}{\mathbb T}

\newcommand{\bR}{\mathbb R}
\newcommand{\bC}{\mathbb C}
\newcommand{\bZ}{\mathbb Z}
\newcommand{\bP}{\mathbb P}

\newcommand{\SO}{\mathop{\rm SO}}

\newcommand{\NS}{\operatorname{NS}}

\newcommand{\Kum}{\operatorname{Kum}}

\renewcommand\Re{\operatorname{Re}}

\newcommand{\beq}{\begin{small} \begin{equation}}
\newcommand{\eeq}{\end{equation} \end{small}}
\newcommand{\beqn}{\begin{small} \begin{equation*}}
\newcommand{\eeqn}{\end{equation*} \end{small}}
\newcommand{\sign}{\operatorname{sign}}

\newcommand\scalemath[2]{\scalebox{#1}{\mbox{\ensuremath{\displaystyle #2}}}}

\title[K3 orientifolds]{Orientifolds for F-theory on K3 surfaces}

\author[C. Doran]{Charles Doran}
\address{Department of Mathematical and Statistical Sciences\\
University of Alberta\\
Edmonton, AB T6G 2G1, Canada\\
and Mathematics Program\\
  Bard College\\
  P.O. Box 5000\\
  Annandale, NY 12504-5000, USA}
\email[Charles Doran]{cdoran@bard.edu}
\author[A. Malmendier]{Andreas Malmendier}
\address{Mathematics and Statistics Department\\
  Utah State University\\
  Logan, UT 84322-3900, USA}
\email[Andreas Malmendier]{andreas.malmendier@usu.edu}
\thanks{AM partially supported by Simons Collaboration Grant 635124.}
\author[S. M\'endez-Diez]{Stefan M\'endez-Diez}
\address{Mathematics Program\\
  Bard College\\
  P.O. Box 5000\\
  Annandale, NY 12504-5000, USA}
\email[Stefan M\'endez-Diez]{smendezdiez@bard.edu }
\author[J. Rosenberg]{Jonathan Rosenberg}
\address{Department of Mathematics\\
University of Maryland\\
College Park, MD 20742-4015, USA} 
\email[Jonathan Rosenberg]{jmr@umd.edu}
  
\begin{document}

\begin{abstract}
We study F-theory orientifolds, starting with products of
two elliptic curves, but focusing mostly on a family of K3 surfaces, lattice polarized by the rank-17 lattice $\langle 8\rangle \oplus 2 D_8(-1)$, generalizing the family (to which it degenerates)
of Kummer surfaces of products of two non-isogenous elliptic curves.
After a thorough study of the complex geometry of this family
and its elliptic fibrations, we proceed to study real structures
on the K3 surfaces in the family which are equivariant with
respect to an elliptic fibration. We also study the physics
of the associated F-theory orientifolds with a particular focus on the impact of the real structure on the charge spectrum.  We also study how
these orientifolds degenerate to the case of isotrivial
Kummer surface fibrations.
\end{abstract}
\keywords{orientifold, $T$-duality,
  K3 surface, real structure, algebraic $K$-theory, $KR$-theory}
\subjclass[2010]{81T30 14J28 14J33 19L50 19E08}

\maketitle

\section{Introduction}

Compactification of F-theory on an elliptically
fibered K3 surface provides a framework to encode type IIB string theory 
on elliptic curves, with the K\"{a}hler modulus of the elliptic curve encoded 
in the complex structure of the elliptic fibers. This paper extends that perspective 
by examining F-theory orientifolds on elliptically fibered K3 surfaces and connecting 
them to $D$-brane classifications using real $K$-theory ($KR$-theory).

In Section 2, we review the role of $KR$-theory in describing D-brane charges in type IIA and type IIB orientifold theories.  This is summarized in three “Physical Assertions” (2.1, 2.2, and 2.3): Firstly, the $D$-brane charges are classified by appropriately twisted Real $K$-theory in the sense of Atiyah.  Secondly, Bott periodicity of the $KR$-classification of $D$-branes corresponds to the physical degrees of freedom in choosing the relative sign of left and right movers.  Thirdly, in the presence of $O$-planes, no coordinate choice is needed to determine the index of $KR^{-i}$ that classifies the $D$-brane charges.

Next, in Section 3, we begin connecting this with geometry by analyzing involutions on products of elliptic curves, $E_1\times E_2$, compatible with projections onto the two factors.  These involutions are essential for distinguishing between type IIA and type IIB orientifold theories, which differ by whether their Ramond sector ground states have opposite or identical chiralities. This distinction plays a central role in classifying $D$-brane spectra and analyzing the corresponding fixed-point loci in orientifold theories.  We focus on the case of trivial $B$-field in all $T$-dual theories to illustrate the $D$-brane charge classification for both type I and type $\tilde{\rm I}$ theories and their $T$-duals.

Our analysis then moves, in Section 4, from products of elliptic curves to their associated Kummer surfaces, constructed by quotienting $E_1\times E_2$ by an involution and resolving the resulting singularities. Kummer surfaces serve as a bridge between the product space $E_1\times E_2$ and more general elliptically fibered K3 surfaces.  An important role is played by the isotrivial fibrations --- fibrations where all elliptic fibers are isomorphic.

Motivated by the physical assertion (4.1) that says that the type I string on $T^4$ is $T$-dual to the type IIB orientifold compactification on the quotient $T^4/{\mathbb Z}_2$, and viewing this, post-resolution, as a Kummer surface of a product of non-isogenous elliptic curves $E_1\times E_2$, we study the duality of involutions on Kum($E_1\times E_2$).  This includes a precise description of the contributions to $KO^\bullet$ and $KR^\bullet$ due to singular Kodaira-type $I_0^*$ fibers in the isotrivial fibrations on Kum($E_1\times E_2$).

The geometric heart of the paper is Section 5, where we introduce a natural extension of the family of K3 surfaces Kum($E_1\times E_2$), replacing the product of elliptic curves by abelian surfaces with a $(2,4)$-polarization.  Our analysis of the geometry and involution structure is complete in all details.  In Section 5.1.1, we establish a normal form with monomial deformation coefficients for this family of generically Picard rank 17 K3 surfaces $X$ with N\'{e}ron-Severi group $\langle 8\rangle \oplus 2 D_8(-1)$, singular fibration strata, and special properties of a distinguished elliptic fibration.  In Sections 5.2-5.4, we make precise the connection between moduli of $X$ and abelian surfaces with a (2,4)-polarization, including both the construction via even eights and a modular parametrization (in the sense of Clingher-Doran \cite{MR2369941}) via genus 2 theta constants that expresses the monomial deformation coefficients in terms of modular forms.  In Section 5.5, we determine the action of the three commuting anti-symplectic involutions on classes of curves in $NS(X)$ and the corresponding invariant sublattices.

Finally, in Section 6, we use the real structures — antiholomorphic involutions — on our K3 surfaces to connect the geometry with the physics, providing a natural setting for understanding the interplay between elliptic fibration structures and D-brane classifications in F-theory. We construct Real normal forms with their associated antiholomorphic involutions in Section 6.1.  In Section 6.2, we use this to make explicit the 2-torsion Brauer twist that relates our normal forms to the Jacobian (Weierstrass normal form) elliptic fibration, including the realization of a representative for the twisting class as an Azumaya algebra.  This all connects back to the physics by considering three families of real K3 surfaces whose string limits give the three different type IIB theories on ${\mathbb P}^1$ with four $I_0^*$ fibers. 

We conclude the paper with a number of open questions and directions for further research.

\section{Type IIA vs IIB Orientifolds}
\label{sec:IIAB}

In this section we will look at some of the defining features of type IIA and IIB orientifolds. In \cite{MR3267662,MR3316647} the two different orientifold theories were distinguished by whether the involution on the compactification space was holomorphic or antiholomorphic. It turns out that this was just a coincidence of the fact that the compactification manifold had complex dimension $1$. We will explore in more detail the differences between the two orientifold theories and see that this has a large impact on classifying their $D$-brane spectra. 

The distinguishing feature between type IIB and IIA theories is that
the left-moving and right-moving ground states in the Ramond sector of
IIB theories have the same chirality. For IIA theories, on the other
hand, the chiralities of the left-moving and right-moving ground states in
the Ramond sector should be opposite. If we start with a IIB theory we can
write the fermionic coordinates in the Ramond sector as
\begin{equation}
	\psi^\mu=\psi_L^\mu+\psi_R^\mu,
\end{equation}
where
\begin{equation}
	\psi_R^\mu=\sum_{n\in\bZ}d_n^\mu e^{-ni(\tau-\sigma)}\text{ and }\psi_L^\mu=\sum_{n\in\bZ}\tilde{d}_n^\mu e^{-ni(\tau+\sigma)},
\end{equation}
where $\tau$ and $\sigma$ are coordinates on the string
worldsheet $\Sigma$.
The Fourier coefficients satisfy
\[
\{d_m^\mu,d_n^\nu\}=\eta^{\mu\nu}\delta_{m+n,0},
\]
and similarly for the left-moving coefficients.
When $m=n=0$ this is just the Dirac algebra, up to a factor of $2$, and we see that
$$\Gamma^\mu=\sqrt{2}d_0^\mu$$
where $\Gamma^\mu$ are the $10$-dimensional Dirac matrices. The chirality operator is
\begin{equation}
	\Gamma_{11}=\Gamma_0\Gamma_1\cdots\Gamma_9.
	\label{eq:chiralop}
\end{equation}
If we call the left-moving ground state in the Ramond sector $|0\rangle_{L,R}$ and the right-moving Ramond ground state $|0\rangle_{R,R}$, then since this a IIB theory we can assume
\begin{equation}
	\Gamma_{11}|0\rangle_{R,R}=|0\rangle_{R,R}\text{  and  }\tilde{\Gamma}_{11}|0\rangle_{L,R}=|0\rangle_{L,R}.
\end{equation}
Since only the relative chirality matters, we can assume they both have chirality $1$. Here we have described the left- and right-moving ground states as Majorana-Weyl spinors, as is always possible in $D=10$ dimensions.

Now we can define a different type II string theory by changing $\psi_R^k\mapsto -\psi_R^k$ for some fixed $k\in\{0,1,\ldots,9\}$. Supersymmetry then requires that the bosonic coordinate transforms as $x_R^k\mapsto -x_R^k$ where $x^\mu=x_L^\mu+x_R^\mu$ are the bosonic coordinates in the original type IIB theory. Therefore the new type II string theory has coordinates
\begin{align}
	\psi^\mu &=\psi_L^\mu+\psi_R^\mu,  \;\;\mu\neq k,\label{eq:symF}\\
	\psi^k &=\psi_L^k-\psi_R^k, \label{eq:antisymF}\\
	x^\mu &=x_L^\mu+x_R^\mu, \;\;\mu\neq k,\label{eq:sym}\\
	x^k &=x_L^k-x_R^k.\label{eq:antisym} 
\end{align}
Note that changing $\psi_R^k\mapsto -\psi_R^k$ transforms $d_0^k\mapsto -d_0^k$ while leaving $d_0^\mu$ unchanged for all $\mu\neq k$. Therefore, $\Gamma_{11}\mapsto -\Gamma_{11}$, while the left-moving parity operator remains unchanged. This reverses the parity of the right-moving Ramond ground state and leaves the parity of the left-moving ground state alone. This means that this type II string theory is a type IIA theory. 

If we perform the transformations $x_R^\mu\mapsto -x_R^\mu$ and $\psi_R^\mu\mapsto -\psi_R^\mu$ for $\mu=k_1$ and $k_2$ in two distinct directions $k_1,k_2$, instead of a single direction then both $d_0^{k_1}$ and $d_0^{k_2}$ change sign. This leaves the parity operator $\Gamma_{11}$ invariant. Therefore the left and right moving Ramond ground states both have chirality $1$ and this must be a type IIB theory. More generally, if we send $\psi_R\mapsto -\psi_R$ for an odd number of dimensions then it defines a type IIA theory and if we do it for an even number of dimensions it defines a type IIB theory. Said another way, for a type IIA theory an odd number, $i$, of the bosonic coordinates must split as $x=x_L-x_R$ while the remaining $10-i$ (also odd) dimensions must split as $x_L+x_R$. For a type IIB theory an even number, $i$, of the bosonic coordinates must split as $x=x_L-x_R$ while the remaining $10-i$ (also even) dimensions must split as $x_L+x_R$.

As we will see now, from the perspective of the string sigma model this corresponds to the well known fact that there are only odd dimensional $D$-branes in the type IIB theory and even dimensional $D$-branes in the type IIA theory. It will however have more importance when we consider type II orientifolds. Let $X^\mu(\tau,\sigma)$ be the embedding of the string worldsheet into the spacetime manifold for an open string. If the string endpoints satisfy Neumann boundary conditions in direction $k$ then
\begin{equation}
\label{eq:NBC}
	X^k(\tau,\sigma)=X_L^k(\tau,\sigma)+X_R^k(\tau,\sigma),
\end{equation} 
but if they satisfy Dirichlet boundary conditions in direction $k$ then
\begin{equation}
\label{eq:DBC}
	X^k(\tau,\sigma)=X_L^k(\tau,\sigma)-X_R^k(\tau,\sigma).
\end{equation} 
If a string endpoint satisfies Dirichlet boundary conditions in $i$ directions then it is fixed in those $i$ directions and can move freely in the other $10-i$ directions that satisfy Neumann conditions. This corresponds to the string endpoint being attached to a $D(9-i)$-brane. This shows that the fact that $i$ must be even for IIB theories and odd for IIA theories is equivalent to the well known fact that the IIB theory contains odd dimensional $D$-branes while the IIA theory contains even dimensional $D$-branes. The story however becomes more interesting when considering orientifolds.

For an orientifold, in addition to the spacetime manifold $X$, we must also consider an involution $\iota\cdot\Omega$ where $\Omega$ is the worldsheet parity operator and $\iota$ is an involution of $X$. Now the worldsheet embeddings, $X^\mu\co \Sigma\to X$, are required to be equivariant so that $\iota\circ X^\mu=X^\mu\circ\Omega$. Here $\Sigma$ is the string worldsheet (an oriented $2$-manifold). The worldsheet parity operator, $\Omega$, interchanges left and right movers. Therefore, if $X^\mu$ satisfies equation \ref{eq:NBC} then $\Omega$ sends $X^\mu$ to itself, but if $X^\mu$ satisfies equation \ref{eq:DBC} then $\Omega$ sends it to $-X^\mu$. This means that $\iota$ must reflect every direction that satisfies equation \ref{eq:DBC}. When we combine this with the fact that type IIA theories must have an odd number of dimensions satisfying equation \ref{eq:DBC} and type IIB theories must have an even number of dimensions satisfying equation \ref{eq:DBC}, we see that type IIA orientifolds must have a spacetime involution that reflects an odd number of dimensions, while type IIB orientifolds require a spacetime involution that reflect an even number of dimensions.


\subsection{Classifying $\mathbf{D}$-branes}
\label{sec:DbraneClass}

$D$-branes are classified by $K$-theory \cite{MR1606278,Witten:1998}.
For the ordinary type IIA and IIB theories on a spacetime $X$, stable $D$-branes are classified by complex $K$-theory
$K^{-i}(X)$. For an orientifold theory $(X,\iota)$,
one needs to use Real $K$-theory in the sense of Atiyah
\cite{MR0206940}, $KR^{-i}(X,\iota)$, which takes
the spacetime involution $\iota$ on $X$ into account
(\cite[\S5.2]{Witten:1998} and \cite{MR3267662}). 
$D$-branes in orientifolds of the form $X/(\iota\cdot\Omega\cdot(-1)^{F_L})$ were shown in \cite[\S5.2]{Witten:1998} to be classified by $KR_\pm(X)$. 
We will use the convention that $\bR^{p,q}$ denotes
$\bR^p\oplus\bR^q$ with an involution that is $+1$ on the
first summand and $-1$ on the second summand.  Then
$S^{p,q}$ denotes the unit sphere in $\bR^{p,q}$, so that
for example $S^{2,0}$ denotes $S^1$ with trivial involution.
We have
\begin{align*}
	KR_\pm(X)&\cong KR(X\times\bR^{2,0})\\
	&\cong KR^{-2}(X).
\end{align*}
What's important is the shift of index by $2$ which is the same as the number of the dimensions reflected, or from the perspective of starting with coordinates compatible with the type I theory equivalent to the number of dimensions that need to be operated on by $(-1)^{F_L}$ to be compatible with $\iota\cdot\Omega$. This pattern can be seen with all of the cases discussed in \cite{MR3267662,MR3316647}. Putting this together with the above discussion we are led to the following.
\begin{physics}
\label{phys:index}
	$D$-brane charges are classified by $KR^{-i}$, appropriately twisted, where $i$ is given by the number of coordinates satisfying equation \ref{eq:antisym}, or equivalently for orientifolds, the number of dimensions that are reflected by the spacetime involution.
\end{physics}
For the ordinary type IIB and IIA theories this is just the statement that $D$-branes are classified by $K^{-i}$ where $i$ is even for the type IIB theory and odd for the IIA theory, but becomes much more interesting when looking at orientifolds. With orientifolds we still see that type IIB orientifolds will be classified by $KR^{-i}$ for $i$ even while type IIA orientifolds will have $i$ odd since IIA orientifolds must have an odd number of reflections to account for the left and right movers having opposite chirality. As is often the case when considering the full orientifold description and $KR$ theory, one of the first distinctions that stands out is the importance of the $8$-periodic Bott periodicity of the $KR$-theory.
\begin{physics}
	Bott periodicity of the $KR$-theory classification of $D$-branes corresponds to the $8$ physical degrees of freedom in choosing which coordinates satisfy equation \ref{eq:antisym}.
\end{physics}
This follows immediately from physical assertion \ref{phys:index} when combined with the fact that there are only $8$ physical degrees of freedom in which to choose the relative sign of the left and right movers. This latter point can most easily be seen in the light cone gauge where the freedom is in the $8$ transverse directions. 

As explained in \cite{MR3316647}, all of the twisting data can be completely determined by the spacetime and the involution. When we combine this with physical assertion \ref{phys:index} we now see that the full $KR$-theory classification can be completely determined by the spacetime and involution alone and does not require any additional input about the type of string theory being considered. There is something still unsatisfying about needing a specific coordinate choice to determine the index $i$, especially when other compactification manifolds might not split as nicely. Luckily, in the presence of $O$-planes,the
index can also be determined by their dimension.
\begin{physics}
\label{phys:Opindex}
	If an orientifold theory contains $O$-planes of dimension $p$ then $D$-branes are classified by the appropriately twisted version of $KR^{-i}$ where $i=9-p$.
\end{physics}
This follows from physical assertion \ref{phys:index} when we take into account the effect a reflection has on the dimension of the fixed set. If we have no reflections and the spacetime involution is not free then the entire space is fixed by the involution, showing the presence of a space filling $O9$-plane. For every reflection we now include it reduces the dimension of the fixed set, or equivalently the $O$-planes, by $1$, showing that the number of reflections equals $9-p$. This fails to work when the involution is free and there are no $O$-planes. As an example consider the type IIB orientifold with spacetime involution $z\mapsto z+\frac{1}{2}$ and the $T$-dual type IIA orientifold with involution $z\mapsto\bar{z}+\frac{1}{2}$. Both of these involutions are free so there are no $O$-planes. Clearly we cannot use the dimension of the $O$-planes to determine the index shift for $KR$-theory. However, we can still use physical assertion \ref{phys:index}. The involution of the the type IIB orientifold when written in real coordinates is $(x,y)\mapsto(x+\frac{1}{2},y)$, while for the type IIA orientifold is $(x,y)\mapsto(x+\frac{1}{2},-y)$. We see that the IIB orientifold has zero reflections while the IIA orientifold has $1$. Following physical assertion \ref{phys:index} we would expect $D$-branes in the IIB orientifold to be classified by $KR^0$ and by $KR^{-1}$ in the IIA orientifold, matching the results found in \cite{MR3267662}.

\section{Involutions on the product of two elliptic curves}
\label{sec:twoellcurves1}

In this section we will consider involutions, $\iota$, on the product of two elliptic curves, $Y=E_1\times E_2$, that are equivariant with respect to the projections $P_i\co Y\to E_i$ for $i=1,2$. This means that $\iota$ must restrict on each
$E_i$ to one of the involutions considered in \cite{MR3267662,MR3316647}, allowing us to build off of that work while giving enough cases to give enough insight on how to move to $K3$ surfaces. Another reason for making this restriction is that we are ultimately interested in the $F$-theory lifts of the elliptic curve orientifolds.

We will first restrict to cases where the $B$-field is trivial in all $T$-dual theories. That means we consider cases where both elliptic curves $E_j$ have purely imaginary complex and complexified K\"ahler structures.  So we can take
each $E_j$ to be of the form $\bC/(\bZ+\tau_j\bZ)$
with $\tau_j$ purely imaginary.

\subsection{The type I theory and its $T$-duals}
\label{sec:typeI}
Let us begin by considering the type I theory and its $T$-duals. The type I theory is the type IIB orientifold with trivial spacetime involution, so when compactified on $E_1\times E_2$ is the orientifold $\bR^{6,0}\times \left(S^{2,0}\right)^4$. There is an $O9$-plane wrapping the four compact dimensions. Therefore the index shift in the $KR$-theory $D$-brane classification is zero and we see that stable $Dp$-branes are classified by
\begin{align}
	KR^{p-5}(E_1\times E_2,\iota) &\cong KR^{p-5}\left(\left(S^{2,0}\right)^4\right)\nonumber\\
	&\cong KO^{p-5}\oplus4KO^{p-6}\oplus6KO^{p-7}\oplus4KO^{p-8}\oplus KO^{p-9}.
\end{align}
$T$-dualities of the type I theory are fairly straightforward. For each direction we perform a $T$-duality $S^{2,0}\to S^{1,1}$. The index shift is by one for each copy of $S^{1,1}$ or equivalently each $T$-duality performed. The matching of $D$-brane charges then follows from the fact that everything splits, along with the isomorphism
\begin{equation}
	KR^*\left(S^{2,0}\right)\cong KR^{*-1}\left(S^{1,1}\right).
\end{equation}
This is a good example of the index shifts discussed in section \ref{sec:DbraneClass} and the relationship between holomorphic/antiholomorphic involutions on the individual elliptic curves and type IIA/B orientifolds discussed more broadly in section \ref{sec:IIAB}. Furthermore, a similar structure will appear in all of the other cases.

There are four different directions we can $T$-dualize. If we give the elliptic curve $E_i$ the complex coordinates $z_i$ then the type I theory is described by the trivial spacetime involution $(z_1,z_2)\mapsto(z_1,z_2)$. Performing a single $T$-duality will give a type IIA orientifold with spacetime involution given by one of the four possibilities 
\begin{equation}
	(z_1,z_2)\mapsto(\pm\bar{z}_1,z_2),(z_1,\pm\bar{z}_2)
	\label{eq:6dIAinv}
\end{equation} depending on which of the four circles we $T$-dualize. In all cases we obtain $\bR^{6,0}\times\left(S^{2,0}\right)^3\times S^{1,1}$ as a Real space, with two $O8^+$-planes, one at each of the two fixed points in $S^{1,1}$ and wrapping the other three compact dimensions. This means that the index shift for classifying $D$-brane charges is $i=1$. $i$ being odd matches with the fact that these are all type IIA orientifolds. Note that these four theories are all related to each other by two $T$-dualities. To go from one where the $m$th circle is $S^{1,1}$ to one where the $n$th circle is $S^{1,1}$ we just perform $T$-dualities along the $m$th and $n$th circles. As mentioned in section \ref{sec:IIAB}, since we need $i$ to be odd for a type IIA orientifold and the compactification manifold has complex dimension $2$, the involutions \ref{eq:6dIAinv} are neither holomorphic nor antiholomorphic, but split as holomorphic on one elliptic curve and antiholomorphic on the other. $Dp$-branes in these IIA orientifold theories are classified by
\begin{align}
	KR^{p-6}(E_1\times E_2,\iota) &\cong KR^{p-6}\left(\left(S^{2,0}\right)^3\times S^{1,1}\right)\nonumber\\
	&\cong KR^{p-6}\left(\left(S^{2,0}\right)^3\right)\oplus KR^{p-5}\left(\left(S^{2,0}\right)^3\right)\nonumber\\
	&\cong KO^{p-6}\oplus 3KO^{p-7}\oplus 3KO^{p-8}\oplus KO^{p-9}\nonumber\\
    &\qquad\oplus KO^{p-5}\oplus 3KO^{p-6}\oplus 3KO^{p-7}\oplus KO^{p-8}\nonumber\\
	&\cong KO^{p-5}\oplus4KO^{p-6}\oplus6KO^{p-7}\oplus4KO^{p-8}\oplus KO^{p-9}.
\end{align}
All four of these theories can be viewed as the type IA theory compactified on a $3$-torus, as can be seen from the brane content.

If we now perform a second $T$-duality along a different circle from the first one we $T$-dualized there are six possibilities. The possible spacetime involutions are the two involutions of the form
\begin{equation}
	(z_1,z_2)\mapsto(\pm z_1,\mp z_2),
\end{equation}
or the four involutions of the form
\begin{equation}
	(z_1,z_2)\mapsto(\pm\bar{z}_1,\pm\bar{z}_2),
\end{equation}
depending on which two circles were $T$-dualized when starting from the type I theory. These are all type IIB orientifolds since they were obtained by $T$-dualizing a type IIA orientifold. As noted in section \ref{sec:IIAB}, a type IIB orientifold must have either a holomorphic or antiholomorphic involution since the compactification manifold has complex dimension $2$. In this case we see both possibilities arise. As noted earlier, however, they are all related by a change of coordinates. This is most easily seen by noting that as Real spaces they are all $\bR^{6,0}\times\left(S^{2,0}\right)^2\times\left(S^{1,1}\right)^2$. They all have four $O7^+$-planes, one located at each of the four fixed points of $S^{1,1}\times S^{1,1}$ and wrapping $S^{2,0}\times S^{2,0}$. This shows us that the index shift is $i=2$ and the $Dp$-brane charges are classified by
\begin{align}
	KR^{p-7}(E_1\times E_2,\iota) &\cong KR^{p-7}\left(\left(S^{2,0}\right)^2\times\left(S^{1,1}\right)^2\right)\nonumber\\
	&\cong KR^{p-7}\left(\left(S^{2,0}\right)^2\right)\oplus 2KR^{p-6}\left(\left(S^{2,0}\right)^2\right)\oplus KR^{p-5}\left(\left(S^{2,0}\right)^2\right)\nonumber\\
	&\cong KO^{p-5}\oplus4KO^{p-6}\oplus6KO^{p-7}\oplus4KO^{p-8}\oplus KO^{p-9}.
\end{align}

$T$-dualizing one of the remaining two compact directions will give a type IIA orientifold. There are four possible involutions,
\begin{equation}
	(z_1,z_2)\mapsto(-z_1,\pm\bar{z}_2),(\pm\bar{z}_1,-z_2).
\end{equation}
Again, since these are IIA orientifolds and the compact space has complex dimension $2$, the spacetime involution is neither holomorphic nor antiholomorphic, but splits as holomorphic on one elliptic curve and antiholomorphic on the other. The compact space for all of these orientifolds is $S^{2,0}\times\left(S^{1,1}\right)^3$ as a Real space. There are eight $O6^+$-planes, one each at the fixed points of $\left(S^{1,1}\right)^3$ and wrapping $S^{2,0}$. Therefore $i=3$ and the $Dp$-branes are classified by
\begin{align}
	KR^{p-8}(E_1\times E_2,\iota) &\cong KR^{p-8}\left(S^{2,0}\times\left(S^{1,1}\right)^3\right)\nonumber\\
	&\cong KR^{p-8}\left(S^{2,0}\right)\oplus 3KR^{p-7}\left(S^{2,0}\right)\nonumber\\
    &\qquad\oplus 3KR^{p-6}\left(S^{2,0}\right)\oplus KR^{p-5}\left(S^{2,0}\right)\nonumber\\
	&\cong KO^{p-5}\oplus4KO^{p-6}\oplus6KO^{p-7}\oplus4KO^{p-8}\oplus KO^{p-9}.
\end{align}

$T$-dualizing the final remaining copy of $S^{2,0}$ gives the type IIB orientifold with spacetime involution
\begin{equation}
	(z_1,z_2)\mapsto(-z_1,-z_2).
\end{equation}
This is $\left(S^{1,1}\right)^4$ as a Real space and has sixteen $O5^+$-planes, one at each of the fixed points. Therefore $i=4$ and the $Dp$-branes are classified by
\begin{align}
	KR^{p-9}(E_1\times E_2,\iota) &\cong KR^{p-9}\left(\left(S^{1,1}\right)^4\right)\nonumber\\
	&\cong KO^{p-5}\oplus4KO^{p-6}\oplus6KO^{p-7}\oplus4KO^{p-8}\oplus KO^{p-9}.
\end{align}
As we can see, the $D$-brane charges in all sixteen of these orientifold theories match, though their sources in terms of wrappings differ between theories. Mathematically, the matching of $D$-branes between the type I theory and the last IIB orientifold we considered, $(z_1,z_2)\mapsto(-z_1,-z_2)$, is described by:
	\begin{theorem}[\cite{MR3305978}]
If $E_1$ and $E_2$ are complex elliptic curves, then there is an
isomorphism
  \[
KO^{4-\bullet}(E_1\times E_2)\xrightarrow{\text{Poinc.\ duality}}
KO_\bullet(E_1\times E_2)
\xrightarrow{\cong} KR^{-\bullet}(E_1\times E_2, \iota),
\]
where $\iota$ is multiplication by $-1$ and the map $KO_\bullet(E_1\times E_2)
\xrightarrow{\cong} KR^{-\bullet}(E_1\times E_2, \iota)$ is
given by the real Baum-Connes isomorphism for a free abelian
group.
\end{theorem}

\subsection{The type \~{I} theory and its $T$-duals}

Let us look next at the type \~I theory and its $T$-duals. It is given by a free involution. That is we start out with the type IIB superstring theory compactified on $E_1\times E_2$ which we give complex coordinates $(z_1,z_2)$ and mod out by the action of $\iota\cdot\Omega$ where $\iota$ is the spacetime involution that fixes the non-compact dimensions and acts on the compact ones as
\begin{equation}
	(z_1,z_2)\mapsto(z_1,z_2)+\delta,
\end{equation}
where $\delta$ is a non-trivial point of order $2$. There are fifteen different choices for $\delta$, but for concreteness we will look at the case
\begin{equation}
	(z_1,z_2)\mapsto(z_1,z_2+\tfrac{1}{2}).
\end{equation}
There are no $O$-planes since the spacetime involution is fixed-point free. Therefore we cannot determine the index shift by the dimension of the $O$-planes, but we can see that the involution contains no reflections, so the index shift is $i=0$. As a Real space the spacetime manifold is 
$\bR^{6,0}\times\left(S^{2,0}\right)^3\times S^{0,2}$. 
Stable $Dp$-brane charges are classified by
\begin{align}
	KR^{p-5}(E_1\times E_2,\iota) &\cong KR^{p-5}\left(\left(S^{2,0}\right)^3\times S^{0,2}\right)\nonumber\\
	&\cong KSC^{p-5}\left(\left(S^1\right)^3\right)
    \qquad(\text{by \cite[Proposition 3.5]{MR0206940}})
    \nonumber\\
	&\cong KSC^{p-5}\oplus3KSC^{p-6}\oplus3KSC^{p-7}\oplus KSC^{p-8}.
\end{align}

If we $T$-dualize a single compact direction there are four possible type IIA orientifolds we can get depending on which dimension we $T$-dualize. Three of the possibilities are free:
\begin{equation}
	(z_1,z_2)\mapsto
    (\pm\bar{z}_1,z_2+\tfrac{1}{2}),(z_1,\bar{z}_2+\tfrac{1}{2}).
	\label{eq:freeIIA1}
\end{equation}
The other possibility,
\begin{equation}
	(z_1,z_2)\mapsto(z_1,-\bar{z}_2+\tfrac{1}{2}),
	\label{eq:tIAinv}
\end{equation}
occurs when you $T$-dualize the $S^{0,2}$. It has two disjoint fixed sets corresponding to $O8$-planes that must have opposite sign, so that the system has net zero $O$-plane charge. The compact space is $\left(S^{2,0}\right)^3\times S^{1,1}_{(+,-)}$ as a Real space. This can be viewed as the 
type $\widetilde{IA}$ theory compactified on 
$\left(S^{2,0}\right)^3$. It has an $O8^+$-plane located 
at one of the fixed points of $S^{1,1}$ wrapping 
$\left(S^{2,0}\right)^3$, and an $O8^-$-plane located at the other fixed point wrapping $\left(S^{2,0}\right)^3$. 

From the point of view of physics we see that the $O$-planes must have opposite charge for charge conservation. From the mathematical perspective this can be seen in two ways. The first is that the involution exchanges $2$-torsion points. Let $\delta_1$ be a $2$-torsion point of $E_1$ and $\delta_2=(\delta_a,\delta_b)$ be a $2$-torsion point in $E_2$ 
with $\delta_a=0$ or $\frac{1}{2}$. Then the spacetime involution acts on the $2$-torsion points as
\begin{equation}
	(\delta_1,0,\delta_b)\longleftrightarrow(\delta_1,\tfrac{1}{2},\delta_b).
\end{equation}
This tells us that the $O$-planes must come in pairs of opposite sign. Perhaps a better mathematical way to describe the sign choice that doesn't depend on a choice of zero is to look at the defining equations of $E_2$ in generalized Legendre form. As explained in \cite{MR3316647}, the zeros of the defining equation will correspond to $O^+$-planes if they are real. If the zero is imaginary then its complex conjugate must also be a zero, corresponding to a pair of $O$-planes of opposite sign.

It is not difficult to carry out calculations similar to those
in section \ref{sec:typeI} to check that the $D$-brane charges
agree in all the dual theories; taking the $O$-plane charges
into account in this is essential.

\subsection{Non-Trivial $B$-field}

There is one orientifold theory on an elliptic curve
with topologically nontrivial $B$-field, namely what
Witten \cite{MR1615617} calls the ``type I theory with
no vector structure.''  Its T-duality grouping includes
two other theories, the IIB theory on $S^{1,1}\times S^{1,1}$
with sign choice $(+,+,+,-)$, and the IIA theory on
a real elliptic curve of species $1$ (see 
\cite[\S5.2.3]{MR3316647}).  One can carry out similar
$D$-brane charge
calculations for this case,  but we omit the details.

\section{Involutions on the Kummer surface of a
  product of two elliptic curves}
\label{sec:twoellcurves}

The simplest K3 surfaces to describe are produced by the
\emph{Kummer construction} (see \cite{MR3077251})
on $A=E_1\times E_2$, a product of two elliptic curves.
Consider the involution $\iota\co x\mapsto -x$ on $A$;
this fixes the $2$-torsion subgroup, which is elementary
abelian of rank $4$, and thus $\iota$ has $16$ fixed points.
The Kummer construction is based on the quotient surface $A'=A/\{\pm 1\}$.
A disk around each fixed point looks like $\bC^2/\{\pm1\}$,
so we can cut out these disks and glue in copies of the disk
bundle of $T_{\text{holo}}\bP^1(\bC)$, which has the same boundary
($\bR\bP^3$).  From the point of view of algebraic geometry, 
this is
equivalent to blowing up the $16$ double points, getting a nonsingular
K3 surface $X=\Kum(E_1\times E_2)$.  Note that the complex structure
of $X$ and the Neron-Severi lattice $\NS(X)$ are completely determined by the
two complex parameters $\tau_1$ and $\tau_2$ (in $\mathbb{H}/\mathrm{PSL}_2(\bZ)$, $\mathbb{H}$
the upper half-plane) giving the complex structures on $E_1$ and $E_2$.
As explained in \cite{MR3077251}, choices of K\"ahler metrics on
the elliptic curves $E_1$ and $E_2$ (plus a choice of a cut-off parameter,
which in the end doesn't matter when we take it sufficiently large)
determine a Ricci-flat K\"ahler metric on $X$ with the same volume
as the singular variety $A'$. (First one glues in an Eguchi-Hanson
metric in place of the disk around each singular point of $A'$,
then one deforms to a smooth Ricci-flat K\"ahler metric.)
When one takes $B$-fields into account
also, $E_1$ and $E_2$ acquire ``complexified K\"ahler parameters''
$\rho_1$ and $\rho_2$, also lying in $\mathbb{H}/\mathrm{PSL}_2(\bZ)$, 
and mirror symmetry for elliptic curves
interchanges $\tau_j$ and $\rho_j$ \cite{MR1633036}.  There is a similar
induced mirror symmetry for Kummer surfaces.

The Kummer construction has appeared many times before in the
physics literature.  In particular, there is believed to be a
duality between type I string theory on a $4$-torus and
type IIA or type IIB string theory on a K3 surface
\cite{MR2033773,MR1402865,MR1408165,MR2010972,MR3065080}.
This duality is described by first showing that the type IIA theory
compactified on K3 is related to the $\SO(32)$ heterotic theory
theory by string-string duality. This gives a duality between the
strongly coupled heterotic theory and a weakly coupled IIA
theory. Points in the moduli space of the heterotic theory with
enhanced gauge symmetry correspond in the IIA theory to the K3
developing $A$-$D$-$E$ type singularities. So at a certain point in
the moduli space, the IIA theory is compactified on the ``orbifold limit
of $K3$,'' $\bT^4/\bZ_2$ with the $\bZ_2$ action given by
reflection \cite{MR1408388}.  (From the point of view of complex
geometry, this is the singular blow-down $A'$ of $\Kum(E_1\times E_2)$.)
$\bT^4/\bZ_2$ is variously referred to as an orbifold and
orientifold in different papers in the literature; we will use term
orientifold here. (We do not require the involution to reverse
orientation on the target space, but we do require equivariance with
respect to the worldsheet parity
operator, which reverses orientation on the string worldsheet.)

One way to see that the type IIA theory dual to the heterotic theory
must be an orientifold compactification is that the type IIA theory
must have $\mathcal{N}=1$ supersymmetry, just as the
heterotic theory does, if they are to be dual.  Since ordinary IIA
theory has $\mathcal{N}=2$ supersymmetry, an orientifolding is needed
to break half of the supersymmetry.
This duality is very similar to the duality
between the type IIB theory compactified on $\bT^2/\bZ_2$ and the
$\SO(32)$ heterotic theory compactified on $\bT^2$. The $\SO(32)$
heterotic theory compactified on $\bT^2$ is $S$-dual to the type I
theory compactified on $\bT^2$. It was shown in
\cite{Olsen:1999,MR1697703} that the type I
theory compactified on $\bT^2$ is $T$-dual to a type IIB orientifold
compactification on $\bT^2/\bZ_2$. Since the orientifold and orbifold
theories are not equivalent we see the duality must relate the
heterotic theory with a type II orientifold. We will now further use
the heterotic/type I $S$-duality in our study of string-string
duality. 

Starting with the weakly coupled type IIA theory compactified on K3,
string-string duality relates this to a strongly coupled $\SO(32)$
heterotic theory compactified on $\bT^4$, which is then $S$-dual to
the weakly coupled type I theory compactified on $\bT^4$. Since this
sequences of dualities relates the weakly coupled type IIA theory
compactified on K3 to the weakly coupled type I theory compactified
on $\bT^4$, we would expect these two theories to be related by a
$T$-duality.
\begin{physics}
	The type I theory compactified on $\bT^4$ is $T$-dual to the
        type IIB orientifold compactification on $\bT^4/\bZ_2$ with
        the $\bZ_2$ action given by reflection. 
\end{physics}

In this section, we study duality of involutions on Kummer surfaces of
products of two elliptic curves, building on the results in
\cite{MR3267662,MR3316647}.

The Kummer surface $X=\Kum(E_1\times E_2)$ can also be viewed as
an elliptically fibered K3 surface, in fact an \emph{isotrivial} one.
This means that there is a morphism of projective varieties (over $\bC$)
$f\co X\to \bP^1$ for which the general fibers $f^{-1}(x)$ are elliptic
curves with constant $j$-invariant (hence all isomorphic to one another).
To see this, consider the composite
\begin{equation}
\label{eq:isotrivial}  
f\co X \xrightarrow{\text{blow-down}} Y'\xrightarrow{\text{proj}_2}
E_2/\{\pm 1\}\cong \bP^1.
\end{equation}
It is obvious that the generic fibers are copies of $E_1$ and that the
singular fibers lie over the images in $\bP^1$ of the four $2$-torsion
points in $E_2$.  The singular fibers are all of Kodaira type
$I_0^*$ (with Euler characteristic $6$ and monodromy $-1$);
conversely, any isotrivial elliptically fibered K3 surface
with four $I_0^*$ fibers is isomorphic to a Kummer surface
(coming from a two-dimensional abelian variety divided out
by the action of $\{\pm 1\}$)
\cite[Proposition 1]{Sawon}.

Recall that a holomorphic involution on a single elliptic curve $E_j$
must be one of the following: (a) the trivial involution $1$, (b)
an involution with exactly four fixed points, such as $x\mapsto -x$, or
(c) a free involution\footnote{In this case,
  the involution does not respect the group structure of
  the elliptic curve.}, given by translation by an element of order $2$.
If we disregard twisting, the $KR$-groups are given  by
\[
KO^\bullet(T^2)=KO^\bullet \oplus KO^{\bullet-1}  \oplus KO^{\bullet-1} \oplus
KO^{\bullet-2}
\]
in case (a), $KO^{\bullet+2}(T^2)$ (the same groups, shifted up in degree by
$2$) in case (b), and $KSC^\bullet \oplus KSC^{\bullet-1}$ in case (c).

For the isotrivial fibration \eqref{eq:isotrivial},
the singular $I_0^*$ fibers are not irreducible, and
have a central $\bP^1$ (of
multiplicity $2$ as a divisor) surrounded by $4$ other $\bP^1$'s, each
joined at a distinct point of the central $\bP^1$.
The complement of the central $\bP^1$
has the structure of the abelian Lie group
$\bC \times (\bZ/2 \mathbb{Z}) \times (\bZ/2 \mathbb{Z})$
\cite[Table VII.3.4]{Miranda} and the Mordell-Weil group
is $(\bZ/2\mathbb{Z}) \times (\bZ/2\mathbb{Z})$. Corresponding to the three kinds of
holomorphic involutions on the elliptic curve $E_1$, we have 
on the Kummer surface $X$
the following holomorphic involutions compatible with
\eqref{eq:isotrivial} (and trivial on the base $\bP^1$):
(a) the trivial involution $1$, (b) the involution $\iota_1$ induced
by multiplication by $(-1,1)$ on $E_1\times E_2$, (c) a Nikulin
involution $\iota_2$ associated to translation by a non-zero element of the
Mordell-Weil group, i.e., a section $\tau$ of $f$ of order exactly $2$.
The involution $\iota_2$ is induced by $(z,w)\mapsto (z+z_0,w)$ on
$E_1\times E_2$, with $z_0$ a $2$-torsion point in $E_1$. Note that
$\iota_2$ has exactly two fixed points in the central $\bP^1$ of
each singular fiber, for the requisite count of $4\cdot 2 = 8$
fixed points in $X$.  The involution $\iota_1$ is given by $x\mapsto -x$
on the complement $\bC \times (\bZ/2\mathbb{Z}) \times (\bZ/2\mathbb{Z})$ of the central $\bP^1$,
and thus has fixed points at $0$ and $\infty$ on each of the
four peripheral $\bP^1$'s in a singular fiber.  Since these points at
$\infty$ correspond to the points of attachment to the central $\bP^1$,
$\iota_1$ restricted to the central $\bP^1$ is a holomorphic involution
with at least $4$ fixed points, and thus must be the identity.  So the fixed
set for $\iota_1$ consists of the four central $\bP^1$'s in the
four singular
fibers along with the four global sections of $f$
(which intersect the singular fibers at
$\{0\}\times (\bZ/2\bZ) \times (\bZ/2\bZ) \subset
\bC \times (\bZ/2\bZ) \times (\bZ/2\bZ)$).  In other words, it
consists of a union of eight disjoint rational curves.

\begin{theorem}
\label{thm:KRI0*}
Let $C$ be an $I_0^*$ fiber of a Kummer elliptic fibration
\textup{\eqref{eq:isotrivial}} and let $\iota_1$ and $\iota_2$ be
the involutions discussed above.  Then
\[
KO^\bullet(C)\cong KO^\bullet(S^2)\oplus 4KO^{\bullet-2}
\cong KO^\bullet \oplus 5 KO^{\bullet-2},
\]
\[
KR^\bullet(C, \iota_1)\cong  KO^\bullet(S^2)\oplus 4KO^{\bullet+2}
\cong KO^\bullet\oplus KO^{\bullet-2} \oplus 4 KO^{\bullet+2},
\]
and
\[
KR^\bullet(C, \iota_2)\cong KR^\bullet(\bP^1(\bC), x\mapsto -x) 
\oplus 2 K^{\bullet} \cong KO^\bullet\oplus KO^{\bullet+2}
\oplus 2 K^{\bullet}.
\]
\end{theorem}
\begin{proof}
  The inclusion of the central $\bP^1$ is split by the map collapsing
  each peripheral $\bP^1$ to its point of attachment with the central
  $\bP^1$, and the splitting is also equivariant for $\iota_1$ and $\iota_2$.
  The three calculations follow immediately, since
  $KR^\bullet(\bC, x\mapsto -x)=KR^\bullet(\bR^{0,2})=KO^{\bullet+2}$
  and $KR^\bullet(\bC\times S^{0,1})\cong K^\bullet$.
\end{proof}
\begin{remark}
\label{rem:I0*}
If we ignore $2$-torsion, Theorem \ref{thm:KRI0*} shows that
for the trivial and $\iota_1$ involutions (and the result would be
the same even if we add a $2$-torsion twist), the $KR$ theory is
of rank $1$ in degrees $0$ mod $4$, of rank $5$ in degrees $2$ mod
$4$, and of rank $0$ in odd degrees.
For the $\iota_2$ involution, however, the $KR$ theory is
of rank $3$ in all even degrees, of rank $0$ in odd degrees. Thus
there cannot be an isomorphism between the $KR$-theory of some
twist of $(C,\iota_1)$ and a shift of some
twist of $(C,\iota_2)$.
\end{remark}

\section{Elliptic K3 surfaces of Picard rank 17}
\label{sec:rank17fam}

We consider complex projective K3 surfaces that admit at
least one Jacobian elliptic fibration. To establish some notation, we
define a \emph{Jacobian elliptic fibration} on a K3 surface $X$ to be a
pair $(\pi_X,\sigma_X)$ consisting of a proper map
of analytic spaces $\pi_X \co X \to
\mathbb{P}^1 =\mathbb{P}(u, v) $, whose general fiber in $\mathbb{P}^2
= \mathbb{P}(X, Y, Z)$ is a smooth curve of genus one, and a section
$\sigma_X \co \mathbb{P}^1 \to X$ of the
elliptic fibration $\pi_X$. The group of sections of the
Jacobian fibration (with $\sigma_X$ taken to be the zero-section)
is the \emph{Mordell-Weil group}
$\mathrm{MW}(X)$.  A Jacobian
elliptic K3 surface $X$ always admits the holomorphic
antisymplectic involution acting as the hyperelliptic involution in each
smooth fiber. A Weierstrass model exhibits $X$ as a double cover of the
Hirzebruch surface $\mathbb{F}_4$, branched on the section
$\sigma_X$ and the trisection defined by its affine
Weierstrass equation \cite[Remark 11.2.2]{Huybrechts}. 
\par We denote the N\'eron-Severi lattice of the K3 surface
$X$ by $\mathrm{NS}(X)$. This lattice is known to
be even and have a signature of $(1, \rho_X-1)$, where
$\rho_X$ is the Picard rank  of $X$.  A lattice
polarization on $X$ is defined as a primitive lattice
embedding $i\colon L \hookrightarrow \mathrm{NS}(X)$, where
$i(L)$ contains a pseudo-ample class. In general, $L$-polarized K3
surfaces form interesting families in the context of mathematical
physics since they are classified, up to isomorphism, by a coarse moduli
space, which is a quasi-projective variety of dimension $20-\rho_L$.  
\par {\small \emph{Notation:} In the following, we utilize the
  following notations: $L_1 \oplus L_2$ refers to the orthogonal
  direct sum of the two lattices $L_1$ and $L_2$, $L(\lambda)$ is
  obtained by multiplying the form of lattice $L$ with $\lambda \in
  \mathbb{Z}$, $\langle M \rangle$ denotes a lattice with Gram matrix $M$
  in some basis. 
$L^\vee$ is the dual lattice of $L$. The lattices $A_n$, $D_m$, and
  $E_k$ are the positive definite root lattices for the respective
  root systems. $H$ is the unique even unimodular hyperbolic rank-two
  lattice, and $N$ is the negative definite rank-eight Nikulin
  lattice, as defined, for instance, in \cite[Sec.~5]{MR728142}.} 
\medskip
\subsection{Imposing the existence of holomorphic involutions} 
\par We will determine specific families of $L$-polarized K3 surfaces,
by imposing additional symmetries that the elliptic fibration shall
exhibit. To do so, we focus on Jacobian elliptic K3 surfaces
$X$ that admit compatible holomorphic involutions.
We consider (i) van-Geemen-Sarti involutions \cite{MR2274533}, i.e.,
symplectic involutions that arise as fiberwise translation by a
section of order two in
$\mathrm{MW}(X)$, and (ii)
antisymplectic involutions that arise as involutions on the base curve
$\mathbb{P}(u, v)$. In the latter case,  the K3 surface $X$
is a base change of order two on a rational elliptic fibration
$\pi_R\colon R \to \mathbb{P}^1=\mathbb{P}(U,V)$,
and the involution is the covering involution of a double cover
$\mathbb{P}(u, v) \to \mathbb{P}(U,V)$ of the corresponding base curve. One easily
checks that the minimal Picard rank  for a K3 surface $X$
admitting either of such holomorphic involutions is 10. In particular,
the lattice polarization is $L = H \oplus N$ in case (i), and $L = H
\oplus E_8(-2)$ in case (ii); see \cite{MR2274533,MR4069236}.  
\par For a Jacobian elliptic K3 surface $X$ to admit two
compatible holomorphic involutions, a Picard rank  of at least 14 is
required. For the existence of two van Geemen-Sarti involutions, the
polarizing lattice has to contain $H \oplus 2 D_4(-1) \oplus
4 A_1(-1)$. Similarly, for the existence of two (additional)
commuting antisymplectic involutions or one van Geemen
Sarti-involution and one antisymplectic involution, the polarizing
lattice has to be $H \oplus K_0(-1)$ or $H \oplus N_0(-1)$,
respectively. Here, $K_0$ and $N_0$ are positive definite lattices of
rank 12,  and determinant $2^6$ and $2^8$, respectively; their Gram matrices
were determined in \cite{MR4444083}. 
\par For a Jacobian elliptic K3 surface $X$ to admit three
compatible holomorphic involutions, a Picard rank  of at least 16 is
required. This can be seen by constructing explicitly the families of Jacobian elliptic K3 surfaces
with either two symplectic involutions and one antisymplectic involution or one 
symplectic involution and two antisymplectic involutions.
The former case is realized as follows:
an $H \oplus 2 D_4(-1) \oplus 4 A_1(-1)$-polarized
K3 surface $X$ admits an additional antisymplectic
involution if $X$ is also a base change of order 2 on the
rational elliptic fibration with singular fibers $6 I_2$ and
Mordell-Weil group $(\mathbb{Z}/2\mathbb{Z})^2 \oplus (A_1^\vee)^{\oplus
  2}$. Then, $X$ admits a Jacobian elliptic fibration with
singular fibers $12 I_2$ and Mordell-Weil group $(\mathbb{Z}/2\mathbb{Z})^2
\oplus \langle 1 \rangle^{\oplus 2}$, whence $X$ must have
Picard rank  16. The general member $X$ of this family
admits the additional symplectic involution obtained by the
composition of the aforementioned antisymplectic involution and the
fiberwise hyperelliptic involution. Via the Nikulin construction, carried out using the latter involution, one
obtains a K3 surface $X'$ equipped with a Jacobian elliptic
fibration with singular fibers $2 I_0^* + 6 I_2$ and Mordell-Weil
group $(\mathbb{Z}/2\mathbb{Z})^2$. One easily checks that this implies a
polarizing lattice for $X'$ equal to $H \oplus E_8(-1) \oplus 6 A_1(-1)$. It is known that
such a K3 surface $X'$ has a birational model as double
cover of the projective plane branched on 6 lines in general
position. The family was studied in \cite{MR1233442,
  MR1267602}. 
\par Similarly, if a van-Geemen-Sarti involution and
two commuting antisymplectic involutions (induced by involutions on
the base curve) are supported on a K3 surface $X$, one can
iterate the above construction and obtain a tower of K3 surfaces. Via the Nikulin construction one first
obtains a K3 surface $X'$ as above, and then using the
second antisymplectic involution directly, one obtains a rational
elliptic fibration $\pi_{R''}\colon R'' \to
\mathbb{P}^1$.  The rational elliptic surface has to have at least one
singular fiber of type $I_0^*$ and a Mordell-Weil group containing a
factor of $\mathbb{Z}/2\mathbb{Z}$. Using results in \cite{MR1104782}, it
follows that the family of K3 surfaces $X'$ of lowest Picard rank  is
obtained if one starts with the rational elliptic surface
$R''$ with singular fibers $I_0^* + 2 I_2 + 2 I_1$ and
Mordell-Weil group $\mathbb{Z}/2\mathbb{Z} \oplus (A_1^\vee)^{\oplus 2}$. In
this case, $X'$ supports a Jacobian elliptic fibration with
singular fibers $2 I_0^* + 4 I_2 + 4 I_1$ and Mordell-Weil group 
$\mathbb{Z}/2\mathbb{Z}  \oplus \langle 1 \rangle^{\oplus 2}$, whence
$X$ has Picard rank  16. 
\par We are interested in the situation when the Jacobian elliptic K3
surface $X$ admits two van-Geemen-Sarti involutions and two
(additional) commuting antisymplectic involutions (induced by
involutions on the base curve). We have the following: 
\begin{proposition}
\label{lem:family}
The general $\langle 8 \rangle \oplus 2 D_8(-1)$-polarized K3
surface admits a Jacobian elliptic fibration with 2 van-Geemen-Sarti
involutions and 2 antisymplectic involutions induced by holomorphic
involutions on the base curve $\mathbb{P}^1$.
\end{proposition}
\begin{proof}
For  a Jacobian elliptic K3 surface $X$ to admit  two
van Geemen-Sarti involutions and two commuting antisymplectic
involutions (induced by involutions on the base curve), the rational
elliptic surface $\pi_{R'}\colon R' \to
\mathbb{P}^1$ has to have singular fibers $I_0^* + 3 I_2$ and
Mordell-Weil group $(\mathbb{Z}/2\mathbb{Z})^2 \oplus A_1^\vee$. In turn the
Jacobian elliptic K3 surface $X'$ has singular fibers $2
I_0^* + 6 I_2$ and Mordell-Weil group
$(\mathbb{Z}/2\mathbb{Z})^2 \oplus \langle 1 \rangle$.
It follows that $X'$  has as a birational model a
double cover of the projective plane branched on six lines tangent to a
conic, and is a Jacobian Kummer surface.  It was proved in
\cite{MR4376094} that $X$ is a Kummer surface associated
with an abelian surface of polarization of type $(1,2)$. It is well
known that the polarization of such a Kummer surface is
$\langle 8 \rangle \oplus 2 D_8(-1)$ as the  transcendental lattice is
$H(2) \oplus H(2) \oplus \langle -8 \rangle$.
\end{proof}
\subsubsection{Establishing a normal form} 
\label{ssec:normal_form}
The motivation for considering the family in
Proposition~\ref{lem:family} is the following: in Section~\ref{sec:twoellcurves} we considered Kummer surfaces associated
with the product of two non-isogenous elliptic curves. We also assumed that each elliptic curve was equipped with a level-2
structure, so that the modular $\lambda$-invariant of one
elliptic curve determined the smooth elliptic fiber of an isotrivial
fibration, whereas the $\lambda$-invariant of the other
determined the cross-ratio of the base points of the singular fibers.
\par It is known that such a K3 surface has a birational model as a
double-quadric 
surface, i.e., a double cover of the Hirzebruch surface
$\mathbb{F}_0=\mathbb{P}^1\times\mathbb{P}^1$ branched along a curve
of bi-degree $(4,4)$, or more precisely, a section in the line bundle
$\mathcal{O}_{\mathbb{F}_0}(4,4)$.   Every such cover has two natural
elliptic fibrations corresponding to the two rulings of the quadric
$\mathbb{F}_0$ coming from the two projections $\pi_i\colon \mathbb{F}_0
\to \mathbb{P}^1$ for $i=1,2$.  In the special case of a Kummer
surface associated with the product of two elliptic curves, the
$(4,4)$-class splits as the union of a $(4,0)$-class and a
$(0,4)$-class. Each class is determined by a quartic polynomial in
$\mathbb{P}^1$, invariant under the action of two non-trivial
holomorphic involutions because of the level structure. For such a
birational model, interchanging the roles of base and fiber can be viewed
as simply interchanging rulings. This establishes a
convenient normal form for the Kummer surfaces associated with
products of elliptic curves equipped with level-2 structure. 
\par As we shall see, the K3 surfaces in Proposition~\ref{lem:family}
admit an analogous normal form:
let us consider the family of double-quadric surfaces $X$ over
$\mathbb{F}_0=\mathbb{P}(u, v) 
\times\mathbb{P}(x, z)$ over $\mathbb{P}^3$, given by the equation 
\beq
\label{eqn:K3_X}
X\colon \ y^2  = \big(\rho u^4 + \kappa u^2 v^2 + \rho v^4\big) x^4 +
\big(\mu u^4 + \lambda u^2 v^2 + \mu v^4\big) x^2 z^2 + \big(\rho u^4
+ \kappa u^2 v^2 + \rho v^4\big) z^4 , 
\eeq
with $[ \rho: \kappa : \lambda : \mu ] \in \mathbb{P}^3$. One checks
the following by a direct computation: 
\begin{lemma}
\label{lem:NormFormisK3}
The minimal resolution of $X$ is a K3 surface for $[ \rho: \kappa :
  \lambda : \mu ] \in  \mathbb{P}^3 \backslash \mathcal{D}$ where we
set 
\beqn
 \mathcal{D}  = \Big\lbrace [ 0 : 0: \lambda : \mu ] \Big\rbrace \cup   \Big\lbrace [ 0 : \kappa: \lambda : 0 ] \Big\rbrace \cup   \Big\lbrace [ \rho : \kappa: \pm 2\kappa : \pm 2\rho ] \Big\rbrace 
 \cup  \Big\lbrace [ \rho : \pm2  \rho: \pm 2\mu :\mu ] \Big\rbrace.
\eeqn 
Moreover, the elliptic fibration is isotrivial for $[ \rho: \kappa :
  \lambda : \mu ] \in  \mathbb{P}^3 \backslash \mathcal{D}$ with
$\kappa \mu=\lambda\rho$. 
\end{lemma}
\begin{proof}
The statement about isotriviality follows by computing the
$j$-invariant of a fiber over $\mathbb{P}(u, v)$. 
\end{proof}
By a slight abuse of notation, we will---from now on---not distinguish
any further between a double-quadric or double-sextic surface and the
K3 surface obtained from it by its minimal resolution. The following is immediate: 
\begin{lemma}
\label{lem:holinv1}
$X$ admits the following commuting holomorphic involutions:
\beqn
\begin{split}
 \xi^{(1)}_X\colon& [u:v] \mapsto  [-u:v], \quad  \xi^{(2)}_X\colon [x:z] \mapsto  [-x:z], \\
 \zeta^{(1)}_X\colon& [u:v] \mapsto  [v:u],\phantom{-}  \quad  \zeta^{(2)}_X\colon [x:z] \mapsto  [z:x],\phantom{-} 
 \end{split}
\eeqn
and $k_X\colon y \mapsto -y$.
\end{lemma}
\par A choice of ruling on $\mathbb{F}_0=\mathbb{P}(u, v)\times\mathbb{P}(x, z)$ induces an elliptic
fibration $\pi_X\colon X \to \mathbb{P}^1$. In
this context, one pair of involutions $(\xi^{(i)}_X,
\zeta^{(i)}_X)$ will constitute two (additional) commuting
antisymplectic involutions induced by involutions on the base
curve. The other pair of involutions $(\xi^{(i')}_X,
\zeta^{(i')}_X)$ with $\{ i, i' \} = \{1, 2\}$ will act
fiberwise and after composition with $k_X$ symplectically. 
\par We write Equation~(\ref{eqn:K3_X}) as the quadric surface $\pi_X\colon X \rightarrow \mathbb{P}(u, v)$, given by
 \beq
 \label{eqn:B12_1b}
  X\colon \quad y^2 = a(u, v) x^4 + b(u, v) x^2z^2 + a(u, v)  z^4 ,
\eeq
with
\beq
a(u, v) =  \rho u^4 + \kappa u^2 v^2 + \rho v^4, \qquad b(u, v) = \mu u^4 + \lambda u^2 v^2 + \mu v^4\,.
\eeq
Over the complex numbers, $X$ always admits sections; see Section~\ref{ssec:real_normal_form} for explicit equations. Thus, as a complex elliptic K3 surface, Equation~(\ref{eqn:B12_1b}) is equivalent to its associated Jacobian elliptic surface. It is easy to check  that the latter admits the following Weierstrass equation:
\begin{equation}
\label{eqn:EFS_0} 
  Y^2 Z  = X \Big( X - \big( b(u, v)  + 2 a(u, v)\big) Z \Big)  \Big( X - \big( b(u, v)  - 2 a(u, v)\big) Z \Big)\,.
\end{equation}
We have the following:
\begin{lemma}
\label{lem:fibrations}
The elliptic fibration $\pi_X\colon X \to
\mathbb{P}(u, v)$ admits sections. In the general case, the elliptic
fibration has singular fibers $12 I_2$ and the torsion Mordell-Weil group
$(\mathbb{Z}/2 \mathbb{Z})^2$. 
\end{lemma}

\par Unless stated otherwise, we will consider the family over the affine set $\rho=1$, i.e.,
for $[1: \kappa : \lambda : \mu ] \in \mathbb{P}^2 \backslash
\mathcal{D}$, in the remainder of this section.  
We make the following:
\begin{remark}
\label{rem:kummer}
If $\kappa \mu=\lambda\rho\neq 0$ then $X$ is isomorphic to the Kummer surface $\mathrm{Kum}(E_1 \times E_2)$ where the elliptic curves are given by
\beq
\begin{split}
 E_1\colon \quad & y_1^2 = u^4 + \frac{\kappa}{\rho}  u^2v^2 +v^4, \\
  E_2\colon \quad & y_2^2 = x^4 + \frac{\mu}{\rho}  x^2z^2 +z^4.
\end{split}
\eeq
\end{remark}
We have the following:
\begin{lemma}
Elliptic fibrations $\pi_X\colon X \to \mathbb{P}(u, v)$  for $(\kappa, \lambda, \mu)$ and
$(\kappa', \lambda', \mu')$ are isomorphic if and only if
$(\kappa, \lambda, \mu) = ((-1)^k \kappa', (-1)^\ell \lambda', (-1)^m \mu')$
for $k+m+\ell \equiv 0 \mod{2}$.
\end{lemma}
\par  With the explicit form for the elliptic fibration and its sections established, it is easy to check that Equation~(\ref{eqn:K3_X}) is indeed the family from Proposition~\ref{lem:family}:
\begin{proposition}
\label{prop:lattice}
For general parameters $[ \rho: \kappa : \lambda : \mu ] \in  \mathbb{P}^3 \backslash \mathcal{D}$ in Equation~(\ref{eqn:K3_X}) one has $\mathrm{NS}(X) \cong \langle 8 \rangle \oplus 2 D_8(-1)$.
\end{proposition}
A general K3 surface $X$ above corresponds to a point $[ \rho: \kappa : \lambda : \mu ] \in  \mathbb{P}^3 \backslash \mathcal{D}$ of the complement of hypersurfaces.  We also have the following:
\begin{corollary}
The elliptic fibration  in Equation~(\ref{eqn:K3_X}) has singular fibers $2 I_4 + 8  I_2$ for $[ \rho: \kappa : \lambda : \mu ] \in  \mathbb{P}^3 \backslash \mathcal{D}$ satisfying one of the following conditions:
\beq
 \rho =0 , \ \kappa \pm 2 \rho = 0, \ \mu \pm 2 \rho = 0, \ \lambda + 2 \mu \pm 2 \kappa \pm 4 \rho  =0 , \  \lambda - 2 \mu \pm 2 \kappa \mp 4 \rho  =0.
\eeq 
\end{corollary}
\subsection{Relation to Abelian surfaces}
\label{ssec:KUM}
In  Equation~(\ref{eqn:K3_X}) we introduced a family of K3 surfaces. We will now explain how this family generalizes Kummer surfaces $\mathrm{Kum}(E_1 \times E_2)$ associated
with the product of two non-isogenous elliptic curves $E_1, E_2$. From
the point of view of the total space, generalization is achieved by
considering Kummer surfaces associated with abelian surfaces that are
no longer isogenous to products of elliptic curves. Previously, we had also
assumed that the elliptic curves $E_1$ and $E_2$ were equipped with
level-2 structures. As we are considering a Jacobian Kummer surface
associated with a smooth genus-2 curve $C$, the analogue of fixing a
level-2 structure on each elliptic curve turns out to be fixing a
level-two structure on the genus-2 curve $C$ \emph{and} choosing a
marking of a G\"opel group $G=  \langle \mathscr{L} \rangle \oplus
\langle \mathscr{L}' \rangle$ on $\mathrm{Jac}(C)$. 
\subsubsection{Construction of a Jacobian Kummer surface}
A smooth genus-2 curve $C$ with level-2 structure is given in affine coordinates $(\xi,\eta)$ by its Rosenhain normal form, the analogue of the Legendre normal form for an elliptic curve:
\beq
\label{eqn:Rosenhain}
 C: \quad \eta^2 = \xi \,\big(\xi-1) \, \big(\xi- \lambda_1\big) \,  \big(\xi- \lambda_2 \big) \,  \big(\xi- \lambda_3\big) \,.
\eeq 
The ordered tuple $(\lambda_1, \lambda_2, \lambda_3)$ -- where the
$\lambda_i$ are pairwise distinct and different from
$(\lambda_4,\lambda_5,\lambda_6)=(0, 1, \infty)$ -- determines a point
in the moduli space of curves of genus two with level-2 structure.
The Weierstrass points of $C$ are the six points  $p_i\colon
(\xi,\eta)=(\lambda_i,0)$ for $i=1, \dots,5$, and  the point $p_6$ at
infinity.  For the remainder of the section, we will assume that $C$ is smooth.
 \par The \emph{Siegel three-fold} is the quasi-projective variety of dimension three, obtained from the Siegel upper half-plane $\mathbb{H}_2$ of degree two, divided by the action of the modular transformations $\Gamma_2:= \mathrm{Sp}_4(\mathbb{Z})$, i.e., 
\beq
 \mathcal{A}_2 =  \mathbb{H}_2 / \Gamma_2 \;.
\eeq
Here, $\mathbb{H}_2$ is the set of two-by-two symmetric period
matrices $\underline{\tau}$ over $\mathbb{C}$ whose imaginary part is
positive definite. In fact, $\mathcal{A}_2$ is precisely the moduli
space of principally polarized abelian varieties which are surfaces of
the form $A = \mathrm{Jac}(C)$ where $C$ is a smooth genus-2 curve. We
define the subgroup $\Gamma_2(2) = \lbrace M \in \Gamma_2 | \, M
\equiv \mathbb{I} \mod{2}\rbrace$  such that
$\Gamma_2/\Gamma_2(2)\cong S_6$ is identified with the permutation
group of six elements, representing the permutations of the
Weierstrass points $(\lambda_1, \lambda_2, \lambda_3, 0, 1,
\infty)$. Then, $\mathcal{A}_2(2)$ is the three-dimensional moduli
space of principally polarized abelian surfaces with level-2 structure
and affine coordinates $(\lambda_1, \lambda_2, \lambda_3)$. 
\par We denote the hyperelliptic involution on $C$ by $\imath_C$. For its symmetric product $C^{(2)}$, the quotient $C^{(2)}/\langle \imath_C \times  \imath_C \rangle$ is realized as a variety in terms of the variables $T=\xi^{(1)}\xi^{(2)}$, $X=\xi^{(1)}+\xi^{(2)}$, and $Y=\eta^{(1)}\eta^{(2)}$ with the affine equation
\beq
\label{kummer_middle}
  Y^2 = T \big(  T  - X +  1 \big)  \prod_{i=1}^3 \big( \lambda_i^2 \, T  -  \lambda_i \, X +  1 \big) \,.
\eeq
We introduce the following six lines $\ell_1, \dots, \ell_6$ in $\mathbb{P}^2 = \mathbb{P}(z_1, z_2, z_3)$:
\beq
\label{eqn:lines}
 \ell_{i} \colon z_1 + \lambda_i^2 z_2 - \lambda_i z_3 =0 \ (\text{for $i=1, 2, 3$}), \  \ell_4 \colon z_3 =0, \  \ell_5 \colon z_1 + z_2 - z_3 =0 , \  \ell_6\colon z_1 =0.
\eeq  
The lines are tangent to the conic $z_3^2 -4 z_1 z_2=0$. Then,
Equation~\eqref{kummer_middle} is precisely the double cover of the
projective plane $\mathbb{P}^2$ branched on $\ell_1, \dots, \ell_6$ in
the affine chart $z_1=1, z_2=T, z_3=X$. In this way, the affine
variety in Equation~\eqref{kummer_middle} completes to a hypersurface
in $\mathbb{P}(1,1,1,3)$ called the \emph{Shioda double-sextic}. The
double-sextic surface has the 15 singular points $p_{ij}$ with
$1 \le i < j \le 6$ which are the images of pairs of Weierstrass points
$\{ p_i, p_j \}$ in $C^{(2)}$. 
\par Conversely, we can start with a singular Kummer variety in
$\mathbb{P}^3$, i.e., a  nodal quartic hypersurface with sixteen
nodes, and reconstruct a configuration of six lines; see
\cite{MR2964027}.  To do so, we fix a singular point $p_0$ and
identify the lines in $\mathbb{P}^3$ through the point $p_0$ with
$\mathbb{P}^2$, and map any line in the tangent cone of $p_0$ to
itself. Then, any other line through $p_0$ meets the quartic surface
generically in two more points and  with multiplicity two at the other
nodes. We obtain a double cover of $\mathbb{P}^2$ branched along  a
plane curve of degree six  where all nodes of the quartic, different
from $p_0$, map to nodes of the sextic. By the genus-degree formula,
the maximal number of nodes is attained when the curve is a union of
six lines, in which case we obtain fifteen nodes $p_{ij}$ apart from
$p_0$. Since $p_0$ is a node, the tangent cone to this point is mapped
to a conic, and this conic is tangent to the six lines. Thus, it
follows that a Shioda double-sextic surface is  birational to a Kummer
surface. 
\subsubsection{Determining the associated G\"opel group}
Next, we consider the principally polarized abelian surface $A=\mathrm{Jac}(C)$ with the standard theta divisor $\Theta \cong [C]$. Each 2-torsion point $p_{ij} \in A[2]$ with $1\le i < j \le 6$ projects to a node on the Kummer surface $\mathrm{Kum}(A)$ which we still denote by $p_{ij}$. We also set $p_0= [0] \in A[2]$.  Translations on $A$ by a point of order two are isomorphisms of the Jacobian and map $A[2]$ to itself.  Any other theta divisor in $\mathrm{NS}(A)$ is a translate of $\Theta$ by $p_{ij}$ and is mapped in $\mathbb{P}^3$ to the intersection of the Kummer quartic with a plane.  We call such a singular plane a \emph{trope}. We obtain a $16_{6}$-configuration of sixteen nodes and sixteen tropes in $\mathbb{P}^3$, where each contains six nodes, and such that the intersection of each two is along two nodes. In the complete linear system $|2 \Theta|$ on $A$, the odd symmetric theta divisors $\Theta_{i6} = p_{i6} + \Theta$ (with $1 \le i \le 6$ and $p_{66}=p_0$) on $\mathrm{Jac}(C)$ are mapped to six tropes $T_{i}$ with $i=1, \dots,6$. Projections of the six tropes $T_{i}$ are precisely the six lines $\ell_i$ in Equation~(\ref{eqn:lines}), and $\ell_i$ contains the six nodes $p_0$ and $p_{ij}$ with $j \neq i$.  The remaining $10$ tropes $\Theta_{ij6}$ with $1\le i < j < k \le 5$  are easily computed as well; see  \cite[Sec.~3.7]{MR1406090}. 
\par The sixteen points of order two in $A=\mathrm{Jac}(C)$ can also be obtained using the embedding of the curve into the connected component of the identity in the Picard group, i.e., $C \hookrightarrow \mathrm{Jac}(C) \cong \mathrm{Pic}^0(C)$ with $p \mapsto [p -p_6]$. The  space $A[2]$ admits a symplectic bilinear form, called the \emph{Weil pairing}. The two-dimensional, maximal isotropic subspaces $G$ of $A[2]$ with respect to the Weil pairing are called \emph{G\"opel groups}. It is easy to check $G \cong (\mathbb{Z}/2 \mathbb{Z})^2$, and that there are exactly 15 inequivalent G\"opel groups. It is also well known that $A'= A/G$ is again a principally polarized abelian surface~\cite[Sec.~23]{MR2514037}. The corresponding isogeny $\Xi'\colon A \to A'$ between principally polarized abelian surfaces has as its kernel $G \leqslant A[2]$ and is called a \emph{$(2,2)$-isogeny}.  It follows that there is also dual $(2,2)$-isogeny $\Xi\colon A' \to A$.  In the case $A=\mathrm{Jac}(C)$ one knows that the $(2,2)$-isogenous abelian surface $A'=A/G$ satisfies $A' =\mathrm{Jac}(C')$ for some smooth curve of genus two $C'$.  This curve  $C'$ can be constructed explicitly: the relationship between the geometric moduli of the two curves was found by Richelot \cite{MR1578135}; see also \cite{MR970659,MR4421430}.  
\par We will fix one G\"opel group $G$ and consider a marking $G \simeq \langle \mathscr{L} \rangle \oplus  \langle \mathscr{L}' \rangle$, with $\mathscr{L}$, $ \mathscr{L}'$ line bundles of order two on the curve $C$. Then,  $\mathscr{L}$ determines a 2-isogeny of abelian surfaces $\Psi \colon B \rightarrow  \mathrm{Jac}(C)$ so that $B$ carries a canonical $(1,2)$-polarization $\mathscr{V} =  \Psi^*(\mathscr{L})$ with $\mathscr{V}^2=4$ and $h^0(\mathscr{V})=2$. The polarization line bundle $\mathscr{V}$ defines a canonical map $\varphi_{\mathscr{V}}\colon B \to \mathbb{P}^1$, such that the linear system $|\mathscr{V}|$ is a pencil on $B$, and each curve in $|\mathscr{V}|$ has self-intersection equal to $4$.  Since we assume $\rho_B=1$, the abelian surface $B$ cannot be a product of two elliptic curves or isogenous to a product of two elliptic curves.  Following the work in \cite{MR946234, MR2729013}, a general member of the pencil over $\mathbb{P}^1_{(u)}$ is, in the generic case, a smooth curve $D_u \subset B$ of genus three.  An explicit equation for the pencil $D_u$ was determined in \cite{MR4376094}.
\par We summarize the construction as follows: (i) we fix a smooth genus-2 curve $C$ with level-2 structure as in~(\ref{eqn:Rosenhain}). (ii) On the associated principally polarizes abelian surface $A=\mathrm{Jac}(C)$ we fix a G\"opel group. Without loss of generality, we can assume $G=\{ 0, p_{15}, p_{23}, p_{46} \}$.   Using the embedding of the curve into the Picard group, we associate $G$ with the pairing of the Weierstrass points of $C$, given by $(\lambda_1,\lambda_5=1)$, $(\lambda_2,\lambda_3)$, $(\lambda_4=0,\lambda_6=\infty)$. The $(2,2)$-isogenous abelian surface $A'=A/G$ is realized as $A' =\mathrm{Jac}(C')$ for some smooth genus-2 curve $C'$. From the quotient map we obtain $\Xi'\colon A \to A'$ and the dual isogeny $\Xi\colon A' \to A$. (iii) For the G\"opel group $G$ we fix a marking $G \simeq \langle \mathscr{L} \rangle \oplus  \langle \mathscr{L}' \rangle$, where line bundles $\mathscr{L}$ and $\mathscr{L}'$ of order two on $C$ are generated by $p_{46}$ and  $p_{15}$, respectively. $\mathscr{L}$ determines a 2-isogeny of abelian surfaces $\Psi \colon B \rightarrow  A=\mathrm{Jac}(C)$ that endows $B$ with a canonical $(1,2)$-polarization. In turn, the pullback $\Psi^*(\mathscr{L}')$ determines a 2-isogeny $\Phi' \colon A' \rightarrow B$ such that $\Xi=\Psi \circ \Phi'$. 
\subsection{Construction via even eights}
\label{ssec:sum}
On the level of Kummer surfaces isogenies between abelian
surfaces induce rational maps. Moreover, Mehran proved in
\cite{MR2804549} that there are fifteen distinct isomorphism classes
of rational double covers $\psi\colon X \dasharrow \mathrm{Kum}(A)$ of
the Kummer surface $\mathrm{Kum}(A)$ associated with the principal
polarized abelian surface $A =\mathrm{Jac}(C)$, such that the preimage
is a Kummer surface $X=\mathrm{Kum}(B)$ associated with an abelian
surface $B$ with the polarization of type $(1,2)$. Mehran computed
that the branching loci giving rise to these 15 distinct isomorphism
classes of double covers are even-eights of exceptional curves on the
Kummer surface $\mathrm{Kum}(A)$  \cite[Prop.~4.2]{MR2804549}: each
even-eight is itself enumerated by an order-two point $p_{i j} \in
A[2]$ with $1 \le i < j \le6$, and given as a sum in the
N\'eron-Severi lattice of the form 
\beq
\label{eqn:MehranEE}
 \Delta_{p_{ij}} = E_{1i} + \dots + \widehat{E_{ij}} + \dots + E_{i6} + E_{1j} + \dots + \widehat{E_{ij}} + \dots + E_{j6} \,,
 \eeq
where $E_{11}=0$, and $E_{i j}$ are the exceptional divisors obtained
by resolving the nodes $p_{i j}$; the hat indicates divisors that are
not part of the even eight. Moreover, Mehran proved that each rational
map $\psi\colon \mathrm{Kum}(B) \dashrightarrow \mathrm{Kum}(A)$
branched on such an even-eight $\Delta_{p_{ij}}$ is induced by an
isogeny $\Psi\colon  B \to A$ of abelian surfaces of  degree two and
vice versa~\cite{MR2804549}.   
\par We then push the entire construction of Section~\ref{ssec:KUM} down to level of projective Kummer surfaces: we construct $X'=\mathrm{Kum}(A)$ as the Shioda double-sextic surface~(\ref{kummer_middle}) and identify even-eights $\Delta_{p_{46}}$ and $\Delta_{p_{15}}$, corresponding to the marked G\"opel group $G$. As before, we start with $p_{46}$: we obtain a rational map $\psi\colon \mathrm{Kum}(B) \dashrightarrow \mathrm{Kum}(A)$ branched on $\Delta_{p_{46}}$ which is induced by the $(1,2)$-isogeny $\Psi\colon  B \to A$ determined by $\mathscr{L}$. The even-eights on $X'$ contain the following $(-2)$-curves:
\beq
\begin{split}
 \Delta_{p_{46}} &= \{ E_{14}, E_{24}, E_{34}, E_{45}, E_{16}, E_{26}, E_{36}, E_{56}\} \,,\\
 \Delta_{p_{15}} &= \{ E_{12}, E_{13}, E_{14}, E_{16}, E_{25}, E_{35}, E_{45}, E_{56}\} \,.\\
\end{split}
\eeq
The even-eight $\Delta_{p_{46}}$ contains 8 exceptional curves corresponding to the nodes in the intersection of $T_4$ or $T_6$ with $T_1, T_2, T_3, T_4$. Similarly, $\Delta_{p_{15}}$ contains 8 exceptional curves corresponding to the nodes in the intersection of $T_1$ or $T_5$ with $T_2, T_3, T_4, T_6$. By construction, $\Delta_{p_{46}}$ and $\Delta_{p_{15}}$ have four divisors in common:
\beq
 \Delta_{p_{46}} \cap \Delta_{p_{15}} = \{ E_{14}, E_{16}, E_{45}, E_{56} \} \,.
\eeq 
Thus, $\psi^*\Delta_{p_{15}}$ determines an even-eight on $X=\mathrm{Kum}(B)$. Repeating the construction, we obtain  a rational map $\phi' \colon \mathrm{Kum}(A') \dashrightarrow \mathrm{Kum}(B)$ induced  by the isogeny $\Phi'\colon  A' \to B$ determined by $\Psi^*(\mathscr{L}')$. 
\par In summary, we have the following isogenies and induced double coverings, relating the constructed Kummer surfaces:
\beq
\label{eqn:digram}
\begin{array}{ccccc}
 A'  = \mathrm{Jac}{(C')}  \cong A/G
 & \overset{\Phi'}{\longrightarrow} & 
 B
 & \overset{\Psi}{\longrightarrow} &
 A = \mathrm{Jac}{(C)} \\
 \downarrow && \downarrow && \downarrow \\
 X'' \cong  \mathrm{Kum}(A')  
  & \overset{\phi'}{\longrightarrow} & 
 X \cong  \mathrm{Kum}(B)  
  & \overset{\psi}{\longrightarrow} & 
 X' \cong  \mathrm{Kum}(A)  
\end{array} 
 \eeq
 It follows immediately:
\begin{lemma}
For general $A$ in Equation~(\ref{eqn:digram}) one has $\mathrm{NS}(X') \cong \mathrm{NS}(X'') \cong H \oplus D_8(-1) \oplus D_4(-1) \oplus A_3(-1)$ and  $\mathrm{NS}(X) \cong \langle 8 \rangle \oplus 2 D_8(-1)$.
\end{lemma}
In terms of a double-sextic surface we have the following:
 \begin{remark}
 \label{rem:mds}
For a Shioda double-sextic surface, fixing a marked G\"opel group $G \simeq \langle \mathscr{L} \rangle \oplus  \langle \mathscr{L}' \rangle$ can be viewed as follows:
we start with the equation
 \beq
  X'\colon \quad y^2 = \ell_1 \, \ell_2  \, \ell_3  \, \ell_4  \, \ell_5  \, \ell_6 \,,
\eeq 
where the six lines in $\mathbb{P}^2=\mathbb{P}(z_1, z_2, z_3)$ are tangent to a common conic. (i) We consider the six lines the projections of the odd tropes $T_i$ with $1 \le i \le 6$. (ii) We choose two pairs of lines, say $\{\ell_4, \ell_6 \}$ and $\{\ell_1, \ell_5 \}$.  (iii) We obtain even-eights $\Delta_{p_{46}}$ and  $\Delta_{p_{15}}$ and, from them, the rational double covers $\psi\colon X \dasharrow X'$ and $\phi'\colon X'' \dasharrow X$, respectively. 
\end{remark}
\subsection{Identification of coefficients as modular forms}
\label{ssec:modular}
We will now show how the triple $(X, \psi, \phi')$ yields the normal form in Equation~(\ref{eqn:K3_X}).
\par From a marked double-sextic surface as in Remark~\ref{rem:mds} one obtains a Weierstrass model of the form:
\beq
\label{eqn:Xp}
 X'\colon \quad y^2 z= x \big( x -  4 st  f(s, t) \, z \big) \big( x -  4 st  g(s, t) \, z\big) \,.
 \eeq
Here, we used linear transformations to move lines $\{\ell_4, \ell_6 \}$ to the central components of fibers in the Weierstrass model over $s=0$ and $t=0$, respectively, and the factors of $4$ are for convenience. Similarly, we used linear transformations to move  lines $\{\ell_1, \ell_5 \}$ to sections of the Weierstrass model given by $[x: y: z] =  [0: 0: 1]$ and $[x: y: z] =  [0: 1: 0]$, respectively. In other words, the zero-section $\sigma_X$ and 2-torsion section correspond to the lines $\ell_5$ and $\ell_1$, respectively. Moreover, the polynomials $f(s, t), g(s, t)$ have degree 2  and satisfy
\beq
 f(s, t) = f(t, s) \,, \qquad g(s, t) = g(t, s) \,,
\eeq 
because the lines $\ell_1, \dots, \ell_6$ are tangent to a common conic. The Weierstrass model~(\ref{eqn:Xp}) has  singular fibers $2 I_0^* + 6 I_2$ and torsion Mordell-Weil group $(\mathbb{Z}/2\mathbb{Z})^2$.  
\par The even-eight $\Delta_{p_{46}}$ consist of the eight
$(-2)$-curves that comprise the non-central components of the
$D_4$-fibers. The reason is that these are precisely the components of
the $D_4$-fibers that are met by the sections corresponding to lines
$\ell_1, \ell_2, \ell_3, \ell_5$. Similarly, the even-eight
$\Delta_{p_{15}}$ consist of the eight $(-2)$-curves that comprise the
4 components of $A_1$-fibers and the 4 non-central components of  two
$D_4$-fibers which are not met by the section corresponding to
$\ell_1, \ell_5$. The latter are the four components even-eights
$\Delta_{p_{46}}$ and $\Delta_{p_{15}}$ have in common.  
\par Double covers branched on the even-eights $\Delta_{p_{46}}$ and $\Delta_{p_{15}}$ are now easily constructed. For $\Delta_{p_{46}}$ the map $\psi=\psi_{\Delta_{p_{46}}} \colon X \dasharrow X'$ is given by
\beq
 \psi \colon  \Big( [u:v], \ [X:Y:Z] \Big)  \mapsto   \Big( [s:t], \ [x:y:z] \Big) =   \Big( [u^2:v^2], \ [u^2v^2X:u^3v^3Y:Z] \Big) .
\eeq
so that the preimage of $X'$ is the K3 surface $X= \mathrm{Kum}(B)$ with Weierstrass model
\beq
\label{eqn:X}
 X\colon \quad Y^2 Z= X \big( X -   4 f(u^2, v^2) \, Z \big) 	 \big( X -   4 g(u^2, v^2) \, Z\big) \,.
 \eeq
Similarly, the double cover $\psi_{\Delta_{p_{15}}}$ can be constructed explicitly: it is a fiberwise 2-isogeny for $X'$. However, as explained in Section~\ref{ssec:sum}, we rather consider the pullback $\phi'$ of $\psi_{\Delta_{p_{15}}}$ via $\psi_{\Delta_{p_{46}}}$, which is a map onto $X$. Then, the preimage of $X$ under $\phi'$ is the K3 surface $X'' = \mathrm{Kum}(A')$; see \cite{MR3995925}; it has the Weierstrass model
\beq
 X''\colon \quad y^2 z= x \Big( x^2 + 2 \big( f(u^2, v^2)  + g(u^2, v^2) \big) x z + \big(f(u^2, v^2) - g(u^2, v^2)\big)^2 z^2 \Big) \,.
\eeq
The following was proved in \cite{MR4444083}:
\begin{lemma}
Let $f, g$ be degree-2 polynomials such that $f(s, t)=f(t, s)$, $g(s, t)=g(t, s)$. Then the K3 surface with Weierstrass model~(\ref{eqn:X}) is isomorphic to the double-quadric surface over $\mathbb{F}_0 = \mathbb{P}(u,v) \times \mathbb{P}(x, z)$, given by
\beq
\label{eqn:double_quadric}
 y^2= \big(f(u^2, v^2) - g(u^2, v^2)\big) \big( x^4 + z^4 \big) + 2 \big( f(u^2, v^2)  + g(u^2, v^2) \big) x^2 z^2   .
\eeq 
\end{lemma}
Comparing Equations~(\ref{eqn:double_quadric}) and~(\ref{eqn:K3_X}),  we immediately find
\beq
\label{eqn:coeffs}
 f(u^2, v^2) = \frac{1}{4} \Big( b(u,v)  + 2 \, a(u, v)  \Big)\,, \qquad g(u^2, v^2) = \frac{1}{4} \Big( b(u,v)  - 2 \, a(u, v)  \Big) \,.
\eeq
We have shown that the normal form for $X$ in
Equation~(\ref{eqn:K3_X}) is equivalent to a particular choice on the
quotient surface $X'$, namely, the choice of two pairs of lines, among
the six lines that comprise the brach locus of the Shioda
double-sextic surface $X'$; see Remark~\ref{rem:mds}. In turn, the ordering of the pairs of
lines determines which ruling for $\mathbb{F}_0=\mathbb{P}(u,
v)\times\mathbb{P}(x, z)$ will be considered the projection of $X$ onto the
base curve. We will now use the connection with Abelian
surfaces, to establish the coefficients of
the normal form as modular forms. 
\par For a period matrix $\underline{\tau} \in \mathbb{H}_2$ we introduce the even genus-2 theta functions $\theta_i(\vec{z}, \underline{\tau})$, $i=1, \dots, 10$. We write for the \emph{theta constants}
\beq
\label{Eqn:theta_short}
 \theta_i \quad \text{instead of} \quad 
  \theta\!\begin{bmatrix} a^{(i)}_1 & a^{(i)}_2 \\ b^{(i)}_1 & b^{(i)}_2 \end{bmatrix}\!\!(\vec{0}, \underline{\tau})
 \quad \text{where $i=1,\dots ,10$,}
\eeq
with  
\beq
\begin{split}
\theta_1 	= \theta\!\begin{bmatrix} 0 &  0 \\ 0 & 0 \end{bmatrix}\!\!, \,  
\theta_2	= \theta\!\begin{bmatrix} 0 &  0 \\ 1 & 1 \end{bmatrix}\!\!, \,  
\theta_3 &	= \theta\!\begin{bmatrix} 0 &  0 \\ 1 & 0 \end{bmatrix}\!\!, \,  
\theta_4 	= \theta\!\begin{bmatrix} 0 &  0 \\ 0 & 1 \end{bmatrix}\!\!, \,   \;
\theta_5 	= \theta\!\begin{bmatrix} 1 &  0 \\ 0 & 0 \end{bmatrix}\!\!, \\
\theta_6	= \theta\!\begin{bmatrix} 1 &  0 \\ 0 & 1 \end{bmatrix}\!\!,  \, 
\theta_7	= \theta\!\begin{bmatrix} 0 &  1 \\ 0 & 0 \end{bmatrix}\!\!, \,  
\theta_8 &	= \theta\!\begin{bmatrix} 1 &  1 \\ 0 & 0 \end{bmatrix}\!\!, \,  
\theta_9 	= \theta\!\begin{bmatrix} 0 &  1 \\ 1 & 0 \end{bmatrix}\!\!, \,  
\theta_{10}=\theta\!\begin{bmatrix} 1 &  1 \\ 1 & 1 \end{bmatrix}\!\!.
\end{split}
\eeq
We also introduce the Siegel modular threefolds $\mathcal{A}_2(2n,4n)$ corresponding to the Igusa subgroups
\beq
 \Gamma_2(2n, 4n) = \lbrace M \in \Gamma_2(2n) | \, \mathrm{diag}(B) =  \mathrm{diag}(C) \equiv \mathbb{I} \mod{4n}\rbrace.
 \eeq
We have the following:
\begin{lemma}[\cite{MR4323344}]
The modular group $\Gamma_2(2)$ is the group of isomorphisms which fixes the 4th power of the theta constants $\theta_i$  for $1\le i \le 10$, $\Gamma_2(2,4)$ fixes their 2nd power, and $\Gamma_2(4,8)$ fixes the theta constants of level $(2,2)$.
\end{lemma}
Under duplication of the modular variable $\underline{\tau} \mapsto 2 \underline{\tau}$, the theta constants $\theta_1$, $\theta_5$, $\theta_7$, $\theta_8$ play a role dual to $\theta_1, \theta_2, \theta_3, \theta_4$. We renumber the former and use the symbol $\Theta$ to mark the fact that they are evaluated at an $(2, 2)$-isogenous abelian surface. That is, we will denote theta constants with doubled modular variable by 
\beq
\label{Eqn:Theta_short}
 \Theta_i  \quad \text{instead of} \quad 
  \theta\!\begin{bmatrix} b^{(i)}_1 & b^{(i)}_2 \\ a^{(i)}_1 & a^{(i)}_2 \end{bmatrix}\!\!(\vec{0}, 2\underline{ \tau})
 \quad \text{where $i=1,\dots ,10$.}
\eeq
The relevance of $\Theta_i$ for $i=1, \dots,4$ is seen as follows:
\begin{lemma}[\cite{MR4323344}]
The holomorphic map $\mathbb{H}_2 \to \mathbb{P}^3$, $\underline{\tau} \mapsto [\Theta_1 : \Theta_2: \Theta_3 : \Theta_4]$, induces an isomorphism
between the Satake compactification $\overline{\mathcal{A}_2(2,4)}$ and $\mathbb{P}^3$.
\end{lemma}
We can now prove the main result of this section, relating the coefficients in Equation~(\ref{eqn:K3_X}) to modular forms:
 \begin{theorem}
\label{thm:modular}
Assume that for $[ \rho: \kappa : \lambda : \mu ] \in  \mathbb{P}^3 \backslash \mathcal{B}$ in Equation~(\ref{eqn:K3_X}) one has $\mathrm{NS}(X) \cong \langle 8 \rangle \oplus 2 D_8(-1)$. There is a period matrix $\underline{\tau}$ with $[\underline{\tau}] \in \mathcal{A}_2(2,4)$ such that 
\beq
\label{eqn:params}
\begin{array}{rclcrcl}
 \rho & = & \Theta_1^2 \Theta_2^2 - \Theta_3^2 \Theta_4^2 , & \quad & \kappa & =&  - 2 \big(\Theta_1^2 \Theta_2^2 + \Theta_3^2 \Theta_4^2\big),\\[0.2em]
 \lambda & = &  \Theta_1^4 +  \Theta_2^4 - \Theta_3^4 - \Theta_4^4 , & \quad & \mu &=& - 2 \big(\Theta_1^4 +  \Theta_2^4 + \Theta_3^4 + \Theta_4^4\big).
\end{array}
\eeq
\end{theorem}
\begin{proof}

We introduce parameters $\Lambda_1, \Lambda_2, \Lambda_3$ such that
\beq
\label{eqn:moduli_as_Lambda}
 \kappa = - \Lambda_1, \quad \mu = \frac{2(2\Lambda_1-\Lambda_2 - \Lambda_3)}{\Lambda_2 - \Lambda_3}, \quad  \lambda = \frac{2(2\Lambda_2 \Lambda_3- \Lambda_1\Lambda_2 - \Lambda_1\Lambda_3)}{\Lambda_2 - \Lambda_3},
\eeq 
or, equivalently,
\beq
 \Lambda_1=-\kappa, \quad \Lambda_2 = - \frac{\lambda + 2\kappa}{\mu+2}, \quad  \Lambda_3 = - \frac{\lambda - 2\kappa}{\mu-2}.
\eeq
Here, we assume that the parameters $\Lambda_1, \Lambda_2, \Lambda_3$ are pairwise different. The Weierstrass form for the elliptic fibration in Equation~(\ref{eqn:B12_1b}) is equivalent to the following (affine) Weierstrass form:
\beq
\label{eqn:B12}
X\colon \quad y^2  =  \prod_{\substack{ \{i, j, k\} = \{ 1, 2, 3\},\\ j<k}} \left( x + \Lambda_i (u^4 +1) + \Lambda_j \Lambda_k u^2\right) \,.
\eeq
On $X$ in Equation~(\ref{eqn:B12}) the involution $k_X \circ \xi^{(i)}_X$ acts as $(u, y) \mapsto (-u, -y)$. The minimal resolution of the quotient is the Jacobian elliptic K3 surface given by
\beq
\label{eqn:K3_Xp}
\begin{aligned}
  X'\colon  \quad \tilde{y}^2  =  t \prod_{\substack{ \{i, j, k\} = \{ 1, 2, 3\},\\ j<k}} \left( x + \Lambda_i (t^2 +1) + \Lambda_j \Lambda_k t \right) \,.
\end{aligned}  
\eeq
 The rational quotient map $\psi \colon X \dashrightarrow X'$ is given by
\beqn
   (u, y, x) \mapsto   (t, \tilde{y}, x)  = (u^2, uy, x).
\eeqn  
In turn, the relation between Equation~(\ref{eqn:K3_Xp}) and~(\ref{kummer_middle}) is as follows: we introduce a square root $l$ such that $l^2=\lambda_1 \lambda_2 \lambda_3$ and assume that $(\Lambda_1, \Lambda_2, \Lambda_3)$ in Equation~(\ref{eqn:K3_Xp}) are related to $(\lambda_1, \lambda_2, \lambda_3)$  by
\beq
\label{eqn:Lambdas}
\Lambda_1 = \frac{\lambda_1 + \lambda_2\lambda_3}{l} \, \quad
\Lambda_2 = \frac{\lambda_2 + \lambda_1\lambda_3}{l} \, \quad
\Lambda_3 = \frac{\lambda_3 + \lambda_1\lambda_2}{l}.
\eeq
The substitution, given by
\beq
\begin{split}
x & = \frac{(\lambda_1^2 T- \lambda_1 X+1)\prod_{i=2,3} (\lambda_i-1) (\lambda_i T-1)}{l (T-X+1)} , \\
& - \frac{(\lambda_1+\lambda_2\lambda_3)(l^2 T^2+1)}{l} - \frac{(\lambda_2+\lambda_1\lambda_3)(\lambda_3+\lambda_1\lambda_2)T}{l},\\
\tilde{y} & = \frac{Y \prod_{i=1,2,3} (\lambda_i-1) (\lambda_i T-1)}{l (T-X+1)^2},\\
t & = l T,
\end{split}
\eeq
in Equation~(\ref{eqn:K3_Xp}) yields Equation~(\ref{kummer_middle}).  
\par The $\lambda$-parameters in the Rosenhain normal~(\ref{eqn:Rosenhain})  can be expressed as ratios of genus-2 theta constants by Picard's lemma. There are 720 choices for such expressions since the forgetful map $\mathcal{A}_2(2) \to \mathcal{A}_2$ is a Galois covering of degree $720 = |S_6|$ since $S_6$ acts on the Weierstrass points in~(\ref{eqn:Rosenhain}) by permutations. Any of the $720$ choices may be used. For example, the one used in \cite{MR2367218} is
\beq
\label{eqn:choice}
\lambda_1 = \frac{\theta_1^2\theta_3^2}{\theta_2^2\theta_4^2} \,, \quad \lambda_2 = \frac{\theta_3^2\theta_8^2}{\theta_4^2\theta_{10}^2}\,, \quad \lambda_3 =
\frac{\theta_1^2\theta_8^2}{\theta_2^2\theta_{10}^2}\,.
\eeq
Note that using Equation~(\ref{eqn:choice}) the square root $l$ can be identified with $l=\theta_1^2\theta_3^2\theta_8^2/(\theta_2^2\theta_4^2\theta_{10}^2)$ if we assume $[\underline{\tau}] \in \mathcal{A}_2(2,4)$; see \cite{MR4421430}. We then have
\beq
\label{eqn:Lambda}
\Lambda_1 = \frac{\theta_8^2}{\theta_{10}^2} + \frac{\theta_{10}^2}{\theta_8^2} \,, \quad  \Lambda_2 = \frac{\theta_1^2}{\theta_2^2} + \frac{\theta_2^2}{\theta_1^2} \,, \quad  \Lambda_3 = \frac{\theta_3^2}{\theta_4^2} + \frac{\theta_4^2}{\theta_3^2}  \,.
\eeq 
\par The following identities are called the \emph{second principal transformations of degree two}~\cite{MR141643, MR168805} for theta constants:
\beq
\label{Eq:degree2doubling}
\begin{array}{lllclll}
\theta_1^2 & = & \Theta_1^2 + \Theta_2^2 + \Theta_3^2 + \Theta_4^2 \,, &\qquad
\theta_2^2 & =&  \Theta_1^2 + \Theta_2^2 - \Theta_3^2 - \Theta_4^2 \,, \\[0.2em]
\theta_3^2 & = &  \Theta_1^2 - \Theta_2^2 - \Theta_3^2 + \Theta_4^2 \,, &\qquad
\theta_4^2 &= &  \Theta_1^2 - \Theta_2^2 + \Theta_3^2 - \Theta_4^2 \,.
\end{array}
\eeq
We also have the following identities:
\beq
\label{Eq:degree2doublingR}
\begin{array}{lllclll}
\theta_5^2 & = & 2 \, \big( \Theta_1 \Theta_3 + \Theta_2  \Theta_4 \big) \,, &\qquad
\theta_6^2 & =&  2 \, \big( \Theta_1 \Theta_3 - \Theta_2  \Theta_4 \big)\,, \\[0.4em]
\theta_7^2 & = & 2 \, \big( \Theta_1 \Theta_4 + \Theta_2  \Theta_3 \big) \,, &\qquad
\theta_8^2 &= &  2 \, \big( \Theta_1 \Theta_2 + \Theta_3  \Theta_4 \big)\,, \\[0.4em]
\theta_9^2 & = & 2 \, \big( \Theta_1 \Theta_4 - \Theta_2  \Theta_3 \big)\,, &\qquad
\theta_{10}^2 &= & 2 \, \big( \Theta_1 \Theta_2 - \Theta_3  \Theta_4 \big) \,.
\end{array}
\eeq
 Plugging Equation~(\ref{eqn:Lambda}) into Equation~(\ref{eqn:moduli_as_Lambda}) and using
Equations~(\ref{Eq:degree2doubling}) and~(\ref{Eq:degree2doublingR})
yields the result. 
\end{proof}
\subsection{Construction of the polarizing and invariant lattice}
\label{sec:lattice}
We have established that the general $X$ is a Kummer surface associated with an abelian surface of the polarization of type $(1,2)$. As established above,  the polarizing lattice of such Kummer surfaces is $L = \langle 8 \rangle \oplus 2 D_8(-1)$ and the general transcendental lattice is $H(2) \oplus H(2) \oplus \langle -8 \rangle$. Let us also denote the sixteen $(-2)$-classes on the Kummer surface by $K_j$ with $0 \le j \le 15$. 

\par We choose one of the sections in Lemma~\ref{lem:fibrations} as zero section $\mathsf{O}$ and denote  its divisor class by $[\mathsf{O}]=K_0$.  As observed before, Equation~(\ref{eqn:B12_1b}) is then transformed into Equation~(\ref{eqn:EFS_0}), which is equivalent to  the following Weierstrass model:
\begin{equation}\label{eqn:EFS}
X\colon \quad Y^2 Z  = X \Big( X^2 - 2 b(u, v) \, XZ + \big(b(u)^2 - 4 a(u, v)^2\big) Z^2 \Big) \,,
\end{equation}
This is precisely Equation~(\ref{eqn:X}) when using Equation~(\ref{eqn:X}). It  determines the Jacobian elliptic fibration  $\pi_X \colon X \to \mathbb{P}^1 =\mathbb{P}(u, v) $ with general fiber in $\mathbb{P}^2 = \mathbb{P}(X, Y, Z)$. 

\par We will consider the fibration in the affine chart $Z=1$, $v=1$. Then,  $a(u)$ and $b(u)$ are the even polynomials of degree four in $u$ given in Equation~(\ref{eqn:B12_1b}), such that there are no singular fibers over $u=0, \infty$, and
\beqn
a(u) = u^4 a(1/u)\,, \quad b(u) = u^4 b(1/u)\,.
\eeqn
The discriminant of the elliptic fibration is $\Delta=16 a(u)^2 \big(b(u)^2-4a(u)^2\big)^2$ and has twelve roots of order two. That is, the fibration has twelve singular fibers of Kodaira type $I_2$ and no other singular fibers. The generic fiber $\mathsf{F}_u =\pi_X^{-1}(u)$ is smooth. The torsion Mordell Weil group is $ \mathrm{MW}(X)_{\mathrm{tor}} = (\mathbb{Z}/2\mathbb{Z})^2$ and $\operatorname{rank} \mathrm{MW}(X)=3$.  Garbagnati \cite{MR2600955, MR3010125} proved that the smooth fiber class $F = [ \mathsf{F}_u ]$ with $F^2=0$ and $F \circ K_0$=1 is given by
\beqn
F = \frac{H-K_0 - K_1 - K_2 - K_3}{2} \,,
\eeqn
and $K_j$ for $4 \le j \le 15$ are realized as the non-neutral components of the reducible fibers of type $A_1$ in the elliptic fibration in~(\ref{eqn:EFS}), and $H$ is a hyperplane class.
\par Moreover, the elliptic fibration is invariant under the action of the hyperelliptic involution $(u,X,Y) \mapsto (u,X,-Y)$ -- which we denote by $p \mapsto -p$ for a point $p \in F$ in a fiber $F$ and three additional involutions. Using Lemma~\ref{lem:holinv1} we check that these involutions are  given by
\beq
\begin{split}\label{eqn:involutions}
\jmath_1:& \quad (u,X,Y) \mapsto \Big(u' = -u,X,Y\Big)\,,\\
\jmath_2:& \quad (u,X,Y) \mapsto \Big(u''= \frac1u,\frac{X}{u^4},\frac{Y}{u^6}\Big) \,,\\
\jmath_3:& \quad (u,X,Y) \mapsto \Big(u''' = -\frac1u,\frac{X}{u^4},-\frac{Y}{u^6}\Big) \,.
\end{split}
\eeq
The involutions $u \mapsto -u$ and $u \mapsto 1/u$ and their composition map singular fibers of Equation~(\ref{eqn:EFS}) to singular fibers, and smooth fibers to smooth fibers. The zero section $\mathsf{O}$, given as the point at infinity in each fiber, and the 2-torsion sections $\mathsf{T}_1, \mathsf{T}_2, \mathsf{T}_3$, given by
\beq
\label{eqn:zero_sections}
\mathsf{T}_1: (X,Y)=(0,0)\,, \quad \mathsf{T}_2: (X,Y)=(b-2a,0)\,, \quad \mathsf{T}_3: (X,Y)=(b+2a,0) \,,
\eeq
are invariant under the involutions $\jmath_1, \jmath_2, \jmath_3$, and the hyperelliptic involution. 
\par  Each 2-torsion section intersects the non-neutral components $K_j$ of eight reducible fibers of type $A_1$ -- which we represent as sets $W_k = \{K_i \mid i \in I_k\}$ for index sets $I_k$ such that $|I_k|=8$ for $k= 1, 2, 3$ -- partitioning the twelve rational curves $K_j$ with $4 \le j \le 15$ into three sets of eight curves with pairwise intersections consisting of four curves, i.e., $|W_j \cap W_k|=4$ and $W_1 \cap W_2 \cap W_3=\emptyset$. None of the twelve reducible fibers are invariant under the action of the involutions $\jmath_1, \jmath_2$.  However, the sets $W_k$ and $W_j \cap W_k$  for $1\le j,k \le 3$ are invariant under $\jmath_1, \jmath_2$.  We may define divisors $\bar{K}_{W_k} = \frac{1}{2} \sum_{n \in I_k} K_n$ with $1 \le k \le 3$, which are known to be elements of the Kummer lattice \cite{MR2600955, MR3010125}, with $\bar{K}_{W_j}\circ \bar{K}_{W_k}=-2-2\delta_{jk}$ for $1\le j,k \le 3$. We also define divisors $\bar{K}_{W_j \cap W_k} = \frac{1}{2} \sum_{n \in I_j \cap I_k} K_n$ with $\bar{K}_{W_j \cap W_k}^2=-2$. By construction, the elements $\bar{K}_{W_k}$ and $\bar{K}_{W_j \cap W_k}$ for $1 \le j, k \le 3$ are invariant under the action of the involutions $\jmath_1, \jmath_2$. 
\par The singular fibers of the fibration~(\ref{eqn:EFS}) arise where the 2-torsion sections collide. This happens as follows:
\begin{center}
\begin{tabular}{crcc}
colliding sections & equation & \# of points & fiber components\\
\hline
$\mathsf{T}_1=\mathsf{T}_2$ & $b-2a=0$ & $4$ & $W_1 \cap W_2$ \\
$\mathsf{T}_1=\mathsf{T}_3$ & $b+2a=0$ & $4$ & $W_1 \cap W_3$ \\
$\mathsf{T}_2=\mathsf{T}_3$ & $a=0$ & $4$ & $W_2 \cap W_3$
\end{tabular}
\end{center}
\par In \cite{MR4484238} non-torsion sections $\mathsf{S}_1, \mathsf{S}_2, \mathsf{S}_3$ for the elliptic fibration $(\pi,\mathsf{O})$ were constructed which are orthogonal with respect to height pairing.  For two arbitrary sections $S'$ and $S''$ of the elliptic fibration, one defines the \emph{height pairing} using the formula
\beq
\label{eqn:height_pairing}
\langle S', S'' \rangle = \chi(\mathcal{O}_X) + \mathsf{O}\circ S' + \mathsf{O}\circ S'' - S' \circ S'' - \sum_{\{s|\Delta=0\}} C_s^{-1}(S',S'') \,,
\eeq
where the holomorphic Euler characteristic is $\chi(\mathcal{O}_X) =2$, and the inverse Cartan matrix $C_s^{-1}$ of a fibre of type $A_1$ located over point $s$ of the discriminant locus $\Delta=0$ contributes $(\frac{1}{2})$ if and only if both $S'$ and $S''$ intersect the non-neutral component. 
\par The sections $\mathsf{S}_1$ and $\mathsf{S}_2$ do not intersect the zero section $\mathsf{O}$ and intersect the non-neutral components of six reducible fibers of type $A_1$ each -- which we represent as complementary sets $V_k = \{K_i \mid i \in J_k\}$ for index sets $J_k$ such that $|V_k|=6$ for $k= 1, 2$ -- partitioning the twelve rational curves $K_j$ with $4 \le j \le 15$ into two disjoint sets of six curves.  The sets $V_1$ and $V_2$ are invariant under the action of the involution $\jmath_1$ and interchanged under the action of $\jmath_2$, and we have $|V_j \cap W_k|=4$ for $j= 1, 2$ and $k=1, 2, 3$.  The section $\mathsf{S}_3$ intersects the non-neutral components of all reducible fibers, and the zero section such that $\mathsf{S}_3 \circ \mathsf{O} =2$.
The sections $\{\mathsf{O}, \mathsf{T}_1, \mathsf{T}_2, \mathsf{T}_3, \mathsf{S}_1, \mathsf{S}_2, \mathsf{S}_3\}$ are linearly independent sections in the Mordell-Weil group, but they do \emph{not} form a basis. Given the explicit form of the sections $\{\mathsf{O}, \mathsf{T}_1, \mathsf{T}_2, \mathsf{T}_3, \mathsf{S}_1, \mathsf{S}_2, \mathsf{S}_3\}$, we compute the intersection pairings for their divisor classes and their corresponding height pairings; the results are shown in Table~\ref{tab:intersection}.  We have the following:
\begin{lemma}
\label{lem:antisymplectic}
The involutions $\jmath_1, \jmath_2, \jmath_3$ are three commuting anti-symplectic involutions of the elliptic fibration with section $(\pi,\mathsf{O})$ with $\jmath_3=-\jmath_1\jmath_2$. The  involutions $\jmath_l$ for $1\le l \le3$ act on the sections $\{\mathsf{O}, \mathsf{T}_1, \mathsf{T}_2, \mathsf{T}_3, \mathsf{S}_1, \mathsf{S}_2, \mathsf{S}_3\}$ as follows:
\begin{center}
\begin{tabular}{c|ccccrrr}
                    & $\mathsf{O}$	& $\mathsf{T}_1$	& $\mathsf{T}_2$	& $\mathsf{T}_3$	& $\mathsf{S}_1$	& $\mathsf{S}_2$	& $\mathsf{S}_3$\\
\hline
$\jmath_1$	& $\mathsf{O}$	& $\mathsf{T}_1$	& $\mathsf{T}_2$	& $\mathsf{T}_3$	& $\mathsf{S}_1$	& $\mathsf{S}_2$	& $-\mathsf{S}_3$\\
$\jmath_2$	& $\mathsf{O}$	& $\mathsf{T}_1$	& $\mathsf{T}_2$	& $\mathsf{T}_3$	& $\mathsf{S}_2$	& $\mathsf{S}_1$	& $\mathsf{S}_3$\\
$\jmath_3$	& $\mathsf{O}$	& $\mathsf{T}_1$	& $\mathsf{T}_2$	& $\mathsf{T}_3$	& $-\mathsf{S}_2$	& $-\mathsf{S}_1$	& $\mathsf{S}_3$
\end{tabular}
\end{center}
\end{lemma}
\par We can now  complete the description of the  sixteen $(-2)$-curves in terms of the elliptic fibration: the elliptic fibration $(\pi,\mathsf{O})$ in Equation~(\ref{eqn:EFS}) admits sections $\{\mathsf{S}'_1 ,\mathsf{S}'_2 , \mathsf{S}'_3\}$  such that the remaining divisor classes $K_1, K_2, K_3$ are represented as follows:
\begin{equation}\label{eqn:KummerCurves}
K_0 = [\mathsf{O}]\, \quad K_1 = [ \mathsf{S}'_1] \,,\quad K_2 = [\mathsf{S}'_2] \,, \quad K_3 = [\mathsf{S}_3'] \,.
\end{equation}
In turn, the corresponding sections are obtained as linear combinations of the three non-torsion sections $\mathsf{S}_1, \mathsf{S}_2, \mathsf{S}_3$, using the elliptic-curve group law in $F$ as follows:
\beq
\mathsf{S}'_1 = 2\mathsf{S}_1, \quad \mathsf{S}'_2 = \mathsf{S}_1+\mathsf{S}_2+\mathsf{S}_3, \quad \mathsf{S}'_3 = \mathsf{S}_1-\mathsf{S}_2+\mathsf{S}_3 .
\eeq
This follows from Table~\ref{tab:intersection}.  All other possible choices are obtained by the action of the hyperelliptic involution and the involutions $\jmath_k$ for $k=1, 2, 3$.

\begin{table}
\parbox{.45\linewidth}{
\scalemath{0.6}{
\begin{tabular}{c||r|r|r|r|r|r|r|r|r|r|r|l}
$\circ$ 	& $F$ 	& $\mathsf{O}$	& $\mathsf{T}_1$	& $\mathsf{T}_2$	& $\mathsf{T}_3$	& $\mathsf{S}'_1$	& $\mathsf{S}'_2$ 	& $\mathsf{S}'_3$ 	& $\mathsf{S}_1$	& $\mathsf{S}_2$	& $\mathsf{S}_3$ & $\mathsf{S}'_2/2$\\
\hline\hline
$F$		& 0		& $1$	& $1$	& $1$	& $1$	& $1$	& $1$	& $1$	& $1$	& $1$	& $1$ 	& $1$\\
$\mathsf{O}$	& $1$	& $-2$	& $0$	& $0$	& $0$	& $0$	& $0$	& $0$	& $0$	& $0$	& $2$	& $0$\\
$\mathsf{T}_1$	& $1$	& $0$	& $-2$	& $0$	& $0$	& $2$	& $2$	& $2$	& $0$	& $0$	& $0$	&$0$\\
$\mathsf{T}_2$	& $1$	& $0$	& $0$	& $-2$	& $0$	& $2$	& $2$	& $2$	& $0$	& $0$	& $0$	& $0$\\
$\mathsf{T}_3$	& $1$	& $0$	& $0$	& $0$	& $-2$	& $2$	& $2$	& $2$	& $0$	& $0$	& $0$	& $0$\\
$\mathsf{S}'_1$	& $1$	& $0$	& $2$	& $2$	& $2$	& $-2$	& $0$	& $0$	& $0$	& $2$	& $4$	& $1$\\
$\mathsf{S}'_2$	& $1$	& $0$	& $2$	& $2$	& $2$	& $0$	& $-2$	& $0$	& $1$	& $1$	& $2$	& $0$\\
$\mathsf{S}'_3$	& $1$	& $0$	& $2$	& $2$	& $2$	& $0$	& $0$	& $-2$	& $1$	& $3$	& $2$	& $1$\\
$\mathsf{S}_1$	& $1$	& $0$	& $0$	& $0$	& $0$	& $0$	& $1$	& $1$	& $-2$	& $2$	& $1$	& $0$\\
$\mathsf{S}_2$	& $1$	& $0$	& $0$	& $0$	& $0$	& $2$	& $1$	& $3$	& $2$	& $-2$	& $1$	& $0$\\
$\mathsf{S}_3$	& $1$	& $2$	& $0$	& $0$	& $0$	& $4$	& $2$	& $2$	& $1$	& $1$	& $-2$ 	& $0$\\
$\mathsf{S}'_2/2$ & $1$ 	& $0$	& $0$	& $0$	& $0$	& $1$	& $0$	& $1$	& $0$	& $0$	& $0$	& $-2$
\end{tabular}}}
\qquad
\parbox{.45\linewidth}{
\scalemath{0.6}{
\begin{tabular}{c||r|r|r|r|r|r|r|r|r|r|l}
$\langle\bullet,\bullet\rangle$ 	& $\mathsf{O}$	& $\mathsf{T}_1$	& $\mathsf{T}_2$	& $\mathsf{T}_3$	& $\mathsf{S}'_1$	& $\mathsf{S}'_2$ 	& $\mathsf{S}'_3$	& $\mathsf{S}_1$	& $\mathsf{S}_2$	& $\mathsf{S}_3$ & $\mathsf{S}'_2/2$\\
\hline\hline
$\mathsf{O}$	& $0$	& $0$ 	& $0$	& $0$ 	& $0$	& $0$	& $0$	& $0$ 	& $0$	& $0$	& $0$\\
$\mathsf{T}_1$	& $0$	& $0$ 	& $0$	& $0$	& $0$	& $0$	& $0$	& $0$ 	& $0$	& $0$	& $0$\\
$\mathsf{T}_2$	& $0$	& $0$ 	& $0$	& $0$ 	& $0$	& $0$	& $0$	& $0$ 	& $0$	& $0$	& $0$\\
$\mathsf{T}_3$	& $0$	& $0$ 	& $0$	& $0$ 	& $0$	& $0$	& $0$	& $0$ 	& $0$	& $0$	& $0$\\
$\mathsf{S}'_1$	& $0$	& $0$ 	& $0$	& $0$	& $4$	& $2$	& $2$	& $2$ 	& $0$	& $0$	& $1$\\
$\mathsf{S}'_2$	& $0$	& $0$ 	& $0$	& $0$ 	& $2$	& $4$	& $2$	& $1$ 	& $1$	& $2$	& $2$\\
$\mathsf{S}'_3$	& $0$	& $0$ 	& $0$	& $0$ 	& $2$	& $2$	& $4$	& $1$ 	& $-1$	& $2$	& $1$\\
$\mathsf{S}_1$	& $0$ 	& $0$	& $0$	& $0$	& $2$ 	& $1$	& $1$	& $1$	& $0$	& $0$	& $\frac{1}{2}$\\
$\mathsf{S}_2$	& $0$ 	& $0$	& $0$	& $0$	& $0$ 	& $1$	& $-1$	& $0$	& $1$	& $0$	& $\frac{1}{2}$\\
$\mathsf{S}_3$	& $0$ 	& $0$	& $0$	& $0$	& $0$ 	& $2$	& $2$	& $0$	& $0$	& $2$	& $1$\\
$\mathsf{S}'_2/2$	& $0$	& $0$ 	& $0$	& $0$ 	& $1$	& $2$	& $1$	& $\frac{1}{2}$ 	& $\frac{1}{2}$	& $1$ & $1$
\end{tabular}}}
\caption{Intersection and Height Pairings}
\label{tab:height}\label{tab:intersection}
\end{table} 
\par The divisor classes for the 2-torsion sections $\mathsf{T}_1$, $\mathsf{T}_2$, $\mathsf{T}_3$ are expressible in terms of the fiber class $F$, the class of the zero section $K_0$, the classes of the non-neutral components of the reducible fibers $K_k$ for $4 \le k \le 15$, and the classes of the sections $\mathsf{S}_1$, $\mathsf{S}_2$, $\mathsf{S}_3$. That is, for $1\le j \le 3$ we make the ansatz
\beq
\begin{split}
[\mathsf{T}_j] & = \beta_j F + \alpha_{j0} K_0 + \sum_{l=1}^{3} \alpha_{jl}   [\mathsf{S}'_l] + \sum_{k=4}^{15} \alpha_{jk} K_k ,
\end{split}
\eeq
with constants $\alpha, \beta, \gamma \in \mathbb{Q}$. We then check the following:
\begin{corollary}
\label{cor:torsion_sections}
The divisor classes of the 2-torsion sections $\mathsf{T}_k$ are given by
\beq
\label{eqn:secT}
[ \mathsf{T}_k ] = 2 F + K_0 - \bar{K}_{W_k} \quad \text{for  $k = 1,  2, 3$.}
\eeq
\end{corollary}
\begin{proof}
The proof follows from $[ \mathsf{T}_k ] \circ F =1$, $[ \mathsf{T}_k ] \circ K_0 =0$, $[ \mathsf{T}_k ] \circ K_j =1$ for $j \in I_k$ and $[ \mathsf{T}_k ] \circ K_j =0$ for $j \not \in I_k$. A computation then yields the result.
\end{proof}
We denote the trivial lattice generated by all fibre components and the zero section by $\operatorname{Triv}(X)$ and the Mordell-Weil lattice by $\operatorname{MWL}(X)$, i.e.,  $\operatorname{MW}(X)/\operatorname{MW}(X)_{\mathrm{tor}}$ equipped with the height pairing in Equation~(\ref{eqn:height_pairing}). Then, we have the following formula relating the determinants of their discriminant groups:
\beq
\label{eqn:determinants}
\big| \operatorname{disc}\big( \mathrm{NS}(X) \big) \big| =  \big| \operatorname{disc}\big( \operatorname{Triv}(X) \big) \big|  \cdot  \big| \operatorname{disc}\big( \operatorname{MWL}(X)\big) \big| / \big( \# \operatorname{MW}(X)_{\mathrm{tor}} \big)^2.
\eeq
In our situation, this equality requires $\operatorname{MWL}(X)$ to have rational (non-integral) discriminant: the elliptic fibration has singular fibers $12 I_2$ and torsion Mordell-Weil group $(\mathbb{Z}/2\mathbb{Z})^2$; thus, any basis for the three non-torsion sections must form a height matrix with determinant $\frac{1}{2}$. In our situation, it is impossible to have a section of nonintegral height\footnote{We thank Adam Logan for explaining this point to us}. However, one checks that $\mathsf{S}'_2 = \mathsf{S}_1 + \mathsf{S}_2 + \mathsf{S}_3$ is divisible by 2, and one obtains a new section $\mathsf{S}'_1/2$: it has height $1$ and pairing 1/2, 1/2, 1 with $\mathsf{S}_1, \mathsf{S}_2, \mathsf{S}_3$, respectively.  The section $\mathsf{S}'_1/2$ intersects six reducible fibers of type $A_1$ -- which we represent as set $V_4 = \{K_i \mid i \in J_4\}$ for index sets $J_4$ such that $|V_4|=6$, $| V_1 \cap V_4| = | V_2 \cap V_4|  =3$ and $| W_i \cap V_4| =4$ for $i= 1, 2, 3$. A minimal set of generators for the Mordell-Weil lattice $\operatorname{MWL}(X)$ is given by $\{\mathsf{S}_1, \mathsf{S}_2, \mathsf{S}'_2/2\}$ and the height matrix is $[1,0,1/2;0,1,1/2;1/2,1/2,1]$ whose determinant equals $\frac{1}{2}$. This is in agreement with Equation~(\ref{eqn:determinants}) since the  determinant of the polarizing lattice is $2^{12}\cdot \frac{1}{2}/4^2=2^7$. 
\par We have the following result that establishes the polarizing lattice of the K3 surface $X$ in terms of the Jacobian elliptic fibration $\pi_X \colon X \to \mathbb{P}^1 =\mathbb{P}(u, v)$:
\begin{proposition}
Given the Jacobian elliptic fibration $\pi_X \colon X \to \mathbb{P}^1 =\mathbb{P}(u, v)$ in Equation~(\ref{eqn:EFS}), the N\'eron-Severi lattice $\mathrm{NS}(X)$ is the overlattice spanned by (i) the lattice generated by $F,  K_0$, the non-neutral components of the reducible fibers of type $A_1$, the sections $\mathsf{S}_1, \mathsf{S}_2, \mathsf{S}'_2/2$, and (ii) the 2-torsion sections $\mathsf{T}_1,  \mathsf{T}_2,  \mathsf{T}_3$.
\end{proposition}
\begin{proof}
We observe that $\langle 8 \rangle \oplus 2 D_8(-1) \cong H \oplus 2 D_4(-1) \oplus A_7(-1)$. Thus, it is enough to prove that a definite lattice isomorphic to $2 D_4 \oplus A_7$ is the overlattice spanned by (i) the lattice generated by $K_j$ for $j=4, \dots, 15$, the classes of $2F + K_0 - [\mathsf{S}_1],2F + K_0  -  [\mathsf{S}_2], 2F + K_0 - [\mathsf{S}'_2/2]$ and (ii) the classes $\bar{K}_{W_1},  \bar{K}_{W_2}, \bar{K}_{W_3}$ in Equation~(\ref{eqn:secT}). A computation of the lattice-theoretic genus shows that the two lattices are in fact isomorphic.
\end{proof}
We also make the following:
\begin{remark}
\label{rem:eigenspaces}
The involutions $\jmath_1$ and $\jmath_2$ act on the $\mathbb{Q}$-basis $\{[\mathsf{S}_1], [\mathsf{S}_2], [\mathsf{S}'_2/2]\}$ as the following matrices:
\beq
M(\jmath_1) = \left( \begin{array}{ccr} 1 & & 1\\ & 1 & 1\\  &  & -1 \end{array}\right), \qquad M(\jmath_2) = \left( \begin{array}{ccc} & 1 & \\ 1 &  & \\  &  & 1 \end{array}\right).
\eeq
We have the following bases for the eigenspaces for the eigenvalues $+1$ and $-1$, respectively:
\beq
\left\lbrace  \langle 1, 0 , 0\rangle , \langle 0, 1, 0 \rangle;  \left\langle -\frac{1}{2}, -\frac{1}{2} , 1\right\rangle \right\rbrace, \qquad
\Big\lbrace \langle 1, 1 , 0\rangle , \langle 0, 0, 1 \rangle; \langle -1, 1 , 0\rangle   \Big\rbrace.
\eeq
\end{remark}
\par Let $R$ denote a rational elliptic surface with section.  We denote the classes of the zero-section on $R$ by $o$ and the smooth fiber of $R$ by $f$, respectively. The N\'eron-Severi lattice is given by $\mathrm{NS}(R) = \langle f, o \rangle \oplus E_8(-1)$ where the intersection pairing for  $\langle f, o \rangle$ equals $[0, 1;1, - \chi(\mathcal{O}_R)=-1]$. We  consider a morphism $\phi\colon \mathbb{P}(u, v) \to \mathbb{P}(s, t)$ of degree two. The pull-back $X$ of $R$ via $\phi$ defines a K3 surface if the fibers of $R$ at the ramification points are reduced.  In this case, pulling back $\mathrm{NS}(R)$ via $\phi^*$ we see that $\langle 2 \rangle \oplus \langle -2 \rangle \oplus E_8(-2)$ embeds into $\mathrm{NS}(X)$. The former lattice has  the invariants $(r, a, \delta) = (10, 10, 1)$ where $r$ is its rank, the determinant of its discriminant group is $2^a$, and $\delta$ is the parity; see \cite{MR2137825} for more details on the invariants $(r, a,\delta)$.
\par In our situation, we start with the following Weierstrass models
\beq
\label{eqn:S} 
\begin{split}
R_1\colon & \quad y^2 z= x \big( x -   4 f(s, t) z \big) 	 \big( x -   4 g(s, t) z\big) ,\\
R_2\colon & \quad y^2 z= x \big( x -   4 f'(s, t) z \big) 	 \big( x -   4 g'(s, t) z\big) ,\\
R_3\colon & \quad y^2 z= x \big( x -   4 f''(s, t) z \big) 	 \big( x -   4 g''(s, t) z\big) ,
\end{split}
\eeq
where the polynomials are given by
\beq
\label{eqn:S_coeffs} 
\begin{split}
f& = f_2 ( s^2 + t^2) + f_1 st, \qquad g  = g_2 ( s^2 + t^2) + g_1 st,\\
f'& =  f_2 s^2 + (f_1 - 2 f_2)t^2, \quad g'  =  g_2 s^2 + (g_1 - 2 g_2)t^2,\\
f''& =  f_2 s^2 + (f_1 + 2 f_2)t^2, \quad g''  =  g_2 s^2 + (g_1 + 2 g_2)t^2. 
\end{split}
\eeq
In the Oguiso-Shioda classification \cite{MR1104782}, all of these are rational elliptic surfaces  of type \#42 with singular fibers $6 I_2$, $\mathrm{MWL} = 2 A_1^{\vee}$, and torsion Mordell-Weil group $(\mathbb{Z}/2\mathbb{Z})^2$.  Then using the following degree-2 maps
\beq
\label{eqn:S_maps} 
\begin{split}
\phi_1 \colon& \quad [u : v] \ \mapsto \ [ s : t] = [ u^2: v^2 ],\\
\phi_2 \colon& \quad [u : v] \ \mapsto \ [ s : t] = [ u^2 + v^2 : uv ],\\
\phi_3 \colon& \quad [u : v] \ \mapsto \ [ s : t] = [ u^2 - v^2 : uv ], 
\end{split}
\eeq
the pull-back of $R_k$ via $\phi_k$ for $k=1, 2, 3$, all yield the same result, namely, the K3 surface $X$ in Equation~(\ref{eqn:EFS}) where we used Equations~(\ref{eqn:coeffs}). We have the following:
\begin{lemma}
\label{lem:inv}
The rank-10 lattices $\phi_k^* \mathrm{NS}(R_k)$ for $k=1, 2, 3$ embed into $\mathrm{NS}(X)$; they are invariant under $\jmath_k$ with  $\phi_k^* \mathrm{NS}(R_k) \cong \phi_j^* \mathrm{NS}(R_j)$ for $j,k  =1, 2, 3$.
\end{lemma}
We call the lattice in Lemma~\ref{lem:inv} the invariant lattice $L^{\jmath}$ of the polarizing lattice $L=\langle 8 \rangle \oplus 2 D_8(-1)$. We have the following:
\begin{proposition}
\label{prop:invlattice}
The invariant lattice $L^{\jmath}$ is  $\langle 2 \rangle \oplus \langle -2 \rangle \oplus E_8(-2)$ with invariants $(r, a, \delta) = (10, 10, 1)$.
\end{proposition}
\begin{proof}
We can compute the invariant lattice $L^{\jmath}$ as the sublattice of the polarizing lattice that is left invariant by the involution $\jmath =\jmath_1$. We denote the classes of the non-neutral components of the reducible fibers for Equation~(\ref{eqn:EFS}) by $K_k$ for $4 \le k \le 15$, so that $\jmath$ interchanges $K_{2j}$ and $K_{2j+1}$ for $j=2, \dots, 7$. The lattice $L^{\jmath}$ is the overlattice spanned by (i) the lattice generated by the classes of $2F \cong [\mathsf{F}_u] +  [\mathsf{F}_{\jmath(u)}] ,  K_0$, the classes  $K_{2j}+ K_{2j+1}$ for $j=2, \dots, 7$, the classes $[\mathsf{S}_1], [\mathsf{S}_2]$, and (ii) the classes of the 2-torsion sections $\mathsf{T}_1,  \mathsf{T}_2,  \mathsf{T}_3$. This follows from Lemma~\ref{lem:antisymplectic} and Remark~\ref{rem:eigenspaces}.  This shows that $r=10$. Moreover, the intersection pairing for  $\langle 2 F, K_0 \rangle$ equals $[0, 2;2, - \chi(\mathcal{O}_X)=-2] \cong \langle 2 \rangle \oplus \langle -2 \rangle$.  This shows $\delta=1$.  The classes  $K_{2j}+ K_{2j+1}$ for $j=2, \dots, 7$ and $[\mathsf{S}_1], [\mathsf{S}_2]$ generate a lattice $S'$. We write $L^{\jmath} = \langle 2 \rangle \oplus \langle -2 \rangle \oplus S(-1)$ where $S$ is an overlattice of $S'$. We use Equation~(\ref{eqn:secT}) to find the divisor classes for  $2F + K_0 - [\mathsf{T}_k]$. The classes $\bar{K}_{W_k}$ descend to vectors $\bar{K}'_{W_k}$ over the basis $\{ K_{2j}+ K_{2j+1} \}_{j=2}^7$ with coefficients in $\mathbb{Z}/2\mathbb{Z}$.  Computing the overlattice generated by $S'$ and $\bar{K}'_{W_1}, \bar{K}'_{W_2}, \bar{K}'_{W_3}$ we obtain $S$. Computing the lattice-theoretic genus of $S$ proves $S \cong E_8(2)$.
\end{proof}
\begin{remark}
We have obtained rational projection maps
$\phi_k\colon X \dasharrow X/\langle \jmath_k \rangle \cong R_k$ so
that $\phi_k^* \mathrm{NS}(R_k)$ are invariant sublattices of
$\mathrm{NS}(X)$ under the action of $\jmath_k$. In particular, this
applies to the pull-back of their Mordell-Weill lattices
$\phi_k^* \mathrm{MWL}(R_k)$.  The composition of $\jmath_k$ with the
hyperelliptic involution $(u,X,Y) \mapsto (u,X,-Y)$ yields three
(symplectic) Nikulin involutions with rational projection maps
$\psi_k\colon X \dasharrow X/\langle (-1) \circ \jmath_k \rangle \cong X'_k$.
We saw in Section~\ref{ssec:sum} that these quotients are Jacobian Kummer
surfaces; see Equation~(\ref{eqn:Xp}). Their Mordell-Weil lattices (with
respect to the induced elliptic fibration) have rank one; in fact they are
$\mathrm{MWL}(X'_k) = \langle 1 \rangle$. Their pull-backs
$\psi_k^* \mathrm{MWL}(X'_k)$ are the anti-invariant sublattices of
$\mathrm{NS}(X)$ under the action of $\jmath_k$. 
\end{remark}
\begin{remark}
The fiberwise translation by a 2-torsion section $\mathsf{T}_i$ composed with $\jmath_k$ on $X$ with $i, k = 1, 2, 3$, induces an Enriques involution on $X$ whose quotient is an elliptic fibration with a rational bisection on an  Enriques surface; see \cite[Sec.~3.6]{MR2818742}.
\end{remark}

\section{Real Structure}
\label{sec:real}
In this section we describe the physics interpretation of our
families of elliptically fibered K3 surfaces $X$. Central
to our interpretation will be the notion of a \emph{real structure}
supported on $X$.  Here we can make use of the following idea.
An elliptic fibration $\pi_X\colon X\to \bP^1$ over $\bP^1$ can be described algebraically
in terms of an elliptic curve defined over the rational function field
$\bC(u)$ (where $u$ represents a affine coordinate on $\bP^1$).  The
$j$-invariant of this elliptic curve is then a rational function of $u$, with
poles at the locations of the singular fibers of the fibration.
Away from the poles, this specializes (for a fixed value of $u$) to
the $j$-function of the elliptic curve fiber of the fibration $X\to \bP^1$.

Now suppose the elliptically fibered K3 is equipped with a real
structure, that is, an antiholomorphic involution compatible
via the fibration projection with an antiholomorphic involution on $\bP^1$.
Then we can regard the base $\bP^1$ with its real structure as a genus $0$
smooth projective curve defined over $\bR$, necessary isomorphic either to
$\bP^1_{\bR}$ (this is called the species $1$ case) or to a conic $Q$ over $\bR$ with
no rational points (this is called the species $0$ case, represented by
the involution $[u : v] \mapsto [-\bar v : \bar u]$ on $\bP^1_{\bC}$).
Furthermore, we obtain from the elliptic fiber a smooth genus $1$ curve
defined either over $\bR(u)$ in the first case or over
$\bR(Q)=\bR(u,v)/(u^2+v^2+1)$ in the second case.

As explained by Kontsevich \cite{MR1403918}, the category of
D-branes in type IIB string theory compactified on a smooth complex projective
variety is expected to be described by the (bounded) derived category
of coherent sheaves on that variety.  This suggests that for an
orientifold theory on such a string background, with the orientifold
involution given by an antiholomorphic involution, the appropriate
category of D-branes should be the derived category of coherent
sheaves on the corresponding \emph{real algebraic} variety.  And a
sign choice on the O-planes should correspond to a derived category of
\emph{twisted} coherent sheaves on this real variety, with the
twisting given by a Brauer group class that vanishes on base change to
the algebraic closure $\bC$ of $\bR$. (If the Brauer group class is
trivial over a certain closed rational point, it should be viewed as a
$+$ sign, and if it gives a nonsplit quaternion algebra, it should be
viewed as a $-$ sign.) This philosophy was enunciated
and tested in \cite{MR3614975}, where it was found to be consistent with
the orientifold dualities on elliptic curves studied in
\cite{MR3316647}. Here we apply the more general
\cite{ramachandran2022derived}, which gives derived equivalences
between genus $1$ curves (not necessarily having rational points) over
a general perfect field and a dual theory given by a Brauer twist
over the Jacobian. 
\subsection{Real normal form}
\label{ssec:real_normal_form}
To define a real structure for the family of K3 surfaces
defined by Equation~\eqref{eqn:K3_X}, we must consider
antiholomorphic involutions that --  under the bundle projection $\pi_X$ ---
are equivariant with respect to
an antiholomorphic involution on the base curve $\bP^1 = \bP(u, v)$.
Note that there are two isomorphism classes of
antiholomorphic involutions $\iota$ on $\bP^1$, with representatives
given in homogeneous coordinates by
$[u:v]\mapsto [\overline u:\overline v]$ and
$[u:v]\mapsto [-\overline v:\overline u]$, respectively.
(Nikulin and Saito call these the \emph{usual} and \emph{spin} cases,
respectively \cite[p.\ 622]{MR2137825}.)
Computing the $j$-invariant of the smooth elliptic fibers for the K3
surface $\pi_X\colon X \to \mathbb{P}(u, v)$ in Equation~\eqref{eqn:K3_X}
defines a family of $j$-functions over $\bP^1$,
and equivariance implies the additional condition
\begin{equation}
\label{eqn:anti_j}
\overline{j\big(\iota(u, v)\big)}=j\big(u, v\big).
\end{equation}
\par Since we aim to study the action of anti-holomorphic involutions on
the fibers of $\pi_X$, it is convenient to slightly modify the normal form in
Equation~\eqref{eqn:K3_X}. To this end, we rescale the coordinates and
parameters according to
\beq
(u, x, y) \mapsto (\omega_{8,1} u, \omega_{8,2} x, \omega_{8,1}^2 \omega_{8,2}^2 \omega_4 y) \,, \quad (\kappa, \mu, \lambda) \mapsto  (\omega_{8,1}^2 \kappa, \omega^2_{8,2} \mu, \omega_{8,1}^2 \omega_{8,2}^2 \lambda) \,,
\eeq
where we assumed $\omega_{8,l}^8=\omega_{4,l}^4=1$ for $l=1,2$. We then obtain
the following modified normal form from Equation~\eqref{eqn:B12_1b}
\beq
\label{eqn:K3normab}
\begin{split}
X\colon \quad y^2& =\omega_2\Big( a(u, v) x^4+b(u, v) x^2z^2+\omega_{4,2}^2 a(u, v)z^4\Big) \\
&= \omega_2\Big( a'(x, z) u^4+b'(x, z) u^2v^2+\omega_{4,1}^2 a'(x, z)v^4\Big) 
\end{split}
\eeq
where we have set
\beqn
\begin{split}
	a(u, v) & = u^4 + \kappa u^2 v^2 + \omega_{4,1}^2 v^4,\quad	b(u, v) = \mu u^4 + \lambda u^2 v^2 + \omega_{4,1}^2 \mu v^4, \\
	a'(x, z) & = x^4 + \mu x^2 z^2 + \omega_{4,2}^2 z^4,\quad 	b'(x, z) = \kappa x^4 + \lambda x^2 z^2 + \omega_{4,2}^2 \kappa z^4, \\
\end{split}
\eeqn
and $\omega_{4,l}=\omega_{8,l}^2 \in \{ \pm 1, \pm i \}$ is a fourth root of unity and $\omega_2=\omega_{4,l}^2\in\{\pm 1\}$ with $l= 1, 2$.
\par Without loss of generality, we can assume that every K3 surface
in the family~\eqref{eqn:K3normab}  satisfying condition~\eqref{eqn:anti_j} has real coefficients
$\kappa, \lambda, \mu \in \bR$. While Equation~\eqref{eqn:anti_j} can also be satisfied
by allowing some coefficients to be purely
imaginary, such choice always turns out to be equivalent to a
surface with all real coefficients after a suitable coordinate transformation. 
Thus, from now on we will only consider K3 surfaces in Equation~\eqref{eqn:K3normab} with real
coefficients $\kappa, \lambda, \mu \in \bR$.
\par Since we have $\omega_2,\omega_{4,l}^2\in\{\pm 1\}$, the polynomials
$a=a(u, v)$ and $b=b(u, v)$ are real along the real line in
$\bP^1$. Here, the real line in $\bP^1$ consists of all
points $[u:v]\in\bP^1$ such that $[\bar{u}:\bar{v}]=[u:v]$. In this
way, the generic fiber of $\pi_X\colon X \to \mathbb{P}(u, v)$ over
the real line is a real genus-1 curve. 
This is also the case along the imaginary line. The
distinction is not significant, so we will
restrict ourselves to considering the real line. 
\par In summary, for a real structure on $X$ in Equation~\eqref{eqn:K3normab} with 
$\kappa, \lambda, \mu \in \bR$,  we are considering a suitable lift of antiholomorphic involution $\iota$,
that defines a (species-1) real structure on the base $\bP^1=\mathbb{P}(u, v)$ and is given by
\begin{align}
\label{eqn:ahi}
	\iota \colon & [u:v]\mapsto[\bar{u}:\bar{v}].
\end{align}
Equation \eqref{eqn:anti_j}  is  satisfied, since the involution $\iota$ sends $(a, b)$ to $(\bar{a},\bar{b})$ since they are polynomials with real coefficients.
\subsubsection{Involutions compatible with the real structure} 
Lemma \ref{lem:holinv1} gives the commuting holomorphic involutions
for the K3 surfaces defined by equation \eqref{eqn:K3_X}. After making
the aforementioned changes to obtain the the form of Equation \eqref{eqn:K3normab} the commuting holomorphic
involutions transform as follows:
\begin{lemma}
\label{lem:holinv}
Let $X$ be a K3 surface defined by Equation \eqref{eqn:K3normab}, then it admits the following commuting holomorphic involutions:
\beqn
\begin{split}
 \xi^{(1)}_X\colon& [u:v] \mapsto  [-u:v], \quad  \xi^{(2)}_X\colon [x:z] \mapsto  [-x:z], \\
 \zeta^{(1)}_X\colon& [u:v] \mapsto  [\omega_{4,1}v:u], \quad  \zeta^{(2)}_X\colon [x:z] \mapsto  [\omega_{4,2}z:x],
 \end{split}
\eeqn
and $k_X\colon y \mapsto -y$.
\end{lemma}
\par Next let us consider the existence of compatible antiholomorphic involutions, starting with involutions acting on the base curve of the elliptic fibration. As explained above, there are two isomorphism classes of such antiholomorphic involutions on the base $\bP^1$, representatives of which are given by
\begin{align}
	\iota \colon & [u:v]\mapsto[\bar{u}:\bar{v}],\\
	\eta \colon & [u:v]\mapsto[-\bar{v}:\bar{u}].
\end{align}
We ask whether $\eta$ can be lifted to give a commuting antiholomorphic involution on $X$. However, to guarantee the invariance of Equation \eqref{eqn:K3normab}, the involution $\eta$ must be modified to have the form
\begin{equation}
	\eta':[u:v]\mapsto[-\sqrt{\omega_{4,1}^2}\bar{v}:\bar{u}],
\end{equation}
so that a lift of $\eta' $, leaving the defining equation for $X$ invariant, is obtained by combining it with it with the involution $y\to \omega_{4,2} \bar{y}$ and an appropriate antiholomorphic involution on $[x:z]$. 
\par For $\omega_{4,1}^2=1$ we have $\eta=\eta'$ and
\begin{align*}
	a(\eta(u,v)) &= a(-\bar{v},\bar{u}) = \bar{v}^4+\kappa\bar{v}^2\bar{u}^2+\bar{u}^4 &= \overline{a(u,v)},\\
	b(\eta(u,v)) &= b(-\bar{v},\bar{u}) = \mu\bar{v}^4+\lambda\bar{v}^2\bar{u}^2+\mu\bar{u}^4 &= \overline{b(u,v)},
\end{align*}
since all the coefficients are real, and $\eta$ satisfies Equation~\eqref{eqn:anti_j}. Furthermore, in this case the involutions $\xi^{(1)}_X, \zeta^{(1)}_X, \iota$, and $\eta$ all commute, and 
we have
\begin{equation}
	(\iota\circ\eta)([u:v])=(\xi^{(1)}_\mathcal{X}\circ \zeta^{(1)}_\mathcal{X})([u:v])=[-v:u].
\end{equation}
We have the following:
\begin{lemma}
For general parameters  $\kappa, \lambda, \mu \in \bR$ and $\omega_{4,1}^2=1$, there is a
fixed point free antiholomorphic involution, compatible with the real structure on $X$.
\end{lemma}
\par When $\omega_{4,1}^2=-1$ the situation is quite different. The reason is following: if we consider $[u:v]$ the fiber coordinates for an elliptic fibration $\pi'_X\colon X \to \bP^1 = \mathbb{P}(x ,z)$, the smooth fibers are now all species $1$ real elliptic curves along
real points in the base.
By construction,  $\eta'([u:v])=[-i\bar{v}:\bar{u}]$ still is a symmetry of 
Equation~\eqref{eqn:K3normab} that  yields Equation~\eqref{eqn:anti_j} since $a(\eta'(u,v))=-\overline{a(u,v)}$ and
$b(\eta'(u,v))=-\overline{b(u,v)}$. However, $\eta'$ is now no longer an
involution since
\begin{align*}
	\eta'(\eta'([u:v])) &=\eta'([-i\bar{v},\bar{u}])
= [-iu:iv]\neq [u:v]\in\bP^1.
\end{align*}
Moreover, unlike in the case of $\omega_{4,1}^2=1$, the antiholomorphic involutions
do not commute with each other, nor do they commute with
the holomorphic involutions. In turn, $\eta \colon [u:v]\mapsto \left[-\bar{v}:\bar{u}\right]$ is
an involution of Equation~\eqref{eqn:K3normab} with
$\omega_{4,1}^2=-1$ if and only if $\kappa=\lambda=0$. For $\kappa=\lambda=0$, Equation~\eqref{eqn:K3normab} reduces to 
\begin{equation}
	y^2=\omega_2\big(x^4 + \mu x^2 z^2 + \omega_{4,2}^2
        z^4\big)\big(u^4-v^4\big).
\end{equation}
Considering $[u:v]$ the coordinates for the elliptic fiber as before,
we see that each smooth fiber has $j$-invariant $j(x,z)=1$ and complex structure $i$.
This is consistent with the fact that for species $1$ real
elliptic curves there are no free antiholomorphic involutions  except in the special case where the real elliptic
curve has complex structure $i$. 
We have the following:
\begin{lemma}
\label{lem:species0}
For parameters  $\kappa, \lambda, \mu \in \bR$ and $\omega_{4,1}^2=-1$, there is no
fixed point free antiholomorphic involution, compatible with the real structure on $X$, 
unless $\kappa=\lambda=0$.  
\end{lemma}
\subsubsection{Relation to the Nikulin-Saito Classification}
Nikulin and Saito \cite{MR2137825,MR2329547} embarked on a large program to
classify K3 surfaces $X$ with a commuting pair consisting of a nonsymplectic
involution $\tau$ and an antiholomorphic involution $\varphi$,
so it is appropriate to discuss where our family fit into this program.  The fundamental
invariant for Nikulin and Saito is the pair $(S,\theta)$, where
$S = H_2(X;\bZ)^\tau$ and $\theta$ is the action of $\varphi$ on $S$.
Here we take $\tau = \jmath_1$ as in Lemma \ref{lem:antisymplectic}
and take for the antiholomorphic involution $\varphi$ the lift of $\imath$ in Equation~(\ref{eqn:ahi}).
We are mostly interested in the ``standard case,'' where $\varphi$
simply sends $u,v,w,x,y,z$ to their complex conjugates.
The calculation of $S$ for our family $X$ is in Proposition \ref{prop:invlattice}.
Note that it has rank $10$ and signature $(1,9)$.  (By
\cite[(RSK3), p.\ 597]{MR2137825}, $S$ always has signature $(1, \cdots)$.)
Our case is much more symmetric than the cases studied in
\cite{MR2137825}, where $S$ only has rank $2$.  However, our situation
is somewhat similar to \cite[\S7]{MR2137825},  where $S$ is a hyperbolic plane
(that's like our $\langle2\rangle \oplus \langle-2\rangle$ summand in $S$).

Now if a homology class $x\in S$ is represented by a $\varphi$-invariant
complex curve, then since an antiholomorphic involution on a Riemann surface
reverses orientation, $\theta(x)=-x$.  First consider the case
where we take $\omega_{4,2}^2=\omega_{4,1}^2=1$.  In this case, the
lattice $S$ is, up to finite index, generated by the nonsingular invariant
fiber class $F$ over $u=0$, $K_0$, $[\mathsf{S}_1], [\mathsf{S}_2]$ and
$K_{2j}+K_{2j+1},2\le j\le 7$.  Here $F$ and $K_0$ are defined over $\bR$,
and up to a factor of $2$ generate the
$\langle2\rangle \oplus \langle-2\rangle$ summand in $S$
so $\theta$ acts on this summand by $-1$.  Since the quadratic form
$E_8$ is irreducible, $\theta$ must act on the other summand by either
$1$ or $-1$, and again we see that the sign is $-1$ by looking
at the $K_j$'s.
\subsection{The associated Jacobian fibration}
We consider the (fiberwise) Jacobian elliptic surface $J_X$ which is the compactification of the Picard variety of the generic fiber of $\pi_X\colon X \to \bP^1$ or, equivalently, 
the relative Picard variety $\mathrm{Pic}^0_{X/\mathbb{P}^1}(\mathbb{C})$. This Jacobian elliptic K3 surface $\pi_{J_X}\colon J_X \to \bP^1$ can be brought into the following Weierstrass form:
\beq
Y^2 Z = 4  X^3 - \omega_2^2 \left( \omega_{4,2}^2 a^2 + \frac{1}{12} b^2 \right) X Z^2 - \omega_2^3 b \left( \frac{1}{216} b^2 - \frac{1}{6} \omega_{4,2}^2 a^2 \right)  Z^3.
\eeq
A shift $X \mapsto X/4 + \omega_2 b(u, v)/6$ (that is, a shift by a polynomial in $\mathbb{R}(u, v)$) yields the following equation which makes the existence of a 2-torsion section over $\mathbb{R}(u, v)$ apparent:
 \beq
 \label{eqn:WEQn}
Y^2 Z = 4 X   \Big( X +  \omega_2  \big( b + 2\omega_{4,2} a \big)  Z \Big)  \Big( X +  \omega_2  \big( b - 2\omega_{4,2} a \big)  Z \Big) .
\eeq
Moreover, the full level-2 structure is manifest  over $k(u, v)$ since the $2$-torsion sections over $k(u, v)$ are given by $(X,Z) =\big(0, 1\big)$, $\big(f(u, v), 1\big)$, $\big( g(u, v), 1 \big)$ and $Y=0$ with
\beq
 f(u, v) = - \omega_2  \big(  b(u, v) - 2\omega_{4,2} a(u, v) \big), \qquad g(u, v ) = - \omega_2  \big( b(u, v) + 2\omega_{4,2} a(u, v) \big) ,
\eeq
where we have set $k=\mathbb{R}[\omega_{4,2}]$ and assumed $\kappa, \lambda, \mu \in \bR$. 
\par From $\pi_{J_X}\colon J_X \to \bP^1$, the original elliptic surface $\pi_{X}\colon X \to \bP^1$ is recovered by a Brauer twist over the Jacobian.  The corresponding twisting class $\nu$ has order $2$, i.e., it is an element 
of $\mathrm{Br}_2(J_X)$ as Equation~(\ref{eqn:WEQn}) always admits bisections. Moreover, under a base change from $\bR(u, v)$ to $\bC(u, v)$  the two elliptic surfaces are isomorphic, due to the existence of sections.   In other words, the twisting class is an element $\nu \in \mathrm{Br}_2\big( \mathbb{R}(J_X) \big)$ over the \emph{real} function field $ \mathbb{R}(J_X)$ whose pull-back to $\mathrm{Br}_2\big( \mathbb{C}(J_X) \big)$, induced by the embedding $\mathbb{R} \hookrightarrow \mathbb{C}$, is trivial, i.e., $\nu_{\mathbb{C}} \equiv 0$.
\subsubsection{The associated Azumaya algebra}
As explained above, the Brauer twist has to have order $2$, since it is killed on base change from
$\bR(u)$ to $\bC(u)$.  Therefore, it is  given by a biquaternion algebra $(r(u), x - b(u)) \otimes (s(u), x - c(u))$, where $r, s\in \bR(u)^\times$
when the Jacobian fibration has the Weierstra{\ss} form $y^2=(x-a(u))(x-b(u))(x-c(u))$, with $a,b,c$ the $2$-torsion sections
of the elliptic fibration, as described in \cite[Theorem 3.6]{MR1869390}.
\par We now construct a representative $\mathcal{V}$ for the twisting class $\nu \in \mathrm{Br}_2(\mathbb{R}(J_X))$ explicitly. The branch locus $\mathcal{B}$ of $J_X$ in Equation~(\ref{eqn:WEQn})
is given by the equation
\beq
0 = 4 X   \Big( X +  \omega_2  \big( b + 2\omega_{4,2} a \big)  Z \Big)  \Big( X +  \omega_2  \big( b - 2\omega_{4,2} a \big)  Z \Big) = \det{ M(X, Z, u, v)},
\eeq
where  $M(X,Y, u, v)$  is the symmetric $(3, 3)$-matrix with coefficients in the field $\mathbb{R}(X, Z, u, v)$, given by
\beq
 M(X, Z, u, v) = \left( \begin{array}{ccc} 4 \omega_2 \omega_{4,2}^2 a(u, v) Z & 0 & 2 X + 2 \omega_2 b(u, v) Z \\ 0 & - X & 0 \\ 2 X + 2 \omega_2 b(u, v) Z & 0 & 4 \omega_2 a(u, v) Z \end{array} \right).
\eeq 
 Here, we used a classical theorem by Hermite stating that the matrix $M$ determining the branch locus $\det{M}=0$ for the Jacobian associated with the genus-1 curve $y^2 = a_0 x^4 + a_1 x^3 z + \dots + a_4 z^4$ is
\beqn
 M = \left( \begin{array}{ccc}  4 a_0 Z & a _1 Z & \frac{2}{3} a_2 + 2 X \\ a_1 Z &  \frac{2}{3} a_2 -  X & a_3 Z \\  \frac{2}{3} a_2 + 2 X  & a_3 Z & 4 a_4 Z \end{array}\right).
\eeqn
We observe that the equation $\det{M}=0$ for $\mathcal{B}$ sits in the Hirzebruch surface $\mathbb{F}_4$ since it is invariant under rescaling $(X, Z, u, v) \mapsto (\lambda X, \lambda Z, u, v)$ and  $(X, Z, u, v) \mapsto (\lambda^4 X, Z,  \lambda  u,  \lambda v)$ with $\lambda \in \mathbb{C}^\times$. Then,  $M$ defines a conic bundle $\mathcal{V} \to \mathbb{F}_4$ which ramifies over $\mathcal{B}$ so that the fiber over each point of $\mathcal{B}$ is the union of two lines. Using projective variables $U, V, W$ in the fiber, the total space of the conic bundle is
\beq
\mathcal{V}\colon \quad  0 = - \Big(U^2 - 4 V W \Big) X + 4 \omega_2 \Big(  a(u, v) \, V^2 + b(u, v)\, VW + \omega_{4,2}^2 a(u, v) W^2 \Big) Z.
\eeq
Due to the branching of order two along $\mathcal{B}$, the conic  bundle extends to $J_X$, see \cite{MR2166182}. In this way, we can obtain an (unramified) conic bundle $\mathcal{V}_{\mathrm{ext}} \to J_X$ that agrees with $\mathcal{V}$ on $\pi_X^{-1}(\mathbb{F}_{4} - \mathcal{B})$. The even Clifford bundle associated with $\mathcal{V}_{\mathrm{ext}} \to J_X$ is an Azumaya algebra, i.e., a bundle of quaternion algebras that represents the class $\nu \in \mathrm{Br}_2(\mathbb{R}(J_X))$.
\par This can be seen as follows: the generic fiber of the conic bundle $\mathcal{V}_{\mathrm{ext}} \to J_X$ is determined by the symmetric matrix $M(X,Y, u, v)$ with coefficients in the field $\mathbb{R}(X, Z, u, v)$. Over the extension field $k=\mathbb{R}[\omega_{4,2}]$, the matrix 
\beq
 R =  \frac{1}{2} \left( \begin{array}{ccc}   \frac{1}{\omega_{4,2}} & 0 & 1 \\ 0 & 4 & 0 \\ 1 & 0 & - \omega_{4,2} \end{array}\right)
\eeq 
diagonalizes $M$ and yields
\beq
\label{eqn:M_diagonalized}
 M' = R^tMR = \left( \begin{array}{ccc}  \dfrac{X +  \omega_2  \big( b + 2\omega_{4,2} a \big)Z}{\omega_{4,2}}  & 0 & 0 \\ 0 & -4 X & 0 \\ 0 & 0 & - \omega_{4,2} \left(X +  \omega_2  \big( b - 2\omega_{4,2} a \big)  Z\right)  \end{array}\right)
\eeq 
and has coefficients in $k(X, Z, u, v)$. Writing the matrix $M' = \mathrm{diag}(\alpha, \beta, \gamma)$, the associated even Clifford algebra is the quaternion algebra with symbol $(-\alpha \beta, - \beta \gamma)$, i.e., it is the rank-4 algebra over $k(J_X)$ with generators $\mathbf{i}, \mathbf{j}$ and relations 
\beqn
 \mathbf{i}^2 = - \alpha \beta, \quad  \mathbf{j}^2 = - \beta \gamma, \quad \mathbf{i} \mathbf{j} = - \mathbf{j} \mathbf{i}.
\eeqn
\par Thus, for $\nu \in \mathrm{Br}_2(\mathbb{R}(J_X))$, $M'$ is given by Equation~(\ref{eqn:M_diagonalized}) over $k= \mathbb{R}$ only if $\omega_{4,2}^2=1$:
\begin{proposition}
Let $\omega_{4,2}^2=1$ and assume that there is no real section in Equation~(\ref{eqn:K3normab}) yielding an isomorphism between a non-singular $X$ and $J_X$ over $\mathbb{R}$. Then
a representative $\mathcal{V}$ for the nontrivial twisting class $\nu \in \mathrm{Br}_2(\mathbb{R}(J_X))$ is the rank-4 algebra over $\mathbb{R}(J_X)$ with generators $\mathbf{i}, \mathbf{j}$ and relations 
\beqn
 \mathbf{i}^2 = - \alpha \beta, \quad  \mathbf{j}^2 = - \beta \gamma, \quad \mathbf{i} \mathbf{j} = - \mathbf{j} \mathbf{i},
\eeqn
where the matrix $M' = \mathrm{diag}(\alpha, \beta, \gamma)$ is given by Equation~(\ref{eqn:M_diagonalized}).
\end{proposition}
\subsubsection{Criterion for existence of a real section}
\label{sec:RealSection}
We can now give some sufficient  criteria for the Brauer class $\nu$ to be trivial in $\mathrm{Br}_2\big( \mathbb{R}(J_X) \big)$. If we write Equation~(\ref{eqn:K3normab}) in the form
\beqn
\begin{split}
X\colon \quad y^2& =  \omega_2  a'(x, z)  \big(u^4 + \omega_{4,1}^2 v^4\big)+   \omega_2  b'(x, z) \, u^2v^2 \\
&= \omega_2 a'(x, z)  \big(u^2 \pm \omega_{4,1} v^2\big)^2+  \omega_2 \big( b'(x, z) \mp 2 \omega_{4,1} a'(x, z) \big) u^2v^2 ,
\end{split}
\eeqn
we can easily identify three sets of four (pairs of) sections for Equation~(\ref{eqn:K3normab}) over $\mathbb{C}(u, v)$. We introduce complex numbers 
$\{\sigma_{i, j} \} $ where for a fixed index $j \in \{ 1, 2, 3\}$  four roots $\sigma_{i, j}$ with $i=1, \dots,4$ are given as
solutions of  
\beq
 a'(\sigma_{i,1}, 1) = 0 , \quad b'(\sigma_{i,2}, 1) - 2 \omega_{4,1} a'(\sigma_{i,2},1)  =0, \quad b'(\sigma_{i,3}, 1) +2 \omega_{4,1} a'(\sigma_{i,3},1)  =0,
\eeq
and satisfy
\beq
  \sigma_{2,j} = -  \sigma_{1,j}, \quad  \sigma_{3,j} =   \frac{\omega_{4,2}}{\sigma_{1,j}}, \quad  \sigma_{4,j} = -  \frac{\omega_{4,2}}{\sigma_{1,j}}.
\eeq
Then, 12 sections for Equation~(\ref{eqn:K3normab}) are  given by $(x, y, z) = (\sigma_{i,j}, \pm y_{i,j}, 1)$ with
\beq
\label{eqn:sections}
\begin{split}
  y_{i, 1} &= \sqrt{\omega_2 b'(\sigma_{i,1}, 1)} \, uv, \\
  y_{i, 2} &= \sqrt{\omega_2 a'(\sigma_{i,2}, 1)} \, (u^2+ \omega_{4,1} v^2) = \sqrt{\omega_{4,1} b'(\sigma_{i,2}, 1)} \, \frac{u^2+ \omega_{4,1} v^2}{\sqrt{2}} , \\
  y_{i, 3} &= \sqrt{\omega_2 a'(\sigma_{i,3}, 1)} \, (u^2-  \omega_{4,1} v^2) = \sqrt{-\omega_{4,1} b'(\sigma_{i,3}, 1)} \, \frac{u^2- \omega_{4,1} v^2}{\sqrt{2}}  .
\end{split}
\eeq 
Thus we have the following:
\begin{lemma}
\label{lem:real_section}
If there is a root $\sigma_{i,1}\in \mathbb{R}$ with $\omega_2 b'(\sigma_{i,1}, 1) \in \mathbb{R}_{>0}$, or, $\sigma_{i,2}\in \mathbb{R}$ with $\omega_{4,1} b'(\sigma_{i,2}, 1)  \in \mathbb{R}_{>0}$, or, $\sigma_{i,3}\in \mathbb{R}$ with $-\omega_{4,1} b'(\sigma_{i,3}, 1)  \in \mathbb{R}_{>0}$ for $i=1, \dots,4$, then  Equation~(\ref{eqn:K3normab})  will admit a section over $\mathbb{R}(u, v)$.
\end{lemma}
\begin{remark}
\label{rem:duality3}
The constructed 12 sections for the fibration $\pi_X \colon X \to \mathbb{P}(u, v)$ are the 12 reducible fiber components of the fibration $\tilde{\pi}_X \colon X \to \mathbb{P}(x, z)$ which has singular fibers $12 I_2$.
This is, swapping the roles of the two rulings for $\mathbb{F}_0=\mathbb{P}(u, v)\times\mathbb{P}(x, z)$ in Equation~(\ref{eqn:K3normab}) interchanges sections and components of reducible fibers.
\end{remark}
\par Let us take a closer look at the first case of Lemma~\ref{lem:real_section} where we are considering the roots $\sigma_{i,1}$. One has the following: the polynomial $a'(x,1)$ has 2 real roots  for $\omega_{4,2} = \pm i$, 4 real roots for $\omega_{4,2}= \pm 1$ and $\mu \le -2$, and no real roots otherwise. Writing $b'(x, 1) = \kappa a'(x,1) + (\lambda - \kappa \mu) x^2$ it follows that for any real root $\sigma_{i,1}$ we have $\omega_2 b'(\sigma_{i,1}, 1) \in \mathbb{R}_{>0}$ if $\sign{\omega_2} =\sign{(\lambda - \kappa\mu)}$. Thus, we have the following:
\begin{corollary}
\label{cor1}
Assume that $\sign{\omega_2} =\sign{(\lambda - \kappa\mu)}$. Then for $\omega_{4,2}^2 = -1$ or  $\omega_{4,2}^2 = 1$, $\mu \le -2$, there is a real section in Equation~(\ref{eqn:K3normab}) yielding an isomorphism between a non-singular $X$ and $J_X$ over $\mathbb{R}$ and, consequently, $\nu \equiv 0$ for $\nu \in \mathrm{Br}_2(\mathbb{R}(J_X))$. 
\end{corollary}

\par In the second and third case of Lemma~\ref{lem:real_section} we must have $\omega_{4,1} = \pm 1$ for a section to be defined over $\mathbb{R}(u, v)$ because of Equation~(\ref{eqn:sections}). The roots $\sigma_{i, 2}$ and $\sigma_{i, 3}$ are solutions of
\beq
 x^4 + \frac{\lambda \mp 2 \omega_{4,1} \mu}{\kappa \mp 2 \omega_{4,1}} x^2 + \omega_{4,2}^2 = 0,
\eeq 
where we assume $\kappa \mp 2 \omega_{4,1} \neq 0$. As before, we have 2 real roots for $\omega_{4,2} = \pm i$, 4 real roots for $\omega_{4,2}= \pm 1$ and $(\lambda \mp 2 \omega_{4,1} \mu)/(\kappa \mp 2 \omega_{4,1})  \le -2$, and no real roots otherwise. Writing
\beq
 a'(x, 1) = \Big(x^4 + \frac{\lambda \mp 2 \omega_{4,1} \mu}{\kappa \mp 2 \omega_{4,1}} x^2 + \omega_{4,2}^2\Big) + \left( \mu -  \frac{\lambda \mp 2 \omega_{4,1} \mu}{\kappa \mp 2 \omega_{4,1}} \right) x^2 
\eeq 
we  can  easily verify whether  $\omega_2 a'(\sigma_{i,2}, 1) \in \mathbb{R}_{>0}$ or $\omega_2 a'(\sigma_{i,3}, 1) \in \mathbb{R}_{>0}$. Thus, we have the following:
\begin{corollary}
\label{cor2}
Assume that $\omega_{4,1}^2 = 1$ and $\sign{\omega_2} =\sign (  \tfrac{\lambda - \kappa \mu }{\pm 2 \omega_{4,1} - \kappa} )$. Then for $\omega_{4,2}^2 = -1$ or  $\omega_{4,2}^2 = 1$, $(\lambda \mp 2 \omega_{4,1} \mu)/(\kappa \mp 2 \omega_{4,1})  \le -2 $, there is a real section in Equation~(\ref{eqn:K3normab}) yielding an isomorphism between a non-singular $X$ and $J_X$ over $\mathbb{R}$ and, consequently, $\nu \equiv 0$ for $\nu \in \mathrm{Br}_2(\mathbb{R}(J_X))$. 
\end{corollary}
\subsection{Real families and their charges}
We now consider three families of real K3 surfaces whose string limits give the three different type IIB theories on $\bP^1$ with four $I_0^*$ fibers considered in \cite{MR3267662,MR3316647}. The three cases are differentiated by the charges of the $I_0^*$ fibers with the three possibilities being $(+,+,+,+)$, $(+,+,-,-)$, and $(+,+,+,-)$. Before going into each case in detail we summarize the results in Table~\ref{tab:3realfam}.
\begin{table}
	\begin{tabular}{||c|c|c|c|c|c|c||}
  \hline 
  $\omega_2$ & $\omega_{4,1}^2$ & $\omega_{4,2}^2$ & Constraints & Locations of $+$ $I_2$'s & Locations of $-$ $I_2$'s & Isotrivial Limit \\
  \hline\hline 
  $1$ & $1$ & $1$ & \thead{$\kappa,\mu<-2$\\ $\lambda>\kappa\mu$}& \thead{$\pm k,\pm k^{-1}$\\$\pm l,\pm l^{-1}$\\$\pm m,\pm m^{-1}$}& & $l,m\to k$ \\
  \hline
  $-1$ & $1$ & $1$ & \thead{$\kappa,\mu>2$\\ $\lambda>2(\kappa+\mu-2)$}& \thead{$ik, ik^{-1}$\\$ il,il^{-1}$\\$ im, im^{-1}$} & \thead{$- ik,- ik^{-1}$\\$- il,- il^{-1}$\\$-im,- im^{-1}$} & $l,m\to k$ \\
  \hline
  $1$ & $-1$ & $-1$ & $\lambda>\kappa\mu$& \thead{$\pm k,ik^{-1}$\\$\pm c,ic^{-1}$\\$\pm \bar{c},i\bar{c}^{-1}$}& \thead{$-ik^{-1}$\\$-ic^{-1}$\\$-i\bar{c}^{-1}$} & $c,\bar{c}\to k$\\
  \hline
\end{tabular}
\caption{Three families of real K3 surfaces whose isotrivial limit has $4$ $I_0^*$ fibers with charge $(+,+,+,+)$, $(+,+,-,-)$, $(+,+,+,-)$. The $+ I_2$'s are the $I_2$ fibers that will merge to form an $I_0^*$ fiber with $+$ charge, and similar for $- I_2$'s. Here $k,l,m\in\bR_{>0}$ and $c\in\bC\backslash(\bR\cup i\bR)$.}
\label{tab:3realfam}
\end{table}

In all three families let us consider the situation where $\mathbb{P}(u, v)$ is the base curve of the elliptically fibered K3 surface $\pi_X \colon X \to \mathbb{P}(u, v)$. The discriminant function then is proportional to
\beq
 a(u, v)^2 \Big( a(u, v) + \omega_{4,2} b(u, v) \Big)^2  \Big( a(u, v) - \omega_{4,2} b(u, v) \Big)^2.
\eeq 

\subsubsection{The $(+, +, +, +)$ family}
\label{sec:++++}
To obtain 12 real roots, marking the location of the 12 singular fibers of type $I_2$, we set $\omega_2=\omega_{4,1}=\omega_{4,2} =1$,
\begin{equation}
\label{eqn:spec2constraints}
\begin{split}
	\kappa, \ \mu,  <  -2 , \quad \text{and}\quad \lambda > \kappa\mu.
	\end{split}
\end{equation}
This shows that
\begin{equation}
\label{eqn:Defchipsi}
	\chi= \frac{\lambda + 2\kappa}{\mu+2}, \ \psi= \frac{\lambda - 2\kappa}{\mu-2} <  -2 .
\end{equation} 
Now introduce $k, l, m$ with
\begin{equation}
\label{eqn:Defklm}
	\begin{split}
	k&=\frac{1}{\sqrt{2}}\sqrt{-\kappa+\sqrt{\kappa^2-4}}	,\quad l=\frac{1}{\sqrt{2}}\sqrt{-\chi+\sqrt{\chi^2-4}},\\
	&\qquad\text{and}\quad  m=\frac{1}{\sqrt{2}}\sqrt{-\psi+\sqrt{\psi^2-4}}	.
	\end{split}
\end{equation}
We see that $k,l,m\in\bR_{>1}$ since $\kappa,\chi,\psi<-2$, and 
\beq
\label{eqn:relat1}
\begin{split}
 \kappa = - \left( k^2 + \frac{1}{k^2} \right) , \quad \chi   =  - \left(l^2 + \frac{1}{l^2}\right) , \quad \text{and} \quad \psi =  - \left(m^2 - \frac{1}{m^2}\right) .
\end{split}
\eeq 
It follows that the roots of $a(u,1)=0$ are given by $\{ \pm k, \pm \tfrac{1}{k} \}$, the roots of $b(u,1)+2 a(u,1)=0$ are given by $\{ \pm l, \pm \tfrac{1}{l} \}$, and the roots $b(u,1)-2 a(u,1)=0$ are given by $\{ \pm m, \pm \tfrac{1}{m} \}$.

We define the quantity
\begin{equation}
	\varepsilon_1=\lambda-\kappa\mu>0,
\end{equation}
and the family of $\varepsilon$-linear functions
\beq
\label{eqn:family}
 \lambda_\varepsilon = \kappa\mu + \varepsilon, \qquad \kappa_\varepsilon = \kappa, \qquad \mu_\varepsilon = \mu,
\eeq 
for all $\varepsilon\in[0,\infty)$. Clearly $\lambda_{\varepsilon_1}=\lambda,$ $\kappa_{\varepsilon_1}=\kappa$, and $\mu_{\varepsilon_1}=\mu$. Using Equation~(\ref{eqn:Defchipsi}) we obtain functions $\chi_\varepsilon$ and $\psi_\varepsilon$ defined on $[0,\infty)$ by
\begin{equation}
	\chi_\varepsilon=\kappa+\frac{\varepsilon}{\mu+2},\quad \psi_\varepsilon=\kappa+\frac{\varepsilon}{\mu-2}.
\end{equation}
This gives us functions $k_\varepsilon$, $l_\varepsilon$, and $m_\varepsilon$ from Equation~(\ref{eqn:Defklm}), given by
\begin{equation}
\begin{split}
	k_\varepsilon=k,\quad l_\varepsilon=k-\frac{k^3}{2(k^4-1)(\mu+2)}\varepsilon+O(\varepsilon^2),\\
	\text{and} \quad m_\varepsilon=k-\frac{k^3}{2(k^4-1)(\mu-2)}\varepsilon+O(\varepsilon^2).
	\end{split}
\end{equation}

For all $\varepsilon\in[0,\infty)$ $(\kappa_\varepsilon,\mu_\varepsilon,\lambda_\varepsilon)$ satisfy the conditions in Equation~(\ref{eqn:spec2constraints}), so
\begin{equation}
	\begin{split}
		\chi_\varepsilon,\; \psi_\varepsilon < -2,\qquad\\
		\quad\text{and}\quad k_\varepsilon,\; l_\varepsilon,\; m_\varepsilon >1.
	\end{split}
\end{equation}
Furthermore, we see that when $\varepsilon>0$
\begin{equation}
	\chi_\varepsilon<\psi_\varepsilon<\kappa_\varepsilon < -2,
\end{equation}
or equivalently,
\begin{equation}
	1 < k_\varepsilon < m_\varepsilon < l_\varepsilon.
\end{equation} 
Thus, we have the following:
\begin{proposition}
\label{prop:++++}
Let $\kappa, \mu , \lambda \in \mathbb{R}$ satisfying Equation~(\ref{eqn:spec2constraints}), and assume $\omega_2=\omega_{4,1}=\omega_{4,2} =1$. Define real-valued functions $\kappa_\varepsilon,  \mu_\varepsilon,  \lambda_\varepsilon  \in \mathbb{R}$, over $\varepsilon \in [0, \infty)$ as in Equation~(\ref{eqn:family}). Then Equation~(\ref{eqn:K3normab}) defines a family of real elliptically fibered K3 surfaces $\pi_{X_\varepsilon} \colon X_\varepsilon \to \mathbb{P}(u, v)$ whose limit for $\varepsilon=0$ is $\mathrm{Kum}(E_1 \times E_2)$ with
\beq
j (E_1) =  \frac{(12+\kappa^2)^3}{108 (4-\kappa^2)^2}, \quad
j (E_2) = \frac{(12+\mu^2)^3}{108 (4-\mu^2)^2}.
\eeq
$X_\varepsilon$ is non-singular for all $\varepsilon\in [0,\infty)$ and has $12$ $I_2$ fibers when $\varepsilon\ne 0$. Moreover, there is a real section in Equation~(\ref{eqn:K3normab}) yielding an isomorphism between every $X_\varepsilon$ and $J_{X_\varepsilon}$ over $\mathbb{R}$ for all $\varepsilon \in [0, \infty)$.
\end{proposition}
\begin{proof}
$X_\varepsilon$ is non-singular by Lemma~\ref{lem:NormFormisK3} since $[ 1: \kappa_\varepsilon : \lambda_\varepsilon : \mu_\varepsilon ]\notin\mathcal{D}$ for all $\varepsilon\in (-\varepsilon_0,\infty)$. Furthermore $k_\varepsilon,l_\varepsilon$, and $m_\varepsilon$ are all distinct for $\varepsilon\neq 0$ and all greater than $1$ for $\varepsilon\in(-\varepsilon_0,\infty)$ so we see that the discriminant has $12$ distinct real roots for $\varepsilon\in(-\varepsilon_0,\infty)\backslash\{0\}$. As explained in Lemma~\ref{lem:NormFormisK3} and Remark~\ref{rem:kummer}, the elliptic fibration is isotrivial for $[ 1: \kappa : \lambda : \mu ] \in  \mathbb{P}^3$ with $\lambda-\kappa\mu=0$. The limiting K3 is isomorphic to the Kummer surface $\mathrm{Kum}(E_1 \times E_2)$ where the elliptic curves are given by
\beq
\begin{split}
 E_1\colon \  y_1^2 = u^4 + \kappa  u^2v^2 +v^4, \quad   E_2\colon \  y_2^2 = x^4 + \mu  x^2z^2 +z^4.
\end{split}
\eeq

When $\varepsilon\in(0,\infty)$ $X_\varepsilon$ satisfies the conditions of Corollary~\ref{cor1} showing there exists a real section.
\end{proof}

In Section~\ref{sec:+++-2} we will be interested in the number of real sections that exist so let us make it explicit for this case now. When determining sections of $X_\varepsilon$ defined by Equation~(\ref{eqn:K3normab}) we consider the roots of $a_\varepsilon'(x,1)$ and $b_\varepsilon'(x,1)\pm 2a'_\varepsilon(x,1)$ as explained in section \ref{sec:RealSection}. We can write these as
\begin{equation}
	\begin{split}
		a_\varepsilon'(x,1)&=x^4+\mu_\varepsilon x^2+1,\\
		b'_\varepsilon(x,1)+2 a'_\varepsilon(x,1)&=(\kappa_\varepsilon+2)(x^4+\chi'_\varepsilon x^2+1),\\
		b'_\varepsilon(x,1)-2 a'_\varepsilon(x,1)&=(\kappa_\varepsilon-2)(x^4+\psi'_\varepsilon x^2+1),
	\end{split}
\end{equation}
where 
\begin{equation}
\label{eqn:chieps}
\begin{split}
	\chi'_\varepsilon=\frac{\lambda_\varepsilon+2\mu_\varepsilon}{\kappa_\varepsilon+2}=\mu+\frac{\varepsilon}{\kappa+2},\\\text{and} \quad \psi'_\varepsilon=\frac{\lambda_\varepsilon-2\mu_\varepsilon}{\kappa_\varepsilon-2}=\mu+\frac{\varepsilon}{\kappa-2}.
	\end{split}
\end{equation}
For $\varepsilon\in(0,\infty)$,
\begin{equation}
	\chi'_\varepsilon<\psi'_\varepsilon<\mu_\varepsilon=\mu<-2.
\end{equation}
Therefore, for all $\varepsilon>0$ there exist $k'_\varepsilon, l'_\varepsilon, m'_\varepsilon \in \bR$ with
\begin{equation}
	1<k'_\varepsilon<m'_\varepsilon<l'_\varepsilon
\end{equation}
such that the roots of $a'_\varepsilon(x,1)=0$ are given by $\{\pm k'_\varepsilon,\pm\tfrac{1}{k'_\varepsilon}\}$, the roots of $b'_\varepsilon(x,1)+2 a'_\varepsilon(x,1)=0$ are given by $\{\pm l'_\varepsilon,\pm\tfrac{1}{l'_\varepsilon}\}$, and the roots of $b'_\varepsilon(x,1)-2 a'_\varepsilon(x,1)=0$ are given by $\{\pm m'_\varepsilon,\pm\tfrac{1}{m'_\varepsilon}\}$.
\begin{corollary}
	Let $\kappa, \mu , \lambda \in \mathbb{R}$ satisfying Equation~(\ref{eqn:spec2constraints}), and assume $\omega_2=\omega_{4,1}=\omega_{4,2} =1$. Define the real-valued functions $\kappa_\varepsilon,  \mu_\varepsilon,  \lambda_\varepsilon  \in \mathbb{R}$, over $\varepsilon \in [0, \infty)$ as in Equation~(\ref{eqn:family}). Then $X_\varepsilon$ defined by Equation~(\ref{eqn:K3normab}) has $12$ pairs of real sections for $\varepsilon>0$.
\end{corollary}
\begin{proof}
	$a'(x,1),b'(x,1)\pm 2a'(x,1)$ are all quartics with four real roots. $a'(x,1)$ has positive leading coefficient whereas $b'(x,1)\pm 2a'(x,1)$ both have negative leading coefficients. This together with the fact that for all $\varepsilon>0$,
	\begin{equation}
		1 < k'_\varepsilon < m'_\varepsilon < l'_\varepsilon.
	\end{equation}
	shows that 
	\begin{equation}
		\begin{split}
			b'(\sigma_{i,1})>0& \quad\text{for}\quad \sigma_{i,1}\in \{\pm k'_\varepsilon,\pm\frac{1}{k'_\varepsilon}\},\\
			b'(\sigma_{i,2})>0& \quad\text{for}\quad \sigma_{i,2}\in \{\pm m'_\varepsilon,\pm\frac{1}{m'_\varepsilon}\},\\
			b'(\sigma_{i,3})<0& \quad\text{for}\quad \sigma_{i,3}\in \{\pm l'_\varepsilon,\pm\frac{1}{l'_\varepsilon}\}.
		\end{split}
	\end{equation}
	on $[0,\infty)$. Therefore all $12$ pairs of sections in Equation~(\ref{eqn:sections}) are real by Lemma~\ref{lem:real_section}.
\end{proof}

Notice that $j(E_i)>1$ since
$$j(E_i)-1=\frac{x_i^2(x_i^2-36)^2}{108(4-x_i^2)^2}>0,$$
where $x_1=\kappa$ and $x_2=\mu$. This corresponds to both elliptic curves being species $2$ real elliptic curves. The isotrivial limit gives back the standard F-theory description of the type IIB theory on $\bP^1$ with $4$ $O7^+$-planes. Recall that we are using notation where an $O^+$-plane carries negative $D$-brane charge, so in this case each $O7^+$-plane carries $(-4)$ units of $D7$-brane charge and requires $4$ $D7$-branes for tadpole cancellation. We see that in the isotrivial limit the $I_2$ fibers at $k_\varepsilon,l_\varepsilon$, and $m_\varepsilon$ merge to form an $I_0^*$ fiber at $k$, describing the the $O7^+$-plane. Similarly, the other $I_0^*$ fibers are formed by merging three $I_2$ fibers and are located at $\pm\tfrac{1}{k}$ and $-k$. In the string limit this gives the type IIB orientifold on $(S^{1,1}\times S^{1,1})_{(+,+,+,+)}$ described in \cite{MR3267662} and \cite{MR3316647}.

\subsubsection{The $(+,+,-,-)$ family}
\label{sec:++--}

This case is very similar to the previous one, but we now want $12$ purely imaginary roots and $(u,v,x,z,y)\to (\bar u,\bar v,\bar x,\bar z,\bar y)$ to be fixed point free. This additionally requires $y^2$ to be negative definite. Therefore, we set $\omega_2=-1$, $\omega_{4,1}^2=\omega_{4,2}^2=1$, 
\begin{equation}
\label{eqn:spec0constraints}
\begin{split}
	\kappa, \ \mu,  >  2 , \quad \text{and}\\
	 \quad \lambda>2(\kappa+\mu-2).
	\end{split}
\end{equation}
The rest follows the same as the previous case with some changes of sign, so we will go through it quickly. This shows that
\begin{equation}
	\psi= \frac{\lambda - 2\kappa}{\mu-2},  \ \chi= \frac{\lambda + 2\kappa}{\mu+2} >  2 .
\end{equation}
We again introduce $k, l, m \in \mathbb{R}_{> 1}$ with
\begin{equation}
\label{eqn:Defklm0}
	\begin{split}
	k&=\frac{1}{\sqrt{2}}\sqrt{\kappa+\sqrt{\kappa^2-4}}	,\quad l=\frac{1}{\sqrt{2}}\sqrt{\chi+\sqrt{\chi^2-4}},\\
	&\qquad\text{and}\quad  m=\frac{1}{\sqrt{2}}\sqrt{\psi+\sqrt{\psi^2-4}}	.
	\end{split}
\end{equation}
It follows that the roots of $a(u,1)=0$ are given by $\{ \pm i k, \pm \tfrac{i}{k} \}$, the roots of $b(u,1)+2 a(u,1)=0$ are given by $\{ \pm i l, \pm \tfrac{i}{l} \}$, and the roots $b(u,1)-2 a(l,1)=0$ are given by $\{ \pm i m, \pm \tfrac{i}{m} \}$.

We define the quantities
\begin{equation}
	\varepsilon_1=\lambda-\kappa\mu,\quad \varepsilon_0=(\kappa-2)(\mu-2)>0,
\end{equation}
and the family of $\varepsilon$-linear functions
\beq
\label{eqn:family0}
 \lambda_\varepsilon = \kappa\mu + \varepsilon, \qquad \kappa_\varepsilon = \kappa, \qquad \mu_\varepsilon = \mu,
\eeq 
for all $\varepsilon\in(-\varepsilon_0,\infty)$. Clearly $\lambda_{\varepsilon_1}=\lambda,$ $\kappa_{\varepsilon_1}=\kappa$, and $\mu_{\varepsilon_1}=\mu$. Furthermore, $(\kappa_\varepsilon,\mu_\varepsilon,\lambda_\varepsilon)$ satisfy the conditions in Equation~(\ref{eqn:spec0constraints}) for all $\varepsilon\in(-\varepsilon_0,\infty)$. Therefore the functions $\chi_\varepsilon$ and $\psi_\varepsilon$ defined in Equation~(\ref{eqn:chieps}) satisfy
\begin{equation}
	\chi_\varepsilon,\ \psi_\varepsilon \ > \ 2
\end{equation}
for all $\varepsilon\in(-\varepsilon_0,\infty)$. Furthermore,
\begin{equation}
	\chi'_\varepsilon,\ \psi'_\varepsilon \ > \ 2
\end{equation} 
as well. This shows us that the zeros of the functions $a_\varepsilon(u,1)$, $b_\varepsilon(u,1\pm 2a_\varepsilon(u,1)$, $a'_\varepsilon(x,1)$, and $b'_\varepsilon(x,1\pm 2a'_\varepsilon(x,1)$ are all purely imaginary, so the functions are all positive definite on $\bR$. Therefore $y$ is purely imaginary over $\bR$ since $\omega_2=-1$. This gives us the following

\begin{proposition}
	Let $\kappa,\mu,\lambda\in\bR$ such that they satisfy Equation~(\ref{eqn:spec0constraints}), and assume $\omega_2=-1$ and $\omega_{4,1}^2=\omega_{4,2}^2=1$. Define $\varepsilon_0=(\kappa-2)(\mu_2)>0$ and real valued functions $\kappa_\varepsilon,\mu_\varepsilon,\lambda_\varepsilon\in\bR$ over $\varepsilon\in(-\varepsilon_0,\infty)$ as in Equation~(\ref{eqn:family0}). Then Equation~(\ref{eqn:K3normab}) defines a family of real elliptically fibered K3 surfaces $\pi_{X_\varepsilon} \colon X_\varepsilon \to \mathbb{P}(u, v)$ whose limit for $\varepsilon=0$ is $\mathrm{Kum}(E_1 \times E_2)$ with
\beq
j (E_1) =  \frac{(12+\kappa^2)^3}{108 (4-\kappa^2)^2}, \quad
j (E_2) = \frac{(12+\mu^2)^3}{108 (4-\mu^2)^2}.
\eeq
$X_\varepsilon$ is non-singular for all $\varepsilon\in (-\varepsilon_0,\infty)$ and has $12$ $I_2$ fibers when $\varepsilon\ne 0$. Moreover, there does not exist a real section in Equation~(\ref{eqn:K3normab}) yielding an isomorphism between every $X_\varepsilon$ and $J_{X_\varepsilon}$ over $\mathbb{R}$ for all $\varepsilon \in (-\varepsilon_0, \infty)$.
\end{proposition}
\begin{proof}
	Everything follows exactly the same as the proof of Proposition~\ref{prop:++++} except that there is no real section over $\bR$, but this follows immediately from the fact that $y$ is purely imaginary over $\bR$.
\end{proof}

The elliptic curves in the limit, $E_i$, are now species $0$ elliptic curves. They have same $j$-invariants as we found in the previous case, because species $0$ and $2$ elliptic curves are related to each other by exchanging the real and imaginary line. In fact, for fixed $\kappa,\mu,\lambda$ satsifying Equation~(\ref{eqn:spec2constraints}) if we send $(u,x,y)\to(iu,ix,iy)$ then $(\kappa,\mu,\lambda)\to(-\kappa,-\mu,\lambda)$, which satisfy Equation~(\ref{eqn:spec0constraints}) and sends $\omega_2=1\to\omega_2=-1$. This shows that this case and the $(+,+,+,+)$ case considered in Section~\ref{sec:++++} are the same except for an exhange of the real and imaginary axes.

The string limit of $X_0$ is the type IIB orientifold compactified on $(S^{1,1}\times S^{1,1})_{(+,+,-,-)}$ considered in \cite{MR3267662} and \cite{MR3316647}. There are again four $O7$-planes, bute they now have different charge. Two of them are normal $O7^+$-planes with $(-4)$ units of $D7$-brane charge each, while the other two are $O7^-$-planes with $(+4)$ units of $D7$-brane charge each. We choose the $O7^+$-planes to be located at $ik$ and $\tfrac{i}{k}$, so the $O7^-$ planes are located at $-ik$ and $-\tfrac{i}{k}$. We would expect the $3$ $I_2$ fibers to that merge to form each $I_0^*$ fiber to have the same charge. Therefore for the $12$ $I_2$ fibers of $X_\varepsilon$ $\varepsilon\ne 0$ we would expect $6$ of them to have $+$ charge and $6$ of them to have $-$ charge.

\subsubsection{The $(+,+,+,-)$ family}
\label{sec:+++-}

In this case we set $\omega_{4,1}^2=\omega_{4,2}^2=-1$ and $\omega_2=1$. If we want there to be a real section over $\bR$ then we also must require
\begin{equation}
	\label{eqn:spec1constraints}
	\lambda>\kappa\mu,
\end{equation}
and just to avoid some singularities, let's assume $\kappa\ne \pm\mu$. Now introduce
\begin{equation}
	k=\frac{1}{\sqrt{2}}\sqrt{-\kappa+\sqrt{\kappa^2+4}}\in\bR_{>0}.
\end{equation}
It follows that the roots of $a(u,1)$ are $\{\pm k,\pm\tfrac{i}{k}\}$. However, in this case
\begin{equation}
\label{eqn:chi1}
	\chi=\frac{\lambda+2\kappa i}{\mu+2i}\quad \text{and}\quad \psi=\frac{\lambda-2\kappa i}{\mu-2i},
\end{equation}
so $\psi=\bar\chi$, and we see that the roots of $b(u,1)\pm 2ia(u,1)$ are neither real nor purely imaginary. Let
\begin{equation}
	c=\frac{1}{\sqrt{2}}\sqrt{-\chi+\sqrt{\chi^2+4}}.
\end{equation}
Then $c$ is a root of $b(u,1)+ 2ia(u,1)=0$ and the other three are $-c$, and
\begin{equation}
	\frac{\pm i}{c}=\frac{\pm 1}{\sqrt{2}}\sqrt{-\chi-\sqrt{\chi^2+4}}.
\end{equation}
The roots of $b(u,1)- 2ia(u,1)=0$ are then $\{\pm\bar{c},\pm\tfrac{i}{\bar{c}}\}$. Note that $c\in\bC\backslash(\bR\cup i\bR)$ and that $|c|\ne 1$, so the $12$ roots of the discriminant are all distinct.

We again define the quantity
\begin{equation}
	\varepsilon_1=\lambda-\kappa\mu>0,
\end{equation}
and the family of $\varepsilon$-linear functions
\beq
\label{eqn:family1}
 \lambda_\varepsilon = \kappa\mu + \varepsilon, \qquad \kappa_\varepsilon = \kappa, \qquad \mu_\varepsilon = \mu,
\eeq 
for all $\varepsilon\in[0,\infty)$. Clearly $\lambda_{\varepsilon_1}=\lambda,$ $\kappa_{\varepsilon_1}=\kappa$, and $\mu_{\varepsilon_1}=\mu$. Furthermore, $(\kappa_\varepsilon,\mu_\varepsilon,\lambda_\varepsilon)$ satisfy the conditions in Equation~(\ref{eqn:spec1constraints}) for all $\varepsilon\in[0,\infty)$. Using Equation~(\ref{eqn:chi1}) we obtain the function
\begin{equation}
	\chi_\varepsilon=\kappa+\frac{\varepsilon}{\mu+2i}=\left(\kappa+\frac{\mu}{\mu^2+4}\varepsilon\right)-\frac{2i}{\mu^2+4}\varepsilon,
\end{equation}
and similarly $\psi_\varepsilon=\bar{\chi}_\varepsilon$. 

The roots of $a'(x,1)=0$ are given by $\{\pm k',\pm\tfrac{i}{k'}\}$ where
\begin{equation}
	k'=\frac{1}{\sqrt{2}}\sqrt{-\mu+\sqrt{\mu^2+4}}\in\bR_{>0},
\end{equation}
the roots of $b'(x,1)+ 2ia'(x,1)=0$ are given by $\{\pm c',\pm\tfrac{i}{c'}\}$, and the roots of $b'(x,1)- 2ia'(x,1)=0$ are given $\{\pm\bar{c}',\pm\tfrac{i}{\bar{c}'}\}$, where
\begin{equation}
	c'=\frac{1}{\sqrt{2}}\sqrt{-\chi'+\sqrt{\chi'^2+4}},\qquad \chi'=\frac{\lambda+2\mu i}{\kappa+2i}.
\end{equation}
Again, due to symmetry $c'\in\bC\backslash(\bR\cup i\bR)$ and that $|c'|\ne 1$. We now have the following

\begin{proposition}
	\label{prop:+++-}
	Let $\kappa,\mu,\lambda\in\bR$ such that $[1:\kappa:\lambda:\mu]\notin \mathcal{D}$ from Lemma~\ref{lem:NormFormisK3}, $\lambda-\kappa\mu>0$, and $\mu\neq \pm\kappa$. Assume $\omega_2=1$ and $\omega_{4,1}^2=\omega_{4,2}^2=-1$ and define real valued functions $\kappa_\varepsilon,\mu_\varepsilon,\lambda_\varepsilon\in\bR$ over $\varepsilon\in[0,\infty)$ as in Equation~(\ref{eqn:family1}). Then Equation~(\ref{eqn:K3normab}) defines a family of real elliptically fibered K3 surfaces $\pi_{X_\varepsilon} \colon X_\varepsilon \to \mathbb{P}(u, v)$ whose limit for $\varepsilon=0$ is $\mathrm{Kum}(E_1 \times E_2)$ with
		\beq
			j (E_1) = - \frac{(\kappa^2-12)^3}{108 (\kappa^2+4)^2}, \quad
			j (E_2) = - \frac{(\mu^2-12)^3}{108 (\mu^2+4)^2}.
		\eeq
$X_\varepsilon$ is non-singular for all $\varepsilon\in [0,\infty)$ and has $12$ $I_2$ fibers when $\varepsilon\ne 0$. Moreover, there exists one set of four real sections in Equation~(\ref{eqn:K3normab}) yielding an isomorphism between every $X_\varepsilon$ and $J_{X_\varepsilon}$ over $\mathbb{R}$ for all $\varepsilon \in [0, \infty)$.
\end{proposition}
\begin{proof}
	For all $\varepsilon>0$ $\lambda_\varepsilon-\kappa_\varepsilon\mu_\varepsilon>0$. The proof follows the same form of the proof of Propsition~\ref{prop:++++} except for the existence of a real section, but the $\varepsilon\to 0$ limiting K3 is isomorphic to the Kummer surface $\Kum{(E_1\times E_2)}$ where the elliptic curves are now given by
	\beq
\begin{split}
 E_1\colon \  y_1^2 = u^4 + \kappa  u^2v^2 -v^4, \quad   E_2\colon \  y_2^2 = x^4 + \mu  x^2z^2 -z^4.
\end{split}
\eeq

The existence of a real section over $\bR$ follows immediately from Corollary~\ref{cor1}. We see from Equation~(\ref{eqn:sections}) that the only real sections are the four possibilities $(x,y,z) = (\pm k', \pm\sqrt{b'(\pm k',1)}uv,1 )$.
\end{proof}

Now $j(E_i)\leq 1$, since
\begin{equation}
	j(E_i)-1=-\frac{x_i^2(x_i^2+36)^2}{108(b^2+4)}\leq 0
\end{equation}
where $x_1=\kappa$ and $x_2=\mu$. $E_i$ are now both species $1$ real elliptic curves. The string limit of $X_0$ gives the type IIB orientifold compactified on $(S^{1,1}\times S^{1,1})_{(+,+,+,-)}$ that was described in \cite{MR3267662} and \cite{MR3316647}. The $O7$-planes at $\pm k$ have charge $+$ (which corresponds to $(-4)$ units of $D7$-brane charge) and the $O7$-planes at $\pm\tfrac{i}{k}$ have opposite signs. Again we'd expect the the three $I_2$ fibers that merged to form each $I_0^*$ fiber in the $\varepsilon=0$ limit to have the same charge as the $O$-plane represented by the $I_0^*$ fiber. This means that the $I_2$ fibers located at $c$ and $\bar{c}$ both have charge $+$, which is very different than in the previous case, the $(+,+,-,-)$ case, where the $I_2$ fibers at $c=i k$ and $\bar{c}=-ik$ had opposite charge. In the current case $c+\bar{c}$ is no longer trivial and it is tempting to use this to explain the difference. We might also want to look at $K_c\pm K_{-c}$, where $K_c$ is the singular fiber at $c$, in an attempt to construct the invariant lattice under $u\to\bar{u}$, similar to Section~\ref{sec:lattice}. However, the charge choices appear more directly determined by the number of real sections over $\bR$. For the $(+,+,+,+)$ case there were $12$ pairs of real sections given by
\begin{equation}
	(x,y,z) = (\pm k', \pm\sqrt{b'(\pm k',1)}uv,1 ),\quad\text{or} \quad (\pm \frac{1}{k'}, \pm\sqrt{b'(\pm \frac{1}{k'},1)}uv,1 ),
\end{equation}
for the $(+,+,-,-)$ there are no real sections, and for this case there $2$ pairs of real sections
\begin{equation}
	(x,y,z) = (\sigma_{i,1}, \pm\sqrt{b'(\sigma_{i,1},1)}uv,1 ),
\end{equation}
for $\sigma_{i,1}=\pm k'$.

To make this clearer, let us look at one more case where there is always exactly one set of $4$ real sections and we break some of the symmetry between the base and fiber to make their physical roles more clear.

\subsubsection{Another $(+,+,+,-)$ family}
\label{sec:+++-2}

In the previous three cases we had symmetry in the form of the roots of the base and fiber. This is equivalent to what was done in \cite{MR3267662} and \cite{MR3316647}, where the only cases considered had the complex structure and complexified K\"ahler structure define the same elliptic curve. Using the extra data of F-theory, there is no longer a need to maintain that symmetry. To better understand the physical roles of the base and fiber let us now break that symmetry. 

Before going into the case at hand let us note some physical implication of choosing $\omega_{4,2}^2$ to be $-1$ versus $1$.
\begin{lemma}
	Let $\omega_{4,2}^2=-1$. Then smooth elliptic fibers for the K3 surface $\pi_X\colon X \to \mathbb{P}(u, v)$ in Equation~\eqref{eqn:K3normab} are species $1$ real elliptic curves over the reals.
\end{lemma}
\begin{proof}
	The roots of $y^2=0$ for a general fiber over the reals are $\tfrac{x}{z}=\pm k, \pm\tfrac{i}{k}$ where
	\begin{equation}
		k=\frac{1}{\sqrt{2}}\sqrt{-\frac{b(u,v)}{a(u,v)}+\sqrt{\left(\frac{b(u,v)}{a(u,v)}\right)^2+4}}\in\bR_{>0},
	\end{equation}
	since$\tfrac{b(u,v)}{a(u,v)}\in\bR$ over the reals.
	The general fiber over the reals has exactly two real roots, so must be a species $1$ real elliptic curve.
\end{proof}
All species $1$ real elliptic curves are equivalent to elliptic curves with complex structure $\tau$ such that
\begin{equation}
	\Re{(\tau)}=\frac{1}{2}.
\end{equation}
Furthermore, in the string limit (if it exists) the complex structure of the $F$-theory fiber corresponds to 
\begin{equation}
\label{eqn:comKahph}
	B+ig_s(u,v),
\end{equation}
where $g_s(u,v)$ is the string coupling. Since the real part of $\tau$ is constant over the reals we immediately get
\begin{corollary}
	When $\omega_{4,2}^2=-1$ the string limit has nontrivial $B$-field, $B=\frac{1}{2}$.
\end{corollary}

To differentiate this from the case $\omega_{4,2}^2=1$ we note the following.
\begin{lemma}
	Let $\omega_{4,2}^2=1$. Then smooth elliptic fibers for the K3 surface $\pi_X\colon X \to \mathbb{P}(u, v)$ in Equation~\eqref{eqn:K3normab} are species $0$ or $2$ real elliptic curves over the reals.
\end{lemma}
\begin{proof}
	The general fiber over the reals 
	\begin{equation}
		y^2=a(u,v)\left(x^4+\frac{b(u,v)}{a(u,v)}x^2z^2+z^4\right)=0
	\end{equation}
	has $4$ real roots if $\tfrac{b(u,v)}{a(u,v)}<-2$ and $0$ otherwise, so must be a species $0$ or $2$ real elliptic curve.
\end{proof}
All species $0$ and $2$ real elliptic curves are equivalent to elliptic curves with complex structure $\tau$ purely imaginary. Since the real part of $\tau$ is constant over the reals we immediately get the following from Equation~(\ref{eqn:comKahph}):
\begin{corollary}
	When $\omega_{4,2}^2=1$ the string limit has trivial $B$-field, $B=0$.
\end{corollary}
For $D=8$ type IIB orientifolds, $B=0,\tfrac{1}{2}$ are the only possibilities. These results match the three cases we looked at above since $(S^{1,1}\times S^{1,1})_(+,+,+,+)$ and $(S^{1,1}\times S^{1,1})_{(+,+,-,-)}$ both have trivial $B$-fields whereas $(S^{1,1}\times S^{1,1})_{(+,+,+,-)}$ has a non-trivial $B$-field.

Now let us consider a specific case to further understand the role of the base versus the fiber. Set $\omega_2=\omega_{4,1}^2=1$, $\omega_{4,2}^2=-1$,
\begin{equation}
\label{eqn:Bconstraints}
\begin{split}
	\kappa, \ \mu,  \ne  \pm 2 , \quad \text{and}\quad \lambda > \kappa\mu.
	\end{split}
\end{equation}
The roots of $a(u,1)=0$ and $\{\pm x,\pm \tfrac{1}{x}\}$ where $x\in \bR^*$ if $\kappa<-2$, $x=e^{i\theta}$ with $\theta\ne \tfrac{n\pi}{2}$ if $-2<\kappa<2$, and $x\in i\bR^*$ if $\kappa>2$. Furthermore the roots $b(u,1)+2ia(u,1)=0 $ are given by $\{\pm c,\pm\tfrac{1}{c}\}$, and the roots of $b(u,1)- 2ia(u,1)=0$ are given $\{\pm\bar{c},\pm\tfrac{1}{\bar{c}}\}$, where
\begin{equation}
	c=\frac{1}{\sqrt{2}}\sqrt{-\chi+\sqrt{\chi^2-4}},\qquad \chi=\frac{\lambda+2\kappa i}{\mu+2i}.
\end{equation}
Note that since $\kappa,\mu\ne \pm 2$, $c\in\bC\backslash(\bR\cup i\bR)$ for all $\lambda\ne  \kappa\mu$. Therefore $|c|\ne 1$ for all $\lambda\ne \kappa\mu$. This shows that all $12$ $I_2$ fibers will be at distinct locations for all $\lambda\in \bR$ except when the fibration is isotrivial. We again define
\begin{equation}
	\chi'= \frac{\lambda + 2\mu}{\kappa+2}, \ \psi'= \frac{\lambda - 2\mu}{\kappa-2} \in \bR ,
\end{equation} 
and introduce $k', l', m'\in\bR_{>0}$ with
\begin{equation}
\label{eqn:Defklm'}
	\begin{split}
	k'&=\frac{1}{\sqrt{2}}\sqrt{-\mu+\sqrt{\mu^2+4}}	,\quad l'=\frac{1}{\sqrt{2}}\sqrt{-\chi'+\sqrt{\chi'^2+4}},\\
	&\qquad\text{and}\quad  m'=\frac{1}{\sqrt{2}}\sqrt{-\psi'+\sqrt{\psi'^2+4}}	.
	\end{split}
\end{equation}
We see that  
\beq
\begin{split}
 \mu = - \left( k'^2 - \frac{1}{k'^2} \right) , \quad \chi'   =  - \left(l'^2 - \frac{1}{l'^2}\right) , \quad \text{and} \quad \psi' =  - \left(m'^2 - \frac{1}{m'^2}\right) .
\end{split}
\eeq 
It follows that the roots of $a'(x,1)=0$ are given by $\{ \pm k', \pm \tfrac{i}{k'} \}$, the roots of $b'(u,1)+2 a'(u,1)=0$ are given by $\{ \pm l', \pm \tfrac{i}{l'} \}$, and the roots $b'(u,1)-2 a'(u,1)=0$ are given by $\{ \pm m', \pm \tfrac{i}{m'} \}$. We now have the following:
\begin{proposition}
\label{prop:+++-2}
		Let $\kappa,\mu,\lambda\in\bR$ such that $\kappa,\mu\ne \pm 2$, $\lambda>\kappa\mu$, $\omega_2=\omega_{4,1}^2=1$ and $=\omega_{4,2}^2=-1$. Define the real valued functions $\kappa_\varepsilon,\mu_\varepsilon,\lambda_\varepsilon\in\bR$ as in Equation~(\ref{eqn:family}) for all $\varepsilon\in[0,\infty)$. Then Equation~(\ref{eqn:K3normab}) defines a family of real elliptically fibered K3 surfaces $\pi_{X_\varepsilon} \colon X_\varepsilon \to \mathbb{P}(u, v)$ whose limit for $\varepsilon=0$ is $\mathrm{Kum}(E_1 \times E_2)$ with
		\beq
j (E_1) =  \frac{(12+\kappa^2)^3}{108 (4-\kappa^2)^2}, \quad
j (E_2) = -\frac{(\mu^2-12)^3}{108 (4+\mu^2)^2}.
\eeq
$X_\varepsilon$ is non-singular for all $\varepsilon\in [0,\infty)$ and has $12$ $I_2$ fibers when $\varepsilon\ne 0$. Moreover, there is a real sections in Equation~(\ref{eqn:K3normab}) yielding an isomorphism between every $X_\varepsilon$ and $J_{X_\varepsilon}$ over $\mathbb{R}$ for all $\varepsilon \in [0,\infty)$.
\end{proposition}
\begin{proof}
	The proof follows the same form of the proof of Propsition~\ref{prop:++++}, but the $\varepsilon\to 0$ limiting K3 is isomorphic to the Kummer surface $\Kum{(E_1\times E_2)}$ where the elliptic curves are now given by
	\beq
\begin{split}
 E_1\colon \  y_1^2 = u^4 + \kappa  u^2v^2 +v^4, \quad   E_2\colon \  y_2^2 = x^4 + \mu  x^2z^2 -z^4.
\end{split}
\eeq
Furthermore, since $\kappa_\varepsilon,\mu_\varepsilon\ne\pm 2$ and $\lambda_\varepsilon>\kappa_\varepsilon\mu_\varepsilon$ for all $\varepsilon\in[0,\infty)$ we see from the above discussion that $a'(x,1)$ always has a real root. Therefore $X_\varepsilon$ satisfies the conditions of Corollary~\ref{cor1} when $\varepsilon\in(0,\infty)$ showing there exists a real section.
\end{proof}

Let us now be more specific about the number of real sections that exist.
\begin{corollary}
	Let $\kappa, \mu , \lambda \in \mathbb{R}$ satisfying Equation~(\ref{eqn:Bconstraints}), $\omega_2=\omega_{4,1}^2=1$, and $\omega_{4,2} =-1$. Define the real-valued functions $\kappa_\varepsilon,  \mu_\varepsilon,  \lambda_\varepsilon  \in \mathbb{R}$, over $\varepsilon \in [0, \infty)$ as in Equation~(\ref{eqn:family}). Then the number of sections in Equations~(\ref{eqn:sections}) that are real over the reals is
	\begin{enumerate}
		\item $6$ pairs if $\kappa<-2$,
		\item $4$ pairs if $-2<\kappa<2$,
		\item $2$ pairs if $\kappa>2$.
	\end{enumerate}
\end{corollary}
\begin{proof}
First recall that Equation~(\ref{eqn:family}) induces real $\varepsilon$-linear functions $\chi'_\varepsilon,\psi'_\varepsilon$ given by Equation~(\ref{eqn:chieps}), repeated here for convenience, 
\begin{equation}
	\chi'_\varepsilon=\mu+\frac{\varepsilon}{\kappa+2}, \quad \psi'_\varepsilon=\mu+\frac{\varepsilon}{\kappa-2},
\end{equation}
which induce equations for the roots of $a'(x,1)$ and $b'(x,1)\pm 2a'(x,1)$, $k'_\varepsilon,l'_\varepsilon,m'_\varepsilon\in\bR_{>0}$ by Equation~(\ref{eqn:Defklm'}). In all cases $a'(x,1),b'(x,1)\pm 2a'(x,1)$ are quartics with two real roots. $a'(x,1)$ always has a positive leading coefficient, but the sign of the leading coefficient $b'(x,1)\pm 2a'(x,1)$ depends on which case we're in.

First let $\kappa<-2$. Then $b'(x,1)\pm 2a'(x,1)$ both have negative leading coefficients and
\begin{equation}
		\chi'_\varepsilon<\psi'_\varepsilon<\mu_\varepsilon=\mu,
\end{equation}
for $\varepsilon>0$. This shows that
\begin{equation}
	0<k_\varepsilon'<m'_\varepsilon<l'_\varepsilon
\end{equation}
for all $\varepsilon>0$. This shows that
\begin{equation}
		\begin{split}
			b'(\sigma_{i,1})>0& \quad\text{for}\quad \sigma_{i,1}\in \{\pm k'_\varepsilon \},\\
			b'(\sigma_{i,2})>0& \quad\text{for}\quad \sigma_{i,2}\in \{\pm m'_\varepsilon\},\\
			b'(\sigma_{i,3})<0& \quad\text{for}\quad \sigma_{i,3}\in \{\pm l'_\varepsilon\}.
		\end{split}
	\end{equation}
Therefore the $6$ pairs of sections in Equation~(\ref{eqn:sections}) where $\sigma_{i,1}\in \{\pm k'_\varepsilon \}$, $\sigma_{i,2}\in \{\pm m'_\varepsilon\}$, or $\sigma_{i,3}\in \{\pm l'_\varepsilon\}$ are real by Lemma~\ref{lem:real_section}. None of the other $6$ pairs of sections in Equation~(\ref{eqn:sections}) are over the reals.

Now let $-2<\kappa<2$. Then $b'(x,1)- 2a'(x,1)$ has a negative leading coefficient and $b'(x,1)+ 2a'(x,1)$ has a positive leading coefficient. Furthermore,
\begin{equation}
	\psi'_\varepsilon<\mu_\varepsilon<\chi'_\varepsilon, \quad\text{so}\quad 0< l'_\varepsilon < k' < m'_\varepsilon 
\end{equation}
for all $\varepsilon>0$. This shows that
\begin{equation}
		\begin{split}
			b'(\sigma_{i,1})>0& \quad\text{for}\quad \sigma_{i,1}\in \{\pm k'_\varepsilon \},\\
			b'(\sigma_{i,2})>0& \quad\text{for}\quad \sigma_{i,2}\in \{\pm m'_\varepsilon\},\\
			b'(\sigma_{i,3})>0& \quad\text{for}\quad \sigma_{i,3}\in \{\pm l'_\varepsilon\}.
		\end{split}
	\end{equation}
	Therefore the only pairs of real sections in Equation~(\ref{eqn:sections}) over the reals  are the four when $\sigma_{i,1}= \pm k'_\varepsilon$ or $\sigma_{i,2}= \pm m'_\varepsilon$ by Lemma~\ref{lem:real_section}.

Finally let $\kappa>2$. Then $b'(x,1)\pm 2a'(x,1)$ both have positive leading coefficients, and
\begin{equation}
	\mu_\varepsilon<\chi'_\varepsilon<\psi'_\varepsilon, \quad\text{so}\quad 0< m'_\varepsilon < l'_\varepsilon < k'_\varepsilon 
\end{equation}
for all $\varepsilon>0$. This shows that
\begin{equation}
		\begin{split}
			b'(\sigma_{i,1})>0& \quad\text{for}\quad \sigma_{i,1}\in \{\pm k'_\varepsilon \},\\
			b'(\sigma_{i,2})<0& \quad\text{for}\quad \sigma_{i,2}\in \{\pm m'_\varepsilon\},\\
			b'(\sigma_{i,3})>0& \quad\text{for}\quad \sigma_{i,3}\in \{\pm l'_\varepsilon\}.
		\end{split}
	\end{equation}
	Therefore the only pairs of real sections in Equation~(\ref{eqn:sections}) over the reals  are the two when $\sigma_{i,1}= \pm k'_\varepsilon$ by Lemma~\ref{lem:real_section}.
\end{proof}

Note that for only two of the cases above do real sections over the reals exist for $\varepsilon<0$.
\begin{corollary}
Let $\kappa, \mu , \lambda \in \mathbb{R}$ satisfying Equation~(\ref{eqn:Bconstraints}), $\omega_2=\omega_{4,1}^2=1$, and $\omega_{4,2} =-1$. Define the real-valued functions $\kappa_\varepsilon,  \mu_\varepsilon,  \lambda_\varepsilon  \in \mathbb{R}$, over $\varepsilon \in (-\infty, 0)$ as in Equation~(\ref{eqn:family}). Then the number of sections in Equations~(\ref{eqn:sections}) that are real over the reals is 
	\begin{enumerate}
		\item $2$ pairs if $-2<\kappa<2$,
		\item $4$ pairs if $\kappa>2$.
	\end{enumerate}
\end{corollary}
\begin{proof}
When $-2<\kappa<2$
\begin{equation}
	1< m'_\varepsilon < k'_\varepsilon < l'_\varepsilon 
\end{equation}
for $\varepsilon<0$. Therefore the only real section over the reals in Equation~(\ref{eqn:sections}) is when $\sigma_{i,3}=\pm l'_\varepsilon$.

When $\kappa>2$
\begin{equation}
	1< k'_\varepsilon < l'_\varepsilon < m'_\varepsilon
\end{equation}
for $\varepsilon<0$. Therefore the only real sections over the reals in Equation~(\ref{eqn:sections}) are when $\sigma_{i,2}=\pm m'_\varepsilon$ or $\sigma_{i,3}=\pm l'_\varepsilon$.
\end{proof}

This shows us that we can extend the family $X_\varepsilon$ of real elliptically fibered K3 surfaces with a real section from Proposition~\ref{prop:+++-2} to include $\varepsilon\in\bR$ when $\kappa>-2$. However, an interesting wall crossing phenomenon occurs when crossing $\varepsilon=0$. For example, when $\kappa>2$ the real sections are $(x,y,z)=(\pm m'_\varepsilon,\pm\sqrt{b'(m'_\varepsilon, 1)} \frac{u^2+ v^2}{\sqrt{2}},1)$, $(\pm l'_\varepsilon,\pm\sqrt{-b'(l'_\varepsilon, 1)} \frac{u^2- v^2}{\sqrt{2}},1) $, whereas for $\varepsilon>0$ the real section is 
\[
(x,y,z)=(k',\sqrt{ b'(k', 1)} uv,1). 
\]
For all $\kappa\in\bR$, if we take the total number of real sections for $\varepsilon>0$ plus the number of real sections for $\varepsilon<0$ we get $6$ pairs. Furthermore, in the isotrivial limit there are always $2$ real sections. This is different than the $(+,+,+,+)$ case where in the isotrivial limit there were $4$ real sections, and this is what determines the sign difference of the $O$-planes. Let us now compare these two cases a little closer.

To draw a direct comparison to the case considered in Section~\ref{sec:++++} let us set $\kappa,\mu<-2$ and $\lambda>\kappa\mu$. Then the only distinction between the two cases is the value of $\omega_{4,2}^2$. When $\omega_{4,2}^2=1$ the $B$-field is trivial and $X_\varepsilon$ has $12$ pairs of real fibers. When $\omega_{4,2}^2=-1$ the $B$-field is nontrivial and $X_\varepsilon$ has $6$ pairs of real fibers. Note that in the isotrival limit the locations of the $I_0^*$ fibers are exactly the same in these two cases, however; since the case with $\omega_{4,2}^2=-1$ has half as many real sections over the reals, its $O$-plane charge should be half as much and we see the $O$-planes should have charge $(+,+,+,-)$. In type IIB orientifold theories the charge spectrum should not depend on the complex structure of the base, but rather the complexified K\"ahler structure. What we see here is that the locations of the $I_0^*$ fibers are determined by the complex structure of the base, but there are two twisting factors that determine their sign, the Brauer group and the $B$-field. In this case we only have a twisting by the $B$-field and in the $(+,+,-,-)$ case we only had twisting by the Brauer group. Notice these two types of twisting interact through $\kappa$. The $B$-field only depends on the value of $\omega_{4,2}^2$ and the locations of the $I_0^*$ fibers depend only on $\kappa$. However, changing $\kappa$ changes which sections are real over different regions of $\varepsilon$.

\section{Conclusion}
\subsection{Physics interpretation and dualities}
We have seen in Section~\ref{sec:real} that our family
of complex elliptically fibered K3's from Section~\ref{sec:rank17fam}
admits three different families of
orientifold structures compatible with
the elliptic fibration and a species-$1$ (or ``usual''
in the terminology of Nikulin-Saito) antiholomorphic involution
on the base $\bP^1$.  These are summarized in Table \ref{tab:3realfam}.
There are no corresponding families of orientifolds for a species-$0$
(or ``spin'') antiholomorphic involution on $\bP^1$
because of Lemma \ref{lem:species0}.  All three families
from Table \ref{tab:3realfam}
admit degeneration to an isotrivial elliptic fibration on
the Kummer surface of a product of two elliptic curves,
where the O-plane charges of the $I_0^*$ fibers are one of
the three possibilities 
$(+,+,+,+)$, $(+,+,-,-)$, and $(+,+,+,-)$. 
In this way, we are led back to the study of dualities
and D- and O-plane charges on orientifolds on products of
elliptic curves, as in Section~ \ref{sec:twoellcurves1} and~\ref{sec:twoellcurves}.

Now we review with a simple but nontrivial example,
namely the case of \eqref{eqn:K3normab} with $\omega_{4,1}^2=\omega_{4,2}^2=1$,
i.e., $y^2=\pm\left(ax^4+bx^2z^2+az^4\right)$, where $a$ and $b$ are
of the form $a = u^4 + \kappa u^2 v^2 + v^4$,
$b = \mu u^4 + \lambda u^2 v^2 + \mu v^4$, $\kappa,\,\lambda,\,\mu\in \bR$
and $[u:v]\in \bP^1$ (the base of our elliptic fibration).  We take the
Galois involution on the base $\bP^1$ to be simply
$[u:v]\mapsto[\overline u:\overline v]$ and on the fibers to be given
by $(x,y,z)\mapsto(\overline x,\overline y, \overline z)$.
With $\mu>0$, $|\kappa|<2$, $|\lambda|<2\mu $, and the minus sign
in front, i.e., $y^2=-\left(ax^4+bx^2z^2+az^4\right)$, we get an elliptic
fibration with fibers over $\bP^1_\bR$ that are species $0$ curves of genus $1$,
since $a$ and $b$ are positive everywhere on $\bP^1_\bR$, 
with Jacobians having the same $j$-invariant, corresponding to changing
the sign in front. This is the case where the twisting by the Brauer group does not vanish, so we expect a physical duality between the
elliptically fibered K3 in the first case, which has no real points,
and one having real locus with species $2$ fibers and a $2$-torsion
twist.  (The twist is $2$-torsion since our genus-$1$ curve over
$\bR(u)$ has rational points over the degree-$2$ extension field $\bC(u)$.)
\subsection{Open problems and future work}

All string theories on orientifolds of $\bR^{8,0}\times E$ where $E$ is an elliptic curve were considered in \cite{MR3267662,MR3316647} except the heterotic theories. All the string theories in \cite{MR3267662,MR3316647} were related to each other by $T$-dualities and the remaining heterotic theories should be related by $S$-dualities. We plan to use heterotic/F-theory duality for the families of real K3 orientifolds considered here to describe the dual heterotic theories. This requires exploring how the real structure impacts the charge spectrum in the heterotic theory and will close the loop of all dual $D=8$ string theories.

Another immediate question to explore after clarrifying the role of the real structure in the charge spectrum of both string theory and F-theory orientifolds is what role does the real structure play in the charge spectrum of M-theory. We also plan to look at how M/F-theory duality can be used in a similar fashion to \cite{Acharya2022} with shifted dimension. This will allow us to describe the type IIA orientifold theories considered in \cite{MR3267662,MR3316647} in terms of M-theory. Understanding the role the real structure plays in the charge spectrum of these different settings will give a greater understanding of how charges transform under more general dualities than $T$-duality.

\bibliographystyle{amsplain}
\bibliography{K3orientifolds}
\end{document}